\documentclass[a4paper,11pt]{article} 
  
\usepackage[english]{babel}
\usepackage[a4paper]{geometry}
\usepackage{amsmath}
\usepackage{amsthm}
\usepackage{amssymb}

\def\bR{\mathbb{R}}
\def\bN{\mathbb{N}}
\def\NN{\mathbb{N}}
\def\bZ{\mathbb{Z}}
\def\cU{\mathcal{U}}
\def\cW{\mathcal{W}}
\def\cV{\mathcal{V}}
\def\cF{\mathcal{F}}
\def\cG{\mathcal{G}}
\def\cL{\mathcal{L}}
\def\cN{\mathcal{N}}
\def\cE{\mathcal{E}}
\def\cK{\mathcal{K}}
\def\cH{\mathcal{H}}

\def\ph{\varphi}

\def\wtph{\widetilde{\varphi}}
\def\wt{\widetilde}

\def\indic{\hbox{\raise-2pt \hbox{\indbf 1}}}

\let\==\equiv

\let\0=\noindent

\def\*{{\hfill\break\null\hfill\break}}

\def\tende#1{\,\vtop{\ialign{##\crcr\rightarrowfill\crcr
             \noalign{\kern-1pt\nointerlineskip}
             \hskip3.pt${\scriptstyle #1}$\hskip3.pt\crcr}}\,}
\def\otto{\,{\kern-1.truept\leftarrow\kern-5.truept\to\kern-1.truept}\,}

\def\tr{{\rm tr}}

\newtheorem{theorem}{Theorem}[section]  
\newtheorem{prop}[theorem]{Proposition}
\newtheorem{lemma}[theorem]{Lemma}

\numberwithin{equation}{section}

\newcommand{\lb}[0]{\left}
\newcommand{\rb}[0]{\right }

\newcommand{\lt}[1]{\ensuremath{ \left \| #1 \right  \|_{2 }  }}
\newcommand{\hnn}[2]{\ensuremath{ \left \| #1 \right  \|_{H^{#2} }  }}
\newcommand{\product}[2]{\ensuremath{\left\langle #1, #2 \right \rangle}}
\newcommand{\re}[0]{\ensuremath{\operatorname{Re}}}

\newcommand{\delt}[0]{\ensuremath{\partial_t}}
\newcommand{\np}[0]{\ensuremath{\mathcal{N}}}

\newcommand{\ay}[0]{\ensuremath{a_y}}
\newcommand{\asy}[0]{\ensuremath{a_y^*}}
\newcommand{\az}[0]{\ensuremath{a_z}}
\newcommand{\asz}[0]{\ensuremath{a_z^*}}
\newcommand{\bx}[0]{\ensuremath{b_x}}
\newcommand{\bsx}[0]{\ensuremath{b_x^*}}
\newcommand{\by}[0]{\ensuremath{b_y}}
\newcommand{\bsy}[0]{\ensuremath{b_y^*}}

\newcommand{\bsz}[0]{\ensuremath{b_z^*}}
\newcommand{\HN}[0]{\ensuremath{\mathcal{H}_N}}
\newcommand{\pt}[0]{\ensuremath{\varphi_t}}
\newcommand{\et}[0]{\ensuremath{\eta_t}}
\newcommand{\Gnt}[0]{\ensuremath{\mathcal{G}_{N,t}}}

\begin{document}

\title{Gross-Pitaevskii Dynamics for Bose-Einstein Condensates}

\author{Christian Brennecke and Benjamin Schlein \\
\\
Institute of Mathematics, University of Zurich\\
Winterthurerstrasse 190, 8057 Zurich, Switzerland}

\maketitle

\begin{abstract}
We study the time-evolution of initially trapped Bose-Einstein condensates in the Gross-Pitaevskii regime. We show that condensation is preserved by the many-body evolution and that the dynamics of the condensate wave function can be described by the time-dependent Gross-Pitaevskii equation. With respect to previous works, we provide optimal bounds on the rate of condensation (i.e. on the number of excitations of the Bose-Einstein condensate). To reach this goal, we combine the method of \cite{LNS}, where fluctuations around the Hartree dynamics for $N$-particle initial data in the mean-field regime have been analyzed, with ideas from \cite{BDS}, where the evolution of Fock-space initial data in the Gross-Pitaevskii regime has been considered.
\end{abstract}

\section{Introduction and Main Results}
\label{sec:intro}

Trapped gases of $N$ bosons in the Gross-Pitaevskii regime can be described by the Hamilton operator 
\begin{equation}\label{eq:Htrap} H^\text{trap}_N = \sum_{j=1}^N \left[ -\Delta_{x_j} + V_\text{ext} (x_j) \right] + \sum_{i<j}^N N^2 V(N(x_i -x_j)) 
\end{equation}
acting on the Hilbert space $L^2_s (\bR^{3N})$, the subspace of $L^2 (\bR^{3N})$ consisting of functions that are symmetric with respect to permutations of the $N$ particles. Here, $V_\text{ext}$ is a confining external potential. As for the interaction potential $V$, we assume it to be non-negative, spherically symmetric and compactly supported (but our results could be easily extended to potentials decaying sufficiently fast at infinity). 

Characteristically for the Gross-Pitaevskii regime, the interaction $N^2 V(N.)$ appearing in (\ref{eq:ham0}) scales with $N$ so that its scattering length is of the order $N^{-1}$. 
The scattering length $a_0$ of the unscaled potential $V$ is defined by the condition that the solution of the zero-energy scattering equation
\begin{equation}\label{eq:0en} \left[ -\Delta + \frac{1}{2} V (x)  \right] f (x) = 0 ,\end{equation}
with the boundary condition $f (x)  \to 1$ for $|x| \to \infty$, has the form  
\begin{equation}\label{eq:fdef} f(x) = 1 - \frac{a_0}{|x|} \end{equation}
outside the support of $V$. Equivalently, $a_0$ is determined by  
\begin{equation}\label{eq:scat-len}
8\pi a_0 = \int V(x) f(x) dx 
\end{equation}
By scaling, (\ref{eq:0en}) also implies that  
\[  \left[ -\Delta + \frac{N^2}{2} V (Nx) \right] f (Nx) = 0 \]
with $f(Nx) \to 1$ for $|x| \to \infty$. In particular, this means that the rescaled potential $N^2 V(N.)$ in (\ref{eq:ham0}) has scattering length $a_0/N$. 

It has been shown in \cite{LSY} (and more recently in \cite{NRS}) that the ground state energy $E_N$ of the Hamilton operator (\ref{eq:Htrap}) is such that 
\begin{equation}\label{eq:GPen} \lim_{N \to \infty} \frac{E_N}{N} = \min_{\substack{\ph \in L^2 (\bR^3): \\ \| \ph \|_2 = 1}} 
\cE_\text{GP}^\text{trap} (\ph) \end{equation}
with the Gross-Pitaevskii energy functional 
\begin{equation}\label{eq:GPen-functr} \cE^\text{trap}_\text{GP} (\ph) = \int \left[ |\nabla \ph (x)|^2 + V_\text{ext} (x)|\ph (x)|^2 + 4\pi a_0 |\ph (x)|^4 \right]  dx \end{equation}

Furthermore, Bose-Einstein condensation in the ground state of (\ref{eq:Htrap}) has been established in \cite{LS}. More precisely, if $\gamma^{(1)}_N = \tr_{2,\dots, N} |\psi_N \rangle \langle \psi_N|$ denotes the one-particle reduced density associated with the ground state of (\ref{eq:Htrap}), it has been shown in \cite{LS} that  
\begin{equation}\label{eq:cond0} \gamma_N^{(1)} \to |\phi_\text{GP} \rangle \langle \phi_\text{GP}| \end{equation}
where $\phi_\text{GP} \in L^2 (\bR^3)$ is the unique non-negative minimizer of (\ref{eq:GPen-functr}), among all $\ph \in L^2 (\bR^3)$ with $\| \ph \|_2 = 1$. The interpretation of (\ref{eq:cond0}) is straightforward: in the ground state of (\ref{eq:Htrap}), all particles, up to a fraction vanishing in the limit of large $N$, are in the same one-particle state $\phi_\text{GP}$. 

In typical experiments, one observes the time-evolution of trapped Bose gases prepared in (or close to) their ground state, resulting from a change of the external fields. As an example, consider the situation in which the trapping potential is switched off at time $t=0$. In this case, the dynamics is described, at the microscopic level, by the many-body Schr\"odinger equation 
\begin{equation}\label{eq:schr0} i\partial_t \psi_{N,t} = H_N \psi_{N,t}  \end{equation}
with the translation invariant Hamilton operator 
\begin{equation}\label{eq:ham0}
H_N = \sum_{j=1}^N -\Delta_{x_j} + \sum_{i<j}^N N^2 V(N(x_i -x_j)) \end{equation}
and with the ground state of (\ref{eq:Htrap}) as initial data. The next theorem shows how the solution of (\ref{eq:schr0}) can be described in terms of the time-dependent Gross-Pitaevskii equation. 
\begin{theorem}\label{thm:main}
Let $V_\text{ext} : \bR^3 \to \bR$ be locally bounded with $V_\text{ext} (x) \to \infty$ as $|x| \to \infty$. Let $V \in L^3 (\bR^3)$ be non-negative, compactly supported and spherically symmetric. Let $\psi_N$ be a sequence in $L^2_s (\bR^{3N})$, with one-particle reduced density $\gamma^{(1)}_N = \tr_{2,\dots , N} |\psi_N \rangle \langle \psi_N|$. We assume that, as $N \to \infty$,    
\begin{equation}\label{eq:assN}
\begin{split} 
a_N &= 1 - \langle \phi_\text{GP} , \gamma^{(1)}_N \phi_\text{GP} \rangle  \to 0  \qquad \text{and } \\
b_N &= \Big| N^{-1} \langle \psi_N, H^\text{trap}_N \psi_N \rangle - \cE^\text{trap}_\text{GP} (\phi_\text{GP}) \Big| \to 0  \end{split} 
\end{equation}
where $\phi_\text{GP} \in H^4 (\bR^3)$ is the unique non-negative minimizer of the Gross-Pitaevskii energy functional (\ref{eq:GPen-functr}). Let $\psi_{N,t} = e^{-i H_N t} \psi_N$ be the solution of (\ref{eq:schr0}) with initial data $\psi_N$ and let $\gamma_{N,t}^{(1)}$ be the one-particle reduced density associated with $\psi_{N,t}$. Then there are constants $C,c > 0$ such that 
\begin{equation}\label{eq:main} 
1 - \langle \ph_t , \gamma^{(1)}_{N,t} \ph_t \rangle \leq C \left[ a_N + b_N + N^{-1} \right] \exp \, (c \exp \, (c|t|)) 
\end{equation}
for all $t \in \bR$. Here $\ph_t$ is the solution of the time-dependent Gross-Pitaevskii equation  
\begin{equation}\label{eq:GPtd} i\partial_t \ph_t = -\Delta \ph_t + 8\pi a_0 |\ph_t|^2 \ph_t \end{equation}
with the initial data $\ph_{t=0} = \phi_\text{GP}$.
\end{theorem}

{\it Remarks:} 
\begin{itemize}
\item[1)] The condition $a_N = 1 - \langle \phi_\text{GP} , \gamma^{(1)}_N \phi_\text{GP} \rangle \to 0$ is equivalent with $\gamma^{(1)}_N \to |\phi_\text{GP} \rangle \langle \phi_\text{GP}|$. Similarly, the bound (\ref{eq:main}) implies that $\gamma_{N,t}^{(1)} \to |\ph_t \rangle \langle \ph_t|$. More precisely, using the fact that $|\ph_t \rangle \langle \ph_t|$ is a rank-one projection, it follows from (\ref{eq:main}) that 
\[ \begin{split} \tr \, \Big|\gamma^{(1)}_{N,t} - |\ph_t \rangle \langle \ph_t| \Big| &\leq 2 \Big\| \gamma^{(1)}_{N,t} - |\ph_t \rangle \langle \ph_t| \Big\|_\text{HS} \\ &\leq 2^{3/2} \big[ 1 - \langle \ph_t , \gamma^{(1)}_{N,t} \ph_t \rangle \big]^{1/2} \\ &\leq C \big[ a_N + b_N + N^{-1} \big]^{1/2} \exp (c \exp (c|t|)) \, . \end{split} \]
Hence, (\ref{eq:main}) is a statement about the stability of Bose-Einstein condensation with respect to the many-body Schr\"odinger equation (\ref{eq:schr0}). 
\item[2)] Existence, uniqueness and decay of the minimizer $\phi_\text{GP}$ of the Gross-Pitaevskii energy functional (\ref{eq:GPen-functr}) have been established in \cite{LSY}. In Theorem \ref{thm:main} we additionally assume that $\phi_\text{GP} \in H^4 (\bR^3)$. This condition follows from elliptic regularity and from the results of \cite{GY} (establishing decay of the derivatives of $\phi_\text{GP}$), under suitable assumptions on $V_\text{ext}$ (for example, if $V_\text{ext} \in C^2 (\bR^3)$ and its derivatives grow at most exponentially at infinity). 
\item[3)] As discussed above, it follows from \cite{LSY,LS} that the assumptions (\ref{eq:assN}) are satisfied if we take $\psi_N$ as the ground state of (\ref{eq:Htrap}). In this case, we expect both $a_N$ and $b_N$ to be of the order $N^{-1}$; indeed, $a_N, b_N \simeq N^{-1}$ has been recently shown in \cite{BBCS}, for systems of bosons trapped in a box with volume one (with periodic boundary conditions), interacting through a sufficiently small potential. In this case, (\ref{eq:main}) implies that 
\[ 1 - \langle \ph_t , \gamma^{(1)}_{N,t} \ph_t \rangle \leq C N^{-1} \exp (c \exp (c |t|)) \]
and therefore that, for every fixed time $t \in \bR$, Bose-Einstein condensation holds with the optimal rate  $N^{-1}$ (meaning that the number of excitations of the condensate remains bounded, uniformly in $N$).
\item[4)] To keep the notation as simple as possible, we consider the time evolution (\ref{eq:schr0}) generated by the translation invariant Hamiltonian (\ref{eq:ham0}). With the same techniques we use to prove Theorem \ref{thm:main}, we could also have included in (\ref{eq:ham0}) an external potential $W_\text{ext}$ (at least if the difference $W_\text{ext} - V_\text{ext}$ is bounded below). Under this assumption, the convergence (\ref{eq:main}) remains true, of course provided we introduce the external potential $W_\text{ext}$ also in the time-dependent Gross-Pitaevskii equation (\ref{eq:GPtd}). Physically, this would describe experiments where the system prepared at equilibrium (in the ground state) is perturbed by a change of the external potential, rather than by switching it off (we could also consider the situation where the external potential depends on time). 
\end{itemize}

Theorem \ref{thm:main} is meant to describe the time-evolution of data prepared in the ground state of the trapped Hamilton operator (\ref{eq:Htrap}). This is the reason why, in (\ref{eq:assN}), we assumed $\psi_N$ to exhibit Bose-Einstein condensation in the minimizer of the Gross-Pitaevskii energy functional (\ref{eq:GPen-functr}). {F}rom the mathematical point of view, one may ask more generally whether it is possible to show that the evolution of an initial data exhibiting Bose-Einstein condensate in an arbitrary one-particle wave function $\ph \in H^1 (\bR^3)$ (not necessarily minimizing the Gross-Pitaevskii functional (\ref{eq:GPen-functr})) continues to exhibit condensation in the solution of (\ref{eq:GPtd}) with initial data $\ph_{t=0} = \ph$, also for $t \not =0$. In the next theorem we show that the answer to this question is positive; the only difference with respect to (\ref{eq:main}) is the fact that, to get the same rate of convergence at time $t$, we need a stronger bound on the condensation of the initial data.
\begin{theorem}\label{thm:main2} 
Assume that $V \in L^3 (\bR^3)$ is non-negative, compactly supported and spherically symmetric. Let $\psi_N$ be a sequence in $L^2_s (\bR^{3N})$, with one-particle reduced density $\gamma^{(1)}_N = \tr_{2,\dots , N} |\psi_N \rangle \langle \psi_N|$. Assume that, for a $\ph \in H^4 (\bR^3)$,    
\begin{equation}\label{eq:wtassN}
\begin{split} 
\wt{a}_N &= \tr \, \big| \gamma^{(1)}_{N}  - | \ph \rangle \langle \ph| \big| \to 0 \qquad \text{and } \\
\wt{b}_N &= \Big| N^{-1} \langle \psi_N, H_N \psi_N \rangle - \cE_\text{GP} (\ph) \Big| \to 0  \end{split} 
\end{equation}
as $N \to \infty$. Here $\cE_\text{GP}$ is the translation invariant Gross-Pitaevskii functional 
\begin{equation}\label{eq:GP-func} \cE_\text{GP} (\ph) = \int \big[ |\nabla \ph|^2 + 4 \pi a_0 |\ph|^4 \big] dx  \end{equation}

Let $\psi_{N,t} = e^{-i H_N t} \psi_N$ be the solution of the Schr\"odinger equation (\ref{eq:schr0}) with initial data $\psi_N$ and let $\gamma_{N,t}^{(1)}$ denote the one-particle reduced density associated with $\psi_{N,t}$. Then  
\begin{equation}\label{eq:main2} 
1 - \langle \ph_t , \gamma^{(1)}_{N,t} \ph_t \rangle \leq C \left[ \wt{a}_N + \wt{b}_N + N^{-1} \right] \exp \, (c \exp \, (c|t|)) 
\end{equation}
where $\ph_t$ denotes the solution of the time-dependent Gross-Pitaevskii equation (\ref{eq:GPtd}), with initial data $\ph_0 = \ph$. 
\end{theorem}


A first proof of the convergence of the reduced density associated with the solution of the Schr\"odinger equation (\ref{eq:schr0}) towards the orthogonal projection onto the solution of the time-dependent Gross-Pitaevskii equation (\ref{eq:GPtd}) was obtained in \cite{ESY1,ESY2,ESY3,ESY4} 
(part of the proof was later simplified in \cite{CHPS}, using also ideas from \cite{KM}). In these works, convergence was established with no control on its rate. A new proof of the convergence towards the Gross-Pitaevskii dynamics was later given in \cite{P}; in this case, convergence was shown to hold with a rate $N^{-\eta}$, for an unspecified $\eta > 0$ (this approach was adapted to two-dimensional systems in \cite{JLP}, to systems with magnetic fields in \cite{O} and to pseudo-spinor condensates in \cite{MO}). More recently, convergence with a rate similar to (\ref{eq:main}), (\ref{eq:main2}) has been proven to hold in \cite{BDS}, for a class of Fock space initial data. The novelty of (\ref{eq:main}), (\ref{eq:main2}) is the fact that convergence is shown with an optimal rate determined by the properties of the $N$-particle initial data. 

More results are available about quantum dynamics in the mean-field regime. In this case, the evolution of the Bose gas is generated by an Hamilton operator of the form
\begin{equation}\label{eq:Hmf} H^{\text{mf}}_N = \sum_{j=1}^N -\Delta_{x_j} + \frac{1}{N} \sum_{i<j}^N V(x_i -x_j) \end{equation}
In the limit $N \to \infty$, the solution of the Schr\"odinger equation $\psi_{N,t} = e^{-iH_N^\text{mf} t} \psi_N$, for initial data $\psi_N$ exhibiting Bose-Einstein condensation in a one-particle wave function $\ph \in L^2 (\bR^3)$, can be approximated by products of the solution of the nonlinear Hartree equation 
\begin{equation}\label{eq:hartree} i\partial_t \ph_t = -\Delta \ph_t + (V*|\ph_t|^2) \ph_t 
\end{equation}
Convergence towards Hartree dynamics has been established in different settings and using different methods in several works, including  \cite{AGT,AB,AFP,AH,AN1,BGM,CH,ES,EY,FKP,FKS,Hepp,KP,
RS,Spohn}. In the mean-field regime, it is also possible to find a norm approximation of the many-body evolution by taking into account fluctuations around the Hartree dynamics (\ref{eq:hartree}); see, for example, \cite{BKS,XC,GMM1,GMM2,KSS,LNS, MPP}. 

It is also interesting to consider the many-body evolution in scaling limits interpolating between the mean-field regime described by the Hamilton operator (\ref{eq:Hmf}) and the Gross-Pitaevskii regime described by (\ref{eq:ham0}). A norm-approximation of the time-evolution in these intermediate regimes was recently obtained in \cite{BCS,GM,K,NN,NN2}. 

To prove Theorem \ref{thm:main} and Theorem \ref{thm:main2} we will combine the strategies used in \cite{BDS} and \cite{LNS}. Let us briefly recall the main ideas of these papers. In \cite{BDS}, the Bose gas was described on the Fock space $\cF = \bigoplus_{n \geq 0} L^2_s (\bR^{3n})$ by the Hamilton operator
\[ \cH_N = \int \nabla_x a_x^* \nabla_x a_x dx + \frac{1}{2} \int N^2 V (N (x-y)) a_x^* a_y^* a_y a_x \, dx dy \]
where $a_x^*, a_x$ are the usual operator valued distributions, creating and, respectively, annihilating a particle at the point $x \in \bR^3$. Notice that $\cH_N$ commutes with the number of particle operator $\cN = \int a_x^* a_x \, dx$, and that its restriction to the sector of $\cF$ with exactly $N$ particles  coincides with (\ref{eq:ham0}). 

On the Fock space $\cF$, a Bose-Einstein condensate can be described by a coherent state of the form $W (\sqrt{N} \ph) \Omega$, where $\Omega = \{ 1, 0, 0,\dots \}$ is the vacuum vector, $\ph \in L^2 (\bR^3)$ is a normalized one-particle orbital, and where, for every $f \in L^2 (\bR^3)$, \[ W(f) = \exp (a^* (f) - a(f)) \]
is a Weyl operator with wave function $f$. Here, we denoted by 
\[ a^* (f) = \int f(x) a_x^* \, dx  \qquad \text{and } a(f) = \int \bar{f} (x) a_x \]
the usual creation and annihilation operators on $\cF$, creating and annihilating a particle with wave function $f$. A simple computation shows that 
\[ W(\sqrt{N} \ph) \Omega = e^{-N/2} \left\{ 1 , N^{1/2} \ph , \dots , \frac{N^{n/2} \ph^{\otimes n}}{\sqrt{n!}} , \dots \right\} \]
In the coherent state $ W(\sqrt{N} \ph) \Omega$, the number of particles is Poisson distributed, with mean and variance equal to $N$. 

On the Fock space $\cF$, it is interesting to study the dynamics of approximately coherent initial states. In the Gross-Pitaevskii regime, however (in contrast with the mean field limit), we cannot expect the evolution of approximately coherent initial data to remain approximately coherent. On every sector of $\cF$ with a fixed number of particles, the coherent states $W(\sqrt{N} \ph) \Omega$ is factorized; it describes therefore uncorrelated particles. On the other hand, already from \cite{ESY2,EMS} and more recently also from \cite{CH2}, we know that, in the Gross-Pitaevskii regime, particles develop substantial correlations. To provide a better approximation of the many-body dynamics, Weyl operators were combined in \cite{BDS} with appropriate Bogoliubov transformations, leading to so-called squeezed coherent states. To be more precise, let $f$ denote the solution of the zero-energy scattering equation (\ref{eq:0en}) and $w= 1-f$ (keep in mind that, for $|x| \gg 1$, $w (x) = a_0 / |x|$). Using $w$, we define   
\begin{equation}\label{eq:kNt} k_{N,t} (x;y) = -N w (N (x-y)) \ph_t (x) \ph_t (y) 
\end{equation}
where $\ph_t$ is the solution of the time-dependent Gross-Pitaevskii equation (\ref{eq:GPtd}). In fact, in \cite{BDS} and also later in the present paper, it is more convenient to replace $\ph_t$ with the solution of the slightly modified, $N$-dependent, Gross-Pitaevskii equation (\ref{eq:GPmod}); to simplify the presentation, we neglect these technical details in this introduction. With (\ref{eq:kNt}), it is easy to check that $k_{N,t} \in L^2 (\bR^3 \times \bR^3)$, with $\| k_{N,t} \|_2$ bounded, uniformly in $N$ and in $t$. This implies that (\ref{eq:kNt}) is the integral kernel of a Hilbert-Schmidt operator. Hence, we can define, on $\cF$, the unitary Bogoliubov transformations
\begin{equation}\label{eq:bogo} T_t = \exp \, \left[ \frac{1}{2} \int dx dy \left( k_{N,t} (x;y)a_x^* a_y^* - \text{h.c.} \right) \right] \end{equation}
whose action on creation and annihilation operators is explicitly given by 
\begin{equation}\label{eq:act-bogo} T_t^* a^* (g) T_t = a^* (\cosh_{k_{N,t}} (g)) + a (\sinh_{k_{N,t}} (\bar{g})) 
\end{equation}
for all $g \in L^2 (\bR^3)$. Here $\cosh_{k_{N,t}}$ and $\sinh_{k_{N,t}}$ are the bounded operators ($\sinh_{k_{N,t}}$ is even Hilbert-Schmidt) defined by the convergent series
\begin{equation}\label{eq:chsh} \cosh_{k_{N,t}} = \sum_{n=0}^\infty \frac{(k_{N,t} \bar{k}_{N,t})^n}{(2n)!}, \quad \text{ and } \quad \sinh_{k_{N,t}} = \sum_{n=0}^\infty \frac{(k_{N,t} \bar{k}_{N,t})^n k_{N,t}}{(2n+1)!} \end{equation}

Using the Bogoliubov transformation $T_t$ to generate correlations at time $t$, it makes sense to study the time evolution of initial data close to the squeezed coherent state $W(\sqrt{N} \ph) T_0 \Omega$, and to approximate it with a Fock space vector of the same form. More precisely, for $\xi_N \in \cF$ close to the vacuum (in a sense to be made precise later), we may consider the time evolution 
\begin{equation}\label{eq:full-fock} 
e^{-i\cH_N t} W(\sqrt{N} \ph) T_0 \xi_N = W(\sqrt{N} \ph_t) T_t \xi_{N,t} \end{equation}
where we defined $\xi_{N,t} = \cU_N (t) \xi_{N}$ and the fluctuation dynamics 
\begin{equation}\label{eq:fluc-fock} \cU_N (t) = T_t^* W^* (\sqrt{N} \ph_t) e^{-i\cH_N t} W(\sqrt{N} \ph_0) T_0 
\end{equation}
In order to show that the one-particle reduced density $\gamma_{N,t}^{(1)}$ associated with the l.h.s. of (\ref{eq:full-fock}) is close to the orthogonal projection onto the solution of the Gross-Pitaevskii equation (\ref{eq:GPmod}), it is enough to prove that the expectation of the number of particles in $\xi_{N,t}$ is small, compared with the total number of particles $N$ (assuming this is true for $\xi_N$, at time $t=0$). In other words, the problem of proving convergence towards the Gross-Pitaevskii dynamics reduces to the problem of showing that the expectation of the number of particles remains approximately preserved by the fluctuation dynamics (\ref{eq:fluc-fock}). In \cite{BDS}, this strategy was used to show that the one-particle reduced density $\gamma^{(1)}_{N,t}$ associated with $\Psi_{N,t} = e^{-i \cH_N t} W(\sqrt{N} \ph) T_0 \xi_N$ is such that   
\[ \| \gamma_{N,t}^{(1)} - |\ph_t \rangle \langle \ph_t| \|_\text{HS} \leq C N^{-1/2} \exp (c \exp (c|t|)) \]
for any $\xi_N \in \cF$ with $\| \xi_N \| = 1$ and such that  
\[ \left\langle \xi_N, \left[ \cN + \cN^2/N + \cH_N \right] \xi_N \right\rangle \leq C \]
uniformly in $N$.  

While the method of \cite{BDS} works well to show convergence towards the Gross-Pitaevskii dynamics for the evolution of Fock space data of the form $W(\sqrt{N} \ph) T_0 \xi_N$, it is difficult to apply it to $N$-particle initial data in $L^2_s (\bR^{3N})$ (a special class of $N$-particle states for which this is indeed possible is discussed in Appendix C of \cite{BDS}). An alternative approach, tailored on $N$-particle initial data, was proposed in \cite{LNS} for bosons in the mean field limit. An important observation in \cite{LNS} (and already in \cite{LNSS}) is the fact that, for a fixed normalized $\ph \in L^2 (\bR^3)$, every $\psi_N \in L^2_s (\bR^{3N})$ can be uniquely represented as 
\begin{equation}\label{eq:UNdef} \psi_N = \sum_{n=0}^N \psi_N^{(n)} \otimes_s \ph^{\otimes (N-n)} 
\end{equation}
for a sequence $\{ \psi_N^{(n)} \}_{n=0}^N$ with $\psi_N^{(n)} \in L^2_{\perp \ph} (\bR^3)^{\otimes_s n}$, the symmetric tensor product of $n$ copies of the orthogonal complement of $\ph$ in $L^2 (\bR^3)$. 

This remark allows us to define a unitary map 
\begin{equation}\label{eq:Uph-def} U (\ph): L^2_s (\bR^{3N}) \to \cF_{\perp \ph}^{\leq N} \quad \text{ through } \quad  U(\ph) \psi_N = \{ \psi_N^{(0)}, \psi_N^{(1)}, \dots , \psi_N^{(N)} \}.\end{equation} 
Here $\cF_{\perp \ph}^{\leq N} = \bigoplus_{n=0}^N L^2_\perp (\bR^3)^{\otimes_s n}$ is the Fock space constructed on the orthogonal complement $L^2_{\perp \ph} (\bR^3)$ of $\ph$, truncated to have at most $N$ particles. The map $U(\ph)$ factors out the condensate described by the one-particle wave function $\ph$ and allows us to focus on its orthogonal excitations. Notice that a similar idea (but with no second quantization) was used in \cite{P, MPP} to identify excitations of the condensate. Using the unitary map (\ref{eq:Uph-def}), we can introduce, for the mean-field dynamics generated by (\ref{eq:Hmf}), a fluctuation dynamics 
\begin{equation}\label{eq:cWmf} \cW^\text{mf}_{N,t} = U (\ph_t) e^{-i H^\text{mf}_N t} U^* (\ph) : \cF_{\perp \ph}^{\leq N} \to \cF_{\perp \ph_t}^{\leq N} 
\end{equation}
where $\ph_t$ is the solution of the time-dependent Hartree equation (\ref{eq:hartree}). Similarly as above, to prove convergence towards Hartree dynamics, it is enough to control the growth of the expectation of the number of particles operator w.r.t. $\cW^\text{mf}_{N,t}$. This strategy was used in \cite{LNS} to find a norm-approximation for the many-body evolution in the mean-field regime. 

It is natural to ask whether the techniques developed in \cite{LNS} to study the time-evolution of bosonic systems in the mean-field regime can also be used to study the dynamics in the Gross-Pitaevskii limit. Similarly as above, where we argued that coherent states are not a good ansatz to describe the evolution of Fock space initial data, we cannot expect here that factorized $N$-particles states of the form $U_{\ph_t}^* \Omega = \ph_t^{\otimes N}$ provide a good approximation for the solution of the Schr\"odinger equation (\ref{eq:schr0}) in the Gross-Pitaevskii regime. Instead, similarly as in \cite{BDS}, we need to modify the ansatz to take into account correlations developed by the many-body evolution. As explained above, in \cite{BDS} correlations were modeled by means of Bogoliubov transformations of the form (\ref{eq:bogo}). Unfortunately, since they do not preserve the number of particles, these Bogoliubov transformations do not leave the space $\cF_{\perp \ph_t}^{\leq N}$, where excitations of the Bose-Einstein condensate are described, invariant. For this reason, to adapt the techniques of \cite{LNS} to the Gross-Pitaevskii regime that we are considering here, we are going to introduce, on $\cF_+^{\leq N}$ modified creation and annihilation operators, defined by  
\begin{equation}\label{eq:bb-intro} b^* (f) = a^* (f) \sqrt{\frac{N-\cN}{N}}, \qquad \text{and } \quad b (f) = \sqrt{\frac{N-\cN}{N}} a (f) 
\end{equation}
for all $f \in L^2_{\perp \ph_t} (\bR^3)$. As we will discuss in the next section, these new fields create and, respectively, annihilate excitations of the Bose-Einstein condensate leaving, at the same time, the total number of particles invariant. We will use the modified creation and annihilation operators to define generalized Bogoliubov transformation having the form 
\begin{equation}\label{eq:St} S_t = \exp  \left[ \frac{1}{2} \int dx dy \, (\eta_{t} (x;y) b_x^* b_y^* - \text{h.c.} ) \right] \end{equation}
for a kernel $\eta_{t} \in L^2 (\bR^3 \times \bR^3)$, orthogonal to $\ph_t$ in both its variables. Compared with the standard Bogoliubov transformations in (\ref{eq:bogo}), (\ref{eq:St}) has an important advantage: it maps $\cF_{\perp \ph_t}^{\leq N}$ back into itself. 

For this reason, with (\ref{eq:St}) we can define the modified fluctuation dynamics $\cW_{N,t} = S_t^* U (\ph_t) e^{-i H_N t} U^* (\ph_0) S_0 : \cF_+^{\leq N} \to \cF_+^{\leq N}$, which will play in our analysis a similar role as (\ref{eq:fluc-fock}) played in \cite{BDS}. To prove Theorem \ref{thm:main} and Theorem \ref{thm:main2} it will then be enough to show a bound for the growth of the expectation of the number of particles with respect to $\cW_{N,t}$. To achieve this goal, we will establish several properties of its generator. Technically, the main challenge we will have to face is the fact that, in contrast with (\ref{eq:act-bogo}), there is no explicit formula for the action of the generalized Bogoliubov transformation (\ref{eq:St}) on creation and annihilation operators. For this reason, we will have to expand expressions like $S_t^* b(g) S_t$ in absolutely convergent infinite series, and we will have to control the contribution of several different terms. The main tool to control these expansions is Lemma \ref{lm:indu} below.

\section{Fock Space}
\label{sec:fock}

In this section, we introduce some notations and we discuss 
some basic properties of operators on Fock space. Let 
\[ \cF = \bigoplus_{n \geq 0} L^2_s (\bR^{3n}) = \bigoplus_{n \geq 0} L^2 (\bR^3)^{\otimes_s n} \]
denote the bosonic Fock space over the one-particle space $L^2 (\bR^3)$. Here $L^2_s (\bR^{3n})$ is the subspace of $L^2 (\bR^{3n})$ consisting of all $\psi \in L^2 (\bR^{3n})$ with 
\[ \psi (x_{\pi 1}, x_{\pi 2}, \dots , x_{\pi n}) = \psi (x_1, \dots , x_n) \]
for all permutations $\pi \in S_n$. We use the notation $\Omega = \{ 1, 0, \dots \} \in \cF$ for the vacuum vector, describing a state with no particles. 

On $\cF$, it is convenient to introduce creation and annihilation operators. For $g \in L^2 (\bR^3)$, we define the creation operator $a^* (g)$ and the annihilation operator $a(g)$ by 
\[ \begin{split} 
(a^* (g) \Psi)^{(n)} (x_1, \dots , x_n) &= \frac{1}{\sqrt{n}} \sum_{j=1}^n g (x_j) \Psi^{(n-1)} (x_1, \dots , x_{j-1}, x_{j+1} , \dots , x_n) \\
(a (g) \Psi)^{(n)} (x_1, \dots , x_n) &= \sqrt{n+1} \int \bar{g} (x) \Psi^{(n+1)} (x,x_1, \dots , x_n) \end{split} \]
Notice that creation operators are linear in their argument, annihilation operators are antilinear. 
Creation and annihilation operators can be extended to closed unbounded operators on $\cF$; $a^* (g)$ is the adjoint of $a(g)$. They satisfy canonical commutation relations
\begin{equation}\label{eq:ccr} [a (g), a^* (h) ] = \langle g,h \rangle , \quad [ a(g), a(h)] = [a^* (g), a^* (h) ] = 0 \end{equation}
for all $g,h \in L^2 (\bR^3)$ (here $\langle g,h \rangle$ denote the usual inner product on $L^2 (\bR^3)$). It is also convenient to introduce operator valued distributions $a_x, a_x^*$, formally  creating and annihilating a particle at $x \in \bR$. They are such that  
\[ a(f) = \int \bar{f} (x) \, a_x \, dx , \quad a^* (f) = \int f(x) \, a_x^* \, dx  \]
and satisfy the commutation relations 
\[ [a_x , a_y^* ] = \delta (x-y) , \quad [a_x , a_y] = [a_x^* , a_y^*] = 0 \]

It is also useful to introduce on $\cF$ the number of particles operator, defined by $(\cN \Psi)^{(n)} = n \Psi^{(n)}$. In terms of operator valued distributions, $\cN$ can be written as
\[ \cN = \int a_x^* a_x \, dx \]
Creation and annihilation operators are bounded by the square root of the number of particles operator, i.e. we have 
\begin{equation}\label{eq:abd} \| a (f) \Psi \| \leq \| f \|_2 \| \cN^{1/2} \Psi \|, \quad \| a^* (f) \Psi \| \leq \| f \|_2 \| (\cN+1)^{1/2} \Psi \| 
\end{equation}
for every $f \in L^2 (\bR^3)$. 

For a one-particle operator $B:L^2 (\bR^3) \to L^2 (\bR^3)$ we define $d\Gamma (B): \cF \to \cF$ through $(d\Gamma (B) \Psi)^{(n)} = \sum_{j=1}^n B_j \psi^{(n)}$, for any $\Psi = \{ \psi^{(n)} \}_{n \in \bN} \in \cF$. Here $B_j = 1\otimes \dots \otimes B \otimes \dots \otimes 1$ acts as $B$ on the $j$-th particles and as the identity on all other particles. If $B$ has the integral kernel $B (x;y)$, we can write 
\[ d\Gamma (B) = \int B (x;y) a_x^* a_y \, dx dy \]
If $B$ is a bounded operator on the one-particle space $L^2 (\bR^3)$, $d\Gamma (B)$ can be bounded with respect to the number of particles operator, i.e. we have the operator inequality 
\begin{equation}\label{eq:dGamma1} \pm d\Gamma (B) \leq \| B \|_\text{op} \, \cN \end{equation}
and (since $d\Gamma (B)$ commutes with $\cN$) also
\[ \| d\Gamma (B) \Psi \| \leq \| B \|_\text{op} \| \cN \Psi \| \]

We will also need bounds for operators on the Fock space, quadratic in creation and annihilation operators, that do not necessarily preserve the number of particles. For $j \in L^2 (\bR^3 \times \bR^3)$, we introduce the notation
\begin{equation}\label{eq:AAdef} \begin{split} A_{\sharp_1, \sharp_2} (j) &= \int a^{\sharp_1} (j_{x}) a^{\sharp_2}_x \, dx = \int j^{\bar{\sharp}_1} (x;y) a_y^{\sharp_1} a_x^{\sharp_2} \, dx dy  \end{split} \end{equation}
where $j_x (y) := j (x;y)$, $\sharp_1 , \sharp_2 \in \{ \cdot,* \}$, $\bar{\sharp}_1 = \cdot$ if $\sharp_1 = *$ and $\bar{\sharp}_1 = *$ if $\sharp_1 = \cdot$, and where we use the notation $a^{\sharp} = a$ if $\sharp = \cdot$, $a^{\sharp} = a^*$ if $\sharp = *$ and, similarly, $j^\sharp = j$ if $\sharp = \cdot$ and $j^\sharp = \bar{j}$ if $\sharp = *$. If $\sharp_1 = \cdot$ and $\sharp_2 = *$ (i.e. if a creation operator lies on the right of an annihilation operator), in order to define $A_{\sharp_1, \sharp_2} (j)$ we also require that $x \to j (x;x)$ is integrable.  In the next lemma, which follows easily from (\ref{eq:abd}), we show how to bound these operators through the number of particles operator $\cN$. 
\begin{lemma} \label{lm:Abds} Let $j \in L^2 (\bR^3 \times \bR^3)$. Then for any $\Psi \in \cF$, 
\[ \begin{split} 
\| A^\flat_{\sharp_1, \sharp_2} (j) \Psi \| &\leq \sqrt{2} \| (\cN+1) \Psi \|   \left\{ \begin{array}{ll}  \| j \|_2 + \int |j(x;x)| dx \qquad &\text{if } \sharp_1 = \cdot, \sharp_2 = * \\  \| j \|_2  \qquad &\text{otherwise} \end{array} \right. \end{split} \]
\end{lemma} 

We will work on certain subspaces of $\cF$. For a fixed $\ph \in L^2 (\bR^3)$ ($\ph$ will later be the condensate wave function), we use the notation $L^2_{\perp \ph} (\bR^3)$ for the orthogonal complement of the one dimensional space spanned by $\ph$ in $L^2 (\bR^3)$. We denote by
\[ \cF_{\perp \ph} = \bigoplus_{n \geq 0} L^2_{\perp \ph} (\bR^3)^{\otimes_s n} \]
the Fock space constructed over $L^2_{\perp \ph} (\bR^3)$. A vector $\Psi = \{ \psi^{(0)}, \psi^{(1)}, \dots \} \in \cF$ lies in $\cF_{\perp \ph}$, if $\psi^{(n)}$ is orthogonal to $\ph$, in each of its coordinate, for all $n \geq 1$, i.e. if
\[ \int \bar{\ph} (x) \, \psi^{(n)} (x,y_1, \dots , y_{n-1}) dx = 0 \]
for all $n \geq 1$. We will also need Fock spaces with truncated number of particles. For $N \in \bN \backslash \{ 0 \}$, we define 
\[ \cF^{\leq N} = \bigoplus_{n=0}^N L^2 (\bR^3)^{\otimes_s n} \qquad \text{and } \quad  \cF_{\perp \ph}^{\leq N} = \bigoplus_{n=0}^N L^2_{\perp \ph} (\bR^3)^{\otimes_s n} \]
as the Fock spaces over $L^2 (\bR^3)$ and over $L^2_{\perp \ph} (\bR^3)$ consisting of states with at most $N$ particles. As already explained in the introduction (but see Section \ref{sec:fluc} for more details), on the space $\cF_{\perp \ph}^{\leq N}$ we will describe orthogonal fluctuations around a condensate with wave function $\ph \in L^2 (\bR^3)$.

On $\cF^{\leq N}$ and $\cF_{\perp \ph}^{\leq N}$, we introduce modified creation and annihilation operators. For $f \in L^2 (\bR^3)$, we define 
\begin{equation}\label{eq:bb-fie} b (f) = \sqrt{\frac{N- \cN}{N}} \, a (f), \qquad \text{and } \quad  b^* (f) = a^* (f) \, \sqrt{\frac{N-\cN}{N}} \end{equation}
We clearly have $b(f), b^* (f) : \cF^{\leq N} \to \cF^{\leq N}$. If moreover $f \perp \ph$ we also have $b(f), b^* (f): \cF^{\leq N}_{\perp \ph} \to \cF^{\leq N}_{\perp \ph}$. As we will discuss in the next section, the importance of these fields arises from the application of the map $U (\ph)$, defined in (\ref{eq:UNdef}), since 
\begin{equation}\label{eq:UaU} 
\begin{split} 
U (\ph) a^* (f) a(\ph) U^* (\ph) &= a^* (f) \sqrt{N-\cN}  = \sqrt{N} b^* (f)  \\
U (\ph) a^* (\ph) a(f) U^* (\ph) &=  \sqrt{N-\cN} \, a(f) = \sqrt{N} \, b (f)
\end{split} \end{equation}
If $\ph$ is the condensate wave function and $f \perp \ph$, the operator $b^* (f)$ excites a particle from the condensate to its orthogonal complement, while $b(f)$ annihilates an excitation back into the condensate. On states exhibiting Bose-Einstein condensation, we expect $a(\ph) , a^* (\ph) \simeq \sqrt{N}$ and thus that the action of modified $b^*$- and $b$-fields is close to the action of the original creation and annihilation operators. 

It is also convenient to introduce operator valued distributions 
\[ b_x = \sqrt{\frac{N-\cN}{N}} \, a_x, \qquad \text{and } \quad  b^*_x = a^*_x \, \sqrt{\frac{N-\cN}{N}} \]
so that 
\[ b(f) = \int \bar{f} (x) \, b_x \, dx , \qquad \text{and } \quad b^* (f) = \int f(x) b^*_x \, dx \]
We find the modified canonical commutation relations
\begin{equation}\label{eq:comm-b}
\begin{split}  [ b_x, b_y^* ] &= \left( 1 - \frac{\cN}{N} \right) \delta (x-y) - \frac{1}{N} a_y^* a_x \\ 
[ b_x, b_y ] &= [b_x^* , b_y^*] = 0 
\end{split} \end{equation}
Furthermore
\begin{equation}\label{eq:comm-b2}
\begin{split}
[b_x, a_y^* a_z] &=\delta (x-y) b_z, \qquad 
[b_x^*, a_y^* a_z] = -\delta (x-z) b_y^*
\end{split} \end{equation}
which leads to $[ b_x, \cN ] = b_x$ and $[ b_x^* , \cN ] = -b_x^*$. {F}rom (\ref{eq:abd}), we immediately obtain the following bounds for the $b$-fields. 
\begin{lemma}\label{lm:bbds}
Let $f \in L^2 (\bR^3)$. For any $\xi \in \cF^{\leq N}$, we have 
\[ \begin{split} 
\| b(f) \xi \| \leq \| f \|_2 \left\| \cN^{1/2} \left( \frac{N-\cN+1}{N} \right)^{1/2} \xi \right\|  \leq \| f \|_2 \| \cN^{1/2} \xi \| \\
\| b^* (f) \xi \| \leq \| f \|_2 \left\| (\cN +1)^{1/2} \left( \frac{N-\cN}{N} \right)^{1/2} \xi \right\|\leq \| f \|_2 \| (\cN + 1)^{1/2} \xi \|
\end{split} \]
Notice, moreover, that since $\cN \leq N$ on $\cF^{\leq N}$, $b(f), b^* (f) : \cF^{\leq N} \to \cF^{\leq N}$ are bounded operators with $\| b(f) \|, \| b^* (f) \| \leq (N+1)^{1/2} \| f \|_2$. 
\end{lemma}

We will also consider quadratic expressions in the $b$ fields. For an integral kernel $j \in L^2 (\bR^3 \times \bR^3)$, we define, similarly to (\ref{eq:AAdef}), 
\begin{equation}\label{eq:nota} \begin{split} 
B_{\sharp_1, \sharp_2} (j) &= \int b^{\sharp_1} (j_{x}) b^{\sharp_2}_x \, dx = \int j^{\bar{\sharp}_1} (x;y) b^{\sharp_1}_y b^{\sharp_2}_x \, dx dy 
\end{split}  \end{equation}
If $\sharp_1 = \cdot$ and $\sharp_2 = *$, we also require that $x \to j(x;x)$ is integrable. From Lemma \ref{lm:Abds}, we obtain the following bounds.
\begin{lemma}\label{lm:Bbds}
Let $j \in L^2 (\bR^3 \times \bR^3)$. Then
\[ \begin{split} \frac{\| B_{\sharp_1,\sharp_2} (j) \Psi \|}{\left\| (\cN+1) \left(\frac{N-\cN+2}{N} \right) \Psi \right\|} &
\leq \sqrt{2}   \left\{ \begin{array}{ll}  \| j \|_2 + \int |j(x;x)| dx \qquad &\text{if } \sharp_1 = \cdot, \sharp_2 = * \\  \| j \|_2  \qquad &\text{otherwise} \end{array} \right. 
\end{split} \]
for all $\Psi \in \cF^{\leq N}$. Since $\cN \leq N$ on $\cF^{\leq N}$, the operator $B_{\sharp_1 , \sharp_2} (j)$ is bounded, with \[ \begin{split} \| B_{\sharp_1, \sharp_2} (j) \| &\leq \sqrt{2} N \left\{ \begin{array}{ll} \| j \|_2 + \int |j(x;x)| dx \quad &\text{if } \sharp_1 = \cdot, \sharp_2 = * \\  \| j \|_2  \qquad &\text{otherwise} \end{array} \right. \end{split} \]
\end{lemma}

{\it Remark:} For $\ph \in L^2 (\bR^3)$, let $q_\ph = 1- |\ph \rangle \langle \ph|$ be the orthogonal projection onto $L^2_{\perp \ph} (\bR^3)$. If $j \in (q_{\ph^{\bar{\sharp}_1}} \otimes q_{\ph^{\bar{\sharp}_2}} ) (L^2 (\bR^3 \times \bR^3))$, we have $B_{\sharp_1, \sharp_2} (j) : \cF^{\leq N}_{\perp \ph} \to \cF^{\leq N}_{\perp \ph}$ (here we use the notation $\bar{\sharp} = *$ if $\sharp = \cdot$ and $\bar{\sharp} = \cdot$ if $\sharp = *$, and $\ph^{\sharp} = \ph$ if $\sharp = *$, $\ph^{\sharp} = \bar{\ph}$ if $\sharp = \cdot$).  

\medskip

We will consider products of several creation and annihilation operators, as well. In particular, two types of monomials in creation and annihilation operators will play an important role in our analysis. We define  
\begin{equation}\label{eq:Pi2} \Pi^{(2)}_{\sharp, \flat} (j_1, \dots , j_n) = \int   b^{\flat_0}_{x_1} a_{y_1}^{\sharp_1} a_{x_2}^{\flat_1} a_{y_2}^{\sharp_2} a_{x_3}^{\flat_2} \dots  a_{y_{n-1}}^{\sharp_{n-1}} a_{x_n}^{\flat_{n-1}} b^{\sharp_n}_{y_n} \, \prod_{\ell=1}^n j_\ell (x_\ell; y_\ell) \, dx_\ell dy_\ell \end{equation}
where $j_k \in L^2 (\bR^3 \times \bR^3)$ for $k=1,\dots ,n$ and where $\sharp = (\sharp_1, \dots , \sharp_n), \flat = (\flat_0, \dots , \flat_{n-1}) \in \{ \cdot, * \}^n$. In other words, for every index $i \in \{ 1, \dots , n\}$, we have either $\sharp_i = \cdot$ (meaning that $a^{\sharp_i} = a$ or $b^{\sharp_i} = b$) or $\sharp_i = *$ (meaning that $a^{\sharp_i} = a^*$ or $b^{\sharp_i} = b^*$) and analogously for $\flat_i$, if $i \in \{0, \dots , n-1\}$. Furthermore, for $\ell = 1, \dots , n-1$, we impose the condition that either $\sharp_\ell = \cdot$ and $\flat_\ell = *$ or $\sharp_\ell=*$ and $\flat_\ell = \cdot$ (so that the product $a_{y_\ell}^{\sharp_\ell} a_{x_{\ell+1}}^{\flat_\ell}$ always preserves the number of particles). If $\flat_{i-1} = \cdot$ and $\sharp_i = *$ (i.e. if the product $a_{x_i}^{\flat_{i-1}} a_{y_i}^{\sharp_i}$ for $i=2,\dots , n$, or the product $b_{x_1}^{\flat_0} a_{y_1}^{\sharp_1}$ for $i=1$, is not normally ordered) we require additionally $x \to j_i (x;x)$ to be integrable. An operator of the form (\ref{eq:Pi2}), with all the properties listed above, will be called 
a $\Pi^{(2)}$-operator of order $n$.

Next, we define 
\begin{equation}\label{eq:Pi1} \Pi^{(1)}_{\sharp,\flat} (j_1, \dots ,j_n;f) = \int b^{\flat_0}_{x_1} a_{y_1}^{\sharp_1} a_{x_2}^{\flat_1} a_{y_2}^{\sharp_2} a_{x_3}^{\flat_2} \dots  a_{y_{n-1}}^{\sharp_{n-1}} a_{x_n}^{\flat_{n-1}} a^{\sharp_n}_{y_n} a^{\flat n} (f)\, \prod_{\ell=1}^n j_\ell (x_\ell; y_\ell) \, dx_\ell dy_\ell 
\end{equation}
where $f \in L^2 (\bR^3)$, $j_k \in L^2 (\bR^3 \times \bR^3)$ for all $k =1, \dots , n$, $\sharp = (\sharp_1, \dots , \sharp_n) \in \{ \cdot , * \}^{n}$, $\flat = (\flat_0, \dots , \flat_n)\in \{ \cdot , * \}^{n+1}$ with the condition that, for all $\ell =1, \dots , n$, we either have $\sharp_\ell = \cdot$ and $\flat_\ell = *$ or $\sharp_\ell = *$ and $\flat_\ell = \cdot$. Additionally, we assume that $x \to j_i  (x;x)$ is integrable, if $\flat_{i-1} = \cdot$ and $\sharp_i = *$ for an $i=1,\dots , n$. An operator of the form (\ref{eq:Pi1}) will be called a $\Pi^{(1)}$-operator of order $n$. Operators of the form $b(f)$, $b^* (f)$, for a $f \in L^2 (\bR^3)$, will be called $\Pi^{(1)}$-operators of order zero. It will also be useful to consider
\begin{equation}\label{eq:wtPi1} \wt{\Pi}^{(1)}_{\sharp,\flat} (j_1, \dots ,j_n;f) = \int a^{\flat_0} (f) a_{x_1}^{\sharp_0} a_{y_1}^{\flat_1} a_{x_2}^{\sharp_1} a_{y_2}^{\flat_2} a_{x_3}^{\sharp_2} \dots a_{y_{n-1}}^{\flat_{n-1}} a_{x_n}^{\sharp_{n-1}}  b_{y_n}^{\flat_n} \prod_{\ell=1}^n j_\ell (x_\ell ; y_\ell) dx_\ell dy_\ell 
\end{equation}
where $f \in L^2 (\bR^3)$, $j_k \in L^2 (\bR^3 \times \bR^3)$ for all $k =1, \dots , n$, $\sharp = (\sharp_0, \dots , \sharp_{n-1}) \in \{ \cdot , * \}^{n}$, $\flat = (\flat_0, \dots , \flat_n) \in \{ \cdot , * \}^{n+1}$ with the condition that, for every $\ell \in \{0,\dots , n-1\}$, either $\sharp_\ell = \cdot$ and $\flat_\ell = *$ or $\sharp_\ell = *$ and $\flat_\ell = \cdot$. As above, we also assume that $x \to j_i  (x;x)$ is integrable, if $\flat_{i-1} = \cdot$ and $\sharp_i = *$, for $i=1,\dots , n$. Observe that 
\[ \Pi^{(1)}_{\sharp, \flat} (j_1, \dots, j_n ; f)^* = \wt{\Pi}_{\sharp', \flat'}^{(1)} (j_n, \dots , j_1 ; f) \]
with $\flat' = (\bar{\flat}_n , \dots, \bar{\flat}_0)$, $\sharp' = (\bar{\sharp}_n, \dots , \bar{\sharp}_1)$, where $\bar{\flat} = \cdot$ if $\flat = *$ and $\bar{\flat} = *$ if $\flat = \cdot$ (and similarly for $\bar{\sharp}$).

Notice that $\Pi^{(2)}$-operators involve two $b$ operators and therefore may create or annihilate up to two excitations of the condensate (depending on the choice of $\flat_0$ and $\sharp_n$, they may also leave the number of excitations invariant). $\Pi^{(1)}$- and $\wt{\Pi}^{(1)}$-operators, on the other hand, create or annihilate exactly one excitation. The conditions on the number of creation and annihilation operators guarantee that $\Pi^{(2)}$-, $\Pi^{(1)}$- and $\wt{\Pi}^{(1)}$-operators always map $\cF^{\leq N}$ back into itself. In the next lemma we collect bounds that we are going to use to control these operators. 
\begin{lemma}\label{lm:Pi-bds}
Let $n \in \bN$, $f \in L^2 (\bR^3)$, $j_1,\dots , j_n \in L^2 (\bR^3 \times \bR^3)$. We assume the operators $\Pi^{(2)}_{\sharp,\flat} (j_1,\dots , j_n)$ and $\Pi^{(1)}_{\sharp,\flat} (j_1,\dots, j_n ; f)$ are defined as in (\ref{eq:Pi2}), (\ref{eq:Pi1}). Then we have the bounds
\begin{equation}\label{eq:Pi-bds} \begin{split} 
\left\| \Pi^{(2)}_{\sharp,\flat} (j_1,\dots ,j_n) \xi \right\| &\leq 6^n \prod_{\ell=1}^n K_\ell^{\flat_{\ell-1}, \sharp_\ell} \left\| (\cN+1)^n \left(1- \frac{\cN-2}{N} \right) \xi \right\| \\
 \left\| \Pi^{(1)}_{\sharp,\flat} (j_1,\dots , j_n;f) \xi \right\| 
 &\leq 6^n \| f \| \prod_{\ell=1}^n K_\ell^{\flat_{\ell-1}, \sharp_\ell} \left\| (\cN+1)^{n+1/2} \left(1- \frac{\cN-2}{N} \right)^{1/2} \xi \right\| 
 \end{split} \end{equation} 
where
\[ K_\ell^{\flat_{\ell-1}, \sharp_\ell} = \left\{ \begin{array}{ll} \|j_\ell \|_2 + \int |j_\ell (x;x)| \, dx \quad &\text{if } \flat_{\ell-1} = \cdot \text{ and } \sharp_\ell = * \\
\| j_\ell \|_2 \quad &\text{otherwise} \end{array} \right. \]
Since $\cN \leq N$ on $\cF^{\leq N}$, it follows that $\Pi^{(2)}_{\sharp,\flat} (j_1, \dots , j_n), \Pi^{(1)}_{\sharp,\flat} (j_1, \dots , j_n ;f)$ are bounded operators on $\cF^{\leq N}$, with 
\[ \begin{split} 
\left\|  \Pi^{(2)}_{\sharp,\flat} (j_1,\dots, j_n) \right\| &\leq (12 N)^n \prod_{\ell=1}^n K_\ell^{\flat_{\ell-1}, \sharp_\ell} \\ 
\left\|  \Pi^{(1)}_{\sharp,\flat} (j_1,\dots ,j_n;f) \right\| &\leq (12 N)^n \sqrt{N} \| f \|_2 \prod_{\ell=1}^n K_\ell^{\flat_{\ell-1}, \sharp_\ell} 
\end{split} \]
\end{lemma}

{\it Remark:} if $j_i \in (q_{\ph^{\bar{\flat}_{i-1}}} \otimes q_{\ph^{\bar{\sharp}_i}}) L^2 (\bR^3 \times \bR^3)$ for all $i =1, \dots , n$ and if $f \in L^2_\perp (\bR^3)$, then $\Pi^{(2)}_{\sharp,\flat} (j_1, \dots , j_n)$ and $\Pi^{(1)}_{\sharp,\flat} (j_1, \dots , j_n; f)$ map $\cF_{\perp \ph}^{\leq N}$ into itself. 

\begin{proof}
We consider operators of the form (\ref{eq:Pi2}). Let us assume, for example, that $\flat_0 = \cdot$ and $\sharp_n = \cdot$. Then we have, writing $b_{x_1} = a_{x_1} (1-\cN/N)^{1/2}$ and $b_{y_n} = a_{y_n} (1-\cN/N)^{1/2}$ and using the pull-through formula $g(\cN) a_x = a_x g(\cN - 1)$,  
\[ \begin{split} \Pi^{(2)}_{\sharp,\flat} &(j_1, \dots , j_n) \\ &= \int a_{x_1} \left(\frac{N-\cN}{N} \right)^{1/2}  a_{y_1}^{\sharp_1} \dots a_{y_{n-1}}^{\sharp_{n-1}} a_{x_n}^{\flat_{n-1}} a_{y_n} \left(\frac{N-\cN}{N} \right)^{1/2} \prod_{\ell=1}^n j_\ell (x_\ell;y_\ell) dx_\ell dy_\ell \\
&= \int a_{x_1} a_{y_1}^{\sharp_1} \dots a_{y_{n-1}}^{\sharp_{n-1}} a_{x_n}^{\flat_{n-1}} a_{y_n} \left(\frac{N-\cN+1}{N} \right)^{1/2} \left(\frac{N-\cN}{N} \right)^{1/2} \prod_{\ell=1}^n j_\ell (x_\ell;y_\ell) dx_\ell dy_\ell \\
&= \prod_{\ell=1}^n A^{\flat_{\ell-1},\sharp_\ell} (j_\ell)  \left(\frac{N-\cN+1}{N} \right)^{1/2} \left(\frac{N-\cN}{N} \right)^{1/2}
\end{split} \]
where we used the definition (\ref{eq:AAdef}). The first bound in (\ref{eq:Pi-bds}) follows therefore from Lemma \ref{lm:Abds}. The other estimates can be shown similarly.
\end{proof}

\section{Generalized Bogoliubov Transformations}
\label{sec:Bog}

For a kernel $\eta \in L^2 (\bR^3 \times \bR^3)$ with $\eta (x;y) = \eta (y;x)$, we define 
\begin{equation}\label{eq:defB} B(\eta) = \frac{1}{2} \int \left[ \eta (x;y) b_x^* b_y^* - \bar{\eta} (x;y) b_x b_y \right] dx dy \end{equation}
Observe that, with the notation introduced in (\ref{eq:nota}), \[ B (\eta) = \frac{1}{2} \left[ B_{*,*} (\eta) - B^*_{*,*} (\eta) \right] = - \frac{1}{2} \left[ B_{\cdot, \cdot} (\eta) - B^*_{\cdot, \cdot} (\eta) \right] \, . \] 
Generalized Bogoliubov transformations are unitary operators having the form 
\begin{equation}\label{eq:gen-Bog} e^{B(\eta)} = \exp \left[ \frac{1}{2} \int (\eta (x;y) b_x^* b_y^* - \bar{\eta} (x;y) b_x b_y ) \right] \end{equation}
It is clear that $B(\eta), e^{B(\eta)} : \cF^{\leq N} \to \cF^{\leq N}$. Furthermore, if $\eta \in (q_\ph \otimes q_{{\ph}}) L^2 (\bR^3 \times \bR^3)$ then we have $B(\eta), e^{B(\eta)} : \cF_{\perp \ph}^{\leq N} \to \cF_{\perp \ph}^{\leq N}$ for any  normalized $\ph \in L^2 (\bR^3)$ (as above, $q_\ph = 1- |\ph \rangle \langle \ph|$ is the projection into the orthogonal complement of $\ph$). It may be helpful to observe that, with the unitary operator $U(\ph)$ defined in (\ref{eq:Uph-def}), we 
can write, according to (\ref{eq:UaU}),  
\begin{equation}\label{eq:Beta-tran} B(\eta) = \frac{1}{2} U(\ph) \int dx dy \left[ \eta (x;y) a_x^* a_y^* \frac{a(\ph) a(\ph)}{N} - \bar{\eta} (x;y) \frac{a^* (\ph) a^* (\ph)}{N} a_x a_y \right] U^* (\ph) \end{equation}
On states exhibiting Bose-Einstein condensation in $\ph$ (so that $a(\ph), a^* (\ph) \simeq \sqrt{N}$), we can therefore expect the generalized Bogoliubov transformation (\ref{eq:gen-Bog}) to be close to the standard Bogoliubov transformation 
\begin{equation}\label{eq:bogo2}  e^{\wt{B} (\eta)} = \exp \left[ \frac{1}{2} \int (\eta (x;y) a_x^* a_y^* - \bar{\eta} (x;y) a_x a_y ) \right] \end{equation}
whose action on creation and annihilation operators is explicitly given by 
\begin{equation}\label{eq:act-Bog} e^{-\wt{B} (\eta)} a(f) e^{\wt{B} (\eta)} = a(\cosh_{\eta} (f)) + a^* (\sinh_{\eta} (\bar{f})) \end{equation}
with the operators $\cosh_\eta , \sinh_\eta$ defined as in (\ref{eq:chsh}). Standard Bogoliubov transformations of the form (\ref{eq:bogo2}) have been used in \cite{BDS} to model correlations in the Gross-Pitaevskii regime, for approximately coherent Fock space initial data. In the present paper, since (\ref{eq:bogo2}) does not map $\cF_{\perp \ph}^{\leq N}$ into itself (it does not respect the truncation $\cN \leq N$), we prefer to work with generalized Bogoliubov transformations of the form (\ref{eq:gen-Bog}). The price that we have to pay is the fact that, in contrast to (\ref{eq:act-Bog}), the action of $\exp (B(\eta))$ on creation and annihilation operators is not explicit. Let us remark here that generalized Bogoliubov transformations of the form $\exp (B(\eta))$ have already been used in \cite{Sei,GS} to study the excitation spectrum in the mean field regime. Here we will need more detailed information on the action of these operators; the rest of this section is therefore devoted to the study of the properties of generalized Bogoliubov transformations. 

First of all, we need the following generalization of Lemma 4.3 of \cite{BDS} (a similar result has also been proven in \cite{Sei}). 
\begin{lemma} \label{lm:Npow} Let $\eta \in L^2 (\bR^3 \times \bR^3)$. Let $B(\eta)$ be the antisymmetric operator defined in (\ref{eq:defB}). For every $n_1, n_2 \in \bZ$ there exists a constant $C = C ( n_1, n_2, \| \eta \|_2 )$ such that 
\[ e^{-B(\eta)} (\cN +1)^{n_1} \left(N + 1 - \cN\right)^{n_2} e^{B(\ph)} \leq C (\cN+1)^{n_1} (N+1 - \cN)^{n_2} \]
as operator inequality on $\cF^{\leq N}$.
\end{lemma}
\begin{proof}
We use Gronwall's inequality. For a fixed $\xi \in \cF^{\leq N}$ and $s \in [0;1]$, let 
\[ f(s) = \left\langle \xi , e^{-s B(\eta)} (\cN+1)^{n_1} (N+1- \cN)^{n_2} e^{s B(\eta)} \xi \right\rangle \]
We compute
\begin{equation}\label{eq:derfs} \begin{split} 
f' (s) &= \langle \xi, e^{-sB(\eta)} \left[ (\cN+1)^{n_1} (N+1- \cN)^{n_2}, B(\eta) \right] e^{sB(\eta)} \xi \rangle \\ &= \left\langle e^{sB(\eta)} \xi , \left\{ (\cN+1)^{n_1} [ (N+1-\cN)^{n_2}, B(\eta)] \right.\right. \\ &\hspace{3cm} \left. \left. + [ (\cN+1)^{n_1}, B(\eta)] (N+1-\cN)^{n_2} \right\} e^{s B(\eta)} \xi \right\rangle \end{split} \end{equation}
{F}rom the pull-through formula $\cN b^* = b^* (\cN+1)$, we conclude that 
\[ \begin{split} [(N+1- \cN)^{n_2}, B(\eta)] &= \frac{1}{2} B_{*,*} (\eta) \, \left[ (N-1-\cN)^{n_2} - (N+1-\cN)^{n_2} \right] + \text{h.c.} \\ 
[(\cN+1)^{n_1}, B(\eta)] &= \frac{1}{2} B_{*, *} (\eta) \, \left[ (\cN+3)^{n_1} - (\cN+1)^{n_1} \right] + \text{h.c.} \end{split} \]
By the mean value theorem, we can find functions 
$\theta_1, \theta_2 : \bN \to (0;2)$ (depending also on $N, n_1, n_2$) such that 
\[ \begin{split} (N-j+1)^{n_2} - (N-j-1)^{n_2} &= 2n_2 (N+1-j-\theta_1 (j))^{n_2-1} \\ (j+3)^{n_1} - (j+1)^{n_1} &= 2n_1 (j+1+\theta_2 (j)) \end{split} \]
Hence, the first term on the r.h.s. of (\ref{eq:derfs}) can be written as 
\[\begin{split} &\langle e^{sB(\eta)} \xi , (\cN+1)^{n_1} [(N+1-\cN)^{n_2}, B(\eta)] e^{sB(\eta)} \xi \rangle \\ &\hspace{2cm} = \frac{1}{2} \langle (\cN+1)^{n_1} e^{sB(\eta)} \xi,\left( B_{*,*} (\eta) (N+1-\cN - \theta_1 (\cN))^{n_2-1}+h.c.\right) e^{sB(\eta)} \xi \rangle \\
&\hspace{2cm} = \frac{1}{2} \langle (\cN+1)^{n_1/2} (N+3-\cN-\theta_1 (\cN-2))^{n_2/2} e^{sB(\eta)} \xi, \\ &\hspace{4cm} B_{*,*} (\eta) (\cN+3)^{n_1/2} (N+1-\cN - \theta_1 (\cN))^{n_2/2-1} e^{sB(\eta)} \xi \rangle\\
&\hspace{2.5cm} + \frac{1}{2} \langle (\cN+1)^{n_1/2} (N+1-\cN - \theta_1 (\cN))^{n_2/2} e^{sB(\eta)} \xi, \\ &\hspace{4cm}   B_{\cdot,\cdot} (\eta) (\cN-1)^{n_1/2}(N+3-\cN-\theta_1 (\cN-2))^{n_2/2-1}e^{sB(\eta)} \xi \rangle \end{split} \]
The Cauchy-Schwarz inequality implies with Lemma \ref{lm:Bbds} 
\[ \begin{split} 
\Big| \langle e^{sB(\eta)} \xi , &(\cN+1)^{n_1} [(N+1-\cN)^{n_2}, B(\eta)] e^{sB(\eta)} \xi \rangle \Big| \\ &\leq C \left\| (\cN+1)^{n_1/2} (N+3 -\cN - \theta_1 (\cN-2))^{n_2/2} e^{sB(\eta)} \xi \right\| \\ &\hspace{3cm} \times \left\| (\cN+3)^{n_1/2 + 1} (N+1-\cN - \theta_1 (\cN))^{n_2} N^{-1}  e^{sB(\eta)} \xi \right\| \end{split}  \]
with a constant $C$ depending on $\| \eta \|_2$. Since on $\cF^{\leq N}$ we have $\cN\leq N$ and since $0 \leq \theta_1 (n) \leq 2$ for all $n \in \bN$, we conclude that
\[ 
\Big| \langle e^{sB(\eta)} \xi , (\cN+1)^{n_1} [(N+1-\cN)^{n_2}, B(\eta)] e^{sB(\eta)} \xi \rangle \Big| \leq C f(s) \]
for a constant $C$ depending on $\| \eta \|_2, n_1, n_2$. The second term on the r.h.s. of (\ref{eq:derfs}) can be bounded similarly. We infer that $f' (s) \leq C f(s)$. Gronwall's inequality implies that $f(s) \leq e^{Cs} f(0)$. Hence, taking $s=1$, and renaming the constant $C$, we obtain
\[ \left\langle \xi , e^{-B(\eta)} (\cN+1)^{n_1} (N+1- \cN)^{n_2} e^{B(\eta)} \xi \right\rangle \leq C \left\langle \xi, (\cN+1)^{n_1} (N+1-\cN)^{n_2} \xi \right\rangle \]
which concludes the proof of the lemma.
\end{proof}

We will need to express the action of the generalized Bogoliubov transformation $e^{B(\eta)}$ on the $b$-fields by means of a convergent series of nested commutators. To this end, we start by noticing that, for $f \in L^2 (\bR^3)$, 
\[\begin{split} e^{-B(\eta)} b (f) e^{B(\eta)} &= b(f) + \int_0^1 ds \, \frac{d}{ds}  e^{-sB(\eta)} b(f) e^{sB(\eta)} \\ &= b(f) - \int_0^1 ds \, e^{-sB(\eta)} [B(\eta), b(f)] e^{s B(\eta)} \\ &= b(f) - [B(\eta),b(f)] + \int_0^1 ds_1 \int_0^{s_1} ds_2 \, e^{-s_2 B(\eta)} [B(\eta), [B(\eta),b(f)] e^{s_2 B(\eta)} 
\end{split} \]
Iterating $m$ times, we obtain
\begin{equation}\label{eq:BCH} \begin{split} 
e^{-B(\eta)} b (f) e^{B(\eta)} = &\sum_{n=1}^{m-1} (-1)^n \frac{\text{ad}^{(n)}_{B(\eta)} (b(f))}{n!} \\ &+ \int_0^{1} ds_1 \int_0^{s_1} ds_2 \dots \int_0^{s_{m-1}} ds_m \, e^{-s_m B(\eta)} \text{ad}^{(m)}_{B(\eta)} (b(f)) e^{s_m B(\eta)} \end{split} \end{equation}
where we introduced the notation $\text{ad}_{B(\eta)}^{(n)} (A)$  defined recursively by
\[ \text{ad}_{B(\eta)}^{(0)} (A) = A \quad \text{and } \quad \text{ad}^{(n)}_{B(\eta)} (A) = [B(\eta), \text{ad}^{(n-1)}_{B(\eta)} (A) ]  \]
We will show later that, under suitable assumptions on $\eta$, the error term on the r.h.s. of (\ref{eq:BCH}) is negligible in the limit $m \to \infty$. This means that the action of the generalized Bogoliubov transformation $B(\eta)$ on $b(f)$ and similarly on $b^* (f)$ can be described in terms of the nested commutators $\text{ad}_{B(\eta)} (A)$, for $A = b(f)$ or $A = b^* (f)$. In the next lemma, we give a detailed analysis of these terms. 

For a kernel $\eta \in L^2 (\bR^3 \times \bR^3)$, we will use the notation
\begin{equation}\label{eq:etas} \eta^{(n)} = \left\{ \begin{array}{ll}  1, \quad &\text{for } n = 0 \\ (\eta \bar{\eta})^\ell, \quad &\text{if } n = 2\ell, \ell \in \bN \backslash \{ 0 \} \\
(\eta \bar{\eta})^\ell \eta \quad &\text{if } n = 2\ell+1, \ell \in \bN \end{array} \right. \end{equation}
Here we, identify $\eta \in L^2 (\bR^3 \times \bR^3)$ with the Hilbert-Schmidt operator acting on $L^2 (\bR^3)$, having integral kernel $\eta$. To avoid keeping track of complex conjugations of $\eta$-kernels, we also introduce the following notation. For $\natural \in \{\cdot, * \}$ we write $\eta_\natural = \eta$, if $\natural = \cdot$, and $\eta_\natural = \bar{\eta}$ if $\natural = *$. More generally, for $n \in \bN$, and $(\natural_1,\dots , \natural_n) \in \{ \cdot ,* \}^n$, we will use the notation $\eta_{\natural}^{(n)} = \eta_{\natural_1} \eta_{\natural_2} \dots \eta_{\natural_n}$, in the sense of products of operators. Also for a function $f \in L^2 (\bR^3)$, we use the notation $f_\natural = f$ if $\natural = \cdot$ and $f_{\natural} = \bar{f}$ if $\natural= *$. 
\begin{lemma}\label{lm:indu}
Let $\eta \in L^2 (\bR^3 \times \bR^3)$ be such that $\eta (x;y) = \eta (y;x)$ for all $x,y \in \bR^3$. Let $B(\eta)$ be defined as in (\ref{eq:defB}). Let $n\in \bN$ and $f \in L^2 (\bR^3)$. Then the nested commutators $\text{ad}^{(n)}_{B(\eta)} (b(f))$ can be written as the sum of exactly $2^n n!$ terms, with the following properties. 
\begin{itemize}
\item[i)] Possibly up to a sign, each term has the form
\begin{equation}\label{eq:Lambdas} \Lambda_1 \Lambda_2 \dots \Lambda_i \frac{1}{N^k} \Pi^{(1)}_{\sharp,\flat} (\eta_{\natural_1}^{(j_1)}, \dots , \eta_{\natural_k}^{(j_k)} ; \eta_{\natural}^{(s)} (f_{\lozenge})) 
\end{equation}
for some $i,k,s \in \bN$, $j_1, \dots ,j_k \in \bN \backslash \{ 0 \}$, $\lozenge \in \{ \cdot, * \}$, $\sharp \in \{ \cdot, * \}^k$, $ \flat \in \{ \cdot, * \}^{k+1}$, $\natural_v \in \{ \cdot, * \}^{j_v}$ for all $v=1,\dots , k$ and $\natural \in \{ \cdot, * \}^s$. In (\ref{eq:Lambdas}), each operator $\Lambda_w : \cF^{\leq N} \to \cF^{\leq N}$ is either a factor $(N-\cN)/N$, a factor $(N+1-\cN)/N$ or an operator of the form
\begin{equation}\label{eq:Pi2-ind} \frac{1}{N^p} \Pi^{(2)}_{\sharp,\flat} (
\eta^{(m_1)}_{\natural_1}, \eta^{(m_2)}_{\natural_2},\dots , \eta_{\natural_p}^{(m_p)}) \end{equation}
for some $p, m_1, \dots , m_p \in \bN \backslash \{ 0 \}$, $\sharp,\flat  \in \{ \cdot , *\}^p$, $\natural_v \in \{ \cdot , * \}^{m_v}$ for all $v =1, \dots , p$. 
\item[ii)] If a term of the form (\ref{eq:Lambdas}) contains $m \in \bN$ factors $(N-\cN)/N$ or $(N+1-\cN)/N$ and $j \in \bN$ factors of the form (\ref{eq:Pi2-ind}) with $\Pi^{(2)}$-operators of order $p_1, \dots , p_j \in \bN \backslash \{ 0 \}$, then 
we have
\begin{equation}\label{eq:totalb} m + (p_1 + 1)+ \dots + (p_j+1) + (k+1) = n+1 \end{equation}
\item[iii)] If a term of the form (\ref{eq:Lambdas}) contains (considering all $\Lambda$-operators and the $\Pi^{(1)}$-operator) the kernels $\eta_{\natural_1}^{(i_1)}, \dots , \eta_{\natural_m}^{(i_m)}$ and the wave function $\eta_{\natural}^{(s)} (f_{\lozenge})$ for some $m, s \in \bN$, $i_1, \dots , i_m \in \bN \backslash \{ 0 \}$, $\natural_r \in \{\cdot, * \}^{i_r}$ for all $r=1, \dots , m$, $\natural \in \{ \cdot ,* \}^s$ then \[ i_1 + \dots + i_m + s = n .\]
\item[iv)] There is exactly one term having the form  
\begin{equation}\label{eq:iv1} \left( \frac{N-\cN}{N} \right)^{n/2} \left(\frac{N+1-\cN}{N} \right)^{n/2} b(\eta^{(n)} (f)) \end{equation}
if $n$ is even, and 
\begin{equation}\label{eq:iv2} - \left(\frac{N-\cN}{N} \right)^{(n+1)/2} \left(\frac{N-\cN+1}{N} \right)^{(n-1)/2} b^* (\eta^{(n)} (\bar{f})) \end{equation}
if $n$ is odd.
\item[v)] If the $\Pi^{(1)}$-operator in  (\ref{eq:Lambdas}) is of order $k \in \bN \backslash \{ 0 \}$, it has either the form  
\[ \int  b^{\flat_0}_{x_1} \prod_{i=1}^{k-1} a^{\sharp_i}_{y_{i}} a^{\flat_i}_{x_{i+1}}  a^*_{y_k} a (\eta^{(2r)} (f)) \prod_{i=1}^k \eta_{\natural_i}^{(j_i)} (x_i ; y_i) dx_i dy_i  \]
or the form 
\[\int b^{\flat_0}_{x_1} \prod_{i=1}^{k-1} a^{\sharp_i}_{y_{i}} a^{\flat_i}_{x_{i+1}}  a_{y_k} a^* (\eta^{(2r+1)} (\bar{f})) \prod_{i=1}^k \eta_{\natural_i}^{(j_i)} (x_i ; y_i) dx_i dy_i \]
for some $r \in \bN$, $j_1, \dots , j_k \in \bN \backslash \{ 0 \}$. If it is of order $k=0$, then it is either given by $b (\eta^{(2r)}_\natural (f_\lozenge))$ or by $b^* (\eta^{(2r+1)}_\natural (f_\lozenge))$, for some $r \in \bN$. 
\item[vi)] For every non-normally ordered term of the form 
\[ \begin{split} &\int dx dy \, \eta^{(i)}_\natural (x;y) a_x a_y^* , \quad \int dx dy \, \eta_\natural^{(i)} (x;y) b_x a_y^* \\  &\int dx dy \, \eta_\natural^{(i)} (x;y) a_x b_y^*, \quad \text{or } \quad \int dx dy \, \eta_\natural^{(i)} (x;y) b_x b_y^*  \end{split} \]
appearing either in the $\Lambda$-operators or in the $\Pi^{(1)}$-operator in (\ref{eq:Lambdas}), we have $i \geq 2$.
\end{itemize}
\end{lemma}

{\it Remark:} Similarly, the nested commutator $\text{ad}^{(n)} (b^* (f))$ can be written as the sum of $2^n n!$ terms of the form 
\[ \frac{1}{N^k} \wt{\Pi}^{(1)}_{\sharp,\flat} (\eta_{\natural_1}^{(j_1)}, \dots, \eta^{(j_k)}_{\natural_{k}} ; \eta_{\natural_{k+1}}^{(\ell)} (f_\lozenge)) \Lambda_1 \Lambda_2 \dots \Lambda_i \]
satisfying properties analogous to those 
listed in i)-vi). 

\begin{proof} We prove the lemma by induction. For $n=0$ all claims are trivially satisfied. For the induction step from $n$ to $n+1$ we first compute, using (\ref{eq:comm-b}) and (\ref{eq:comm-b2}) the commutators
\begin{equation}\label{2.2.Betacommutators}
    \begin{split}
    [B(\eta), b_z ] &= - \frac{N-\cN}N b^*(\eta_z) +  \frac1N\int dxdy\,\eta(x;y)\bsx\asy \az\\
    &=-b^*(\eta_z)\frac{N+1-\cN}N  + \frac1N\int dxdy\,\eta(x;y)\az\asy\bsx ,\\
    [B(\eta), \bsz ] &= -b(\eta_z)\frac{N-\cN}N  +  \frac1N\int dxdy\,\bar{\eta} (x;y)\asz\ay\bx\\
    &=- \frac{N+1-\cN}N b(\eta_z)  +   \frac1N\int dxdy\,\bar{\eta} (x;y)\bx\ay \asz ,\\
    [B(\eta), a^*_z a_w] & = [B(\eta), a_w \asz]= - \bsz b^*(\eta_w) - b(\eta_z) b_w,\\
    [B(\eta), N-\cN]&=[B(\eta), N+1-\cN]= \int dx dy\, (\eta(x,y)\bsx\bsy + \bar{\eta} (x;y) \bx\by).
    \end{split}
    \end{equation}
{F}rom $ \operatorname {ad}_{B(\eta)}^{(n+1)}(b(f))=[B(\eta),\operatorname{ad}_{B(\eta)}^{(n)}(b(f))]$ and by linearity, it is enough to analyze
    \begin{equation} \label{eq:comm-step}
    \left[ B(\eta), \Lambda_1 \Lambda_2\dots \Lambda_i N^{-k} \Pi^{(1)}_{\sharp,\flat} (\eta_{\natural_1}^{(j_1)}, \dots , \eta_{\natural_k}^{(j_k)} ;  \eta_{\natural_{k+1}}^{(\ell)} (f_{\lozenge}) ) \right] 
\end{equation}
with the operator $ \Lambda_1 \Lambda_2\dots \Lambda_i N^{-k} \Pi_{\sharp,\flat}^{(1)} (\eta_{\natural_1}^{(j_1)}, \dots , \eta_{\natural_k}^{(j_k)} ;  \eta_{\natural}^{(s)} (f_{\lozenge}))$ satisfying properties (i) to (vi). Applying Leibniz rule $[A,BC] = [A,B] C + B [A,C]$, the commutator (\ref{eq:comm-step}) is given by a sum of terms, where $B(\eta)$ is either commuted with a $\Lambda$-operator, or with the $\Pi^{(1)}$-operator. 

Let's consider first the case that $B(\eta)$ is commuted with a $\Lambda$-operator, assuming further that $\Lambda$ is either the operator $(N-\cN)/N$ or the operator $(N+1-\cN)/N$. The last line in \eqref{2.2.Betacommutators} implies that such an operator $\Lambda$ is replaced, after commutation with $B(\eta)$, by the sum \begin{equation}\label{eq:Pi2-repl} N^{-1} \Pi^{(2)}_{*,*} (\eta) + N^{-1}\Pi^{(2)}_{\cdot, \cdot} (\bar{\eta}).
\end{equation}
With this replacement, we generate two terms contributing to $\text{ad}^{(n+1)}_{B(\eta)} (b(f))$. Let us check that these new terms satisfy the properties (i)-(vi) (of course, with $n$ replaced by $(n+1)$). (i) is obviously true. Also (ii) remains valid, because replacing a factor $(N-\cN)/N$ or $(N+1-\cN)/N$ by one of the two summands in (\ref{eq:Pi2-repl}), the index $m$ will decrease by one, but there will be an additional factor of $2$ because we added a $\Pi^{(2)}$-operator of the order one. Since exactly one additional kernel $\eta_\natural$ is inserted, also (iii) continues to hold true. The factor $\Pi^{(1)}$ is not affected by the replacement, hence the new terms will continue to satisfy (v). Furthermore, since both terms in (\ref{eq:Pi2-repl}) are normally ordered, also (vi) remains valid, by the induction assumption. We observe, finally, that the two terms we generated here do not have the form appearing in (iv).
    
Next, we consider the commutator of $B(\eta)$ with a $\Lambda$-operator of the form $\Lambda =N^{-p} \Pi^{(2)}_{\sharp,\flat} (\eta_{\natural_1}^{(m_1)},\dots,\eta_{\natural_p}^{(m_p)})$ for a $p\in \bN$ ($p\leq n$ by (ii)). By definition
    \begin{equation}\label{eq:L2}
    \Lambda =  N^{-p} \int \, b^{\flat_0}_{x_1} \prod_{i=1}^{p-1} a^{\sharp_i}_{y_i} a^{\flat_i}_{x_{i+1}}  b_{y_p}^{\sharp_p} \prod_{i=1}^{p}  \eta^{(m_i)}_{\natural_i} (x_i; y_i) dx_i dy_i 
    \end{equation}
If $ [B(\eta), \cdot]$ hits $ b^{\flat_0}_{x_1}$, the first two relations in \eqref{2.2.Betacommutators}  imply that $\Lambda$ is replaced by a sum of two operators, the first one being either 
\begin{equation}\label{eq:first-L} \begin{split} -\frac{N-\cN}{N} N^{-p} \Pi^{(2)}_{\sharp,\wt{\flat}} &(\eta_{\natural_1}^{(m_1+1)}, \eta_{\natural_2}^{(m_2)}, \dots , \eta_{\natural_p}^{(m_p)}) \quad \text{or } \\  &\quad -\frac{N+1-\cN}{N} N^{-p} \Pi^{(2)}_{\sharp,\wt{\flat}} (\eta_{\natural_1}^{(m_1+1)}, \eta_{\natural_2}^{(m_2)}, \dots , \eta_{\natural_p}^{(m_p)}) \end{split} \end{equation}
depending on whether $\flat_0 = \cdot$ or $\flat_0 = *$ (here $\wt{\flat} = (\bar{\flat}_0, \flat_1, \dots , \flat_{p-1})$ with $\bar{\flat}_0 = \cdot$ if $\flat_0 = *$ and $\bar{\flat}_0 = *$ if $\flat_0 = \cdot$). The second operator emerging when $[B(\eta), \cdot ]$ hits $b_{x_1}^{\flat_0}$ is a $\Pi^{(2)}$-operator of order $(p+1)$, given by
\begin{equation}\label{eq:second-L} N^{-(p+1)} \Pi^{(2)}_{\wt{\sharp},\wt{\flat}} (\eta_{\natural_0}, \eta^{(m_1)}_{\natural_1}, \dots , \eta^{(m_p)}_{\natural_p}) \end{equation}
where $\wt{\sharp} = (\bar{\flat}_0, \sharp_1, \dots , \sharp_p)$, $\wt{\flat} = (\bar{\flat}_0, \flat_0, \dots , \flat_{p-1})$ and $\natural_0 = \flat_0$. 

For both terms (\ref{eq:first-L}) and (\ref{eq:second-L}), (i) is clearly correct and also (ii) remains true (when we replace (\ref{eq:L2}) with (\ref{eq:first-L}), the number of $(N-\cN)/N$ or $(N-\cN+1)/N$-operators increases by one, while everything else remains unchanged; similarly, when we replace (\ref{eq:L2}) with (\ref{eq:second-L}), the order of the $\Pi^{(2)}$-operator increases by one, while the rest remains unchanged). (iii) remains true as well, since, in (\ref{eq:first-L}), the power $m_1 +1$ of the first $\eta$-kernel is increased by one unit and, in (\ref{eq:second-L}), there is one additional factor $\eta$, compared with (\ref{eq:L2}). (v) remains valid, since the $\Pi^{(1)}$-operator on the right is not affected by this commutator. (vi) remains true in (\ref{eq:first-L}), because $m_1+1 \geq 2$ . It remains true also in (\ref{eq:second-L}). In fact, according to (\ref{2.2.Betacommutators}), when switching from (\ref{eq:L2}) to (\ref{eq:second-L}), we are effectively replacing $b \to b^* a^* a$ or $b^* \to b a a^*$. Hence, the first pair of operators in (\ref{eq:second-L}) is always normally ordered. As for the second pair of creation and annihilation operators (the one associated with the kernel $\eta_{\natural_1}^{(m_1)}$ in (\ref{eq:second-L})), the first field is of the same type as the original $b$-field appearing in (\ref{eq:L2}); hence non-normally ordered pairs cannot be created. Finally, we remark that the terms we generated here are certainly not of the form in (iv) (because for terms as in (iv) all $\Lambda$-factors must be either $(N-\cN)/N$ or $(N+1-\cN)/N$, and this is not the case, for terms containing (\ref{eq:first-L}) or (\ref{eq:second-L})). 

The same arguments can be applied if $B(\eta)$ hits the factor $b^{\sharp_p}_{y_p}$ on the right of (\ref{eq:L2}) (in this case, we use the identities for the first two commutators in \eqref{2.2.Betacommutators} having the $b$-field to the left of the factors $(N+1-\cN)/N$ and $(N-\cN)/N$ and to the right of the $a_z a_y^*$ and $a_z^* a_y$ operators). 

If now $B(\eta)$ hits a term $a^*_{y_{r}}a_{x_{r+1}}$ or $a_{y_{r}} a^*_{x_{r+1}}$ in (\ref{eq:L2}), for an $r=1, \dots, p-1$, then \eqref{2.2.Betacommutators} implies that $\Lambda = {N^{-p}} \Pi^{(2)}_{\sharp,\flat} (\eta_{\natural_1}^{(m_1)}, \dots , \eta_{\natural_p}^{(m_p)})$ is replaced by the sum of 
the two terms, given by
\begin{equation}\label{eq:replaa} - \left[ N^{-r} \Pi^{(2)}_{\sharp', \flat'} (\eta^{(m_1)}_{\natural_1}, \dots, \eta_{\natural'_r}^{(m_r+1)}) \right] \left[ N^{-(p-r)} \Pi^{(2)}_{\sharp^{''},\flat^{''}} (\eta_{\natural_{r+1}}^{(m_{r+1})}, \dots, \eta_{\natural_p}^{(m_p)}) \right] \end{equation}
and by
\begin{equation}\label{eq:replaa2} - \left[ N^{-r} \Pi^{(2)}_{\sharp^{'''}, \flat^{'}} (\eta^{(m_1)}_{\natural_1}, \dots, \eta_{\natural'_r}^{(m_r)}) \right] \left[ N^{-(p-r)} \Pi^{(2)}_{\sharp^{''},\flat^{'''}} (\eta_{\natural^{'}_{r+1}}^{(m_{r+1}+1)}, \dots, \eta_{\natural_p}^{(m_p)}) \right] \end{equation}
with $\flat' = (\flat_0, \dots , \flat_{r-1})$, $\flat'' = (\flat_r, \dots , \flat_{p-1})$, $\flat^{'''} = (\bar{\flat}_r, \flat_{r+1}, \dots , \flat_{p-1})$ and with $\sharp' = (\sharp_1, \dots , \sharp_{r-1},\bar{\sharp}_r)$, $\sharp^{''} = (\sharp_{r+1}, \dots , \sharp_p)$, $\sharp^{'''} = (\sharp_1, \dots , \sharp_r)$ (here, we denote $\bar{\sharp}_r = *$ if $\sharp_r = \cdot$ and $\bar{\sharp}_r = \cdot$ if $\sharp_r = *$, and similarly for $\bar{\flat}_{r-1}$). 
The precise form of $\natural'_r$ and $\natural'_{r+1}$ does not play an important role (they are given by $\natural'_r = (\natural_r, \sharp_r)$ and $\natural'_{r+1} = (\natural_{r+1}, \flat_r)$). The new terms containing (\ref{eq:replaa}) and (\ref{eq:replaa2}) clearly satisfy (i). Furthermore, (ii) remains true because the contribution of the original $\Lambda$ to the sum in (\ref{eq:totalb}), which was given by $(p+1)$ is now replaced by $(r+1)+(p-r+1) = p+2$. Clearly, (iii) remains true as well, since, for both terms (\ref{eq:replaa}) and (\ref{eq:replaa2}), the total powers of the $\eta$-kernels is increased exactly by one. As before, the terms we generated do not have the form (iv). (v) continues to hold true, because the $\Pi^{(1)}$ term is unaffected. As for (vi), we observe that non-normally ordered pairs can only be created where $\sharp_r$ is changed to $\bar{\sharp}_{r}$ (in the term where $\sharp'$ appears) or where $\flat_r$ is changed to $\bar{\flat}_r$ (in the term where $\flat'''$ appears). In both cases, however, the change $\sharp_r \to \bar{\sharp}_r$ and $\flat_r \to \bar{\flat}_r$ comes together with an increase in the power of the corresponding $\eta$-kernel (i.e. $\eta^{(m_r)}_{\natural_r}$ is changed to $\eta_{\natural'_r}^{(m_r+1)}$ in the first case, while $\eta^{(m_{r+1})}_{\natural_{r+1}}$ is changed to $\eta^{(m_{r+1} + 1)}_{\natural'_{r+1}}$ in the second case). Since $m_r + 1, m_{r+1} + 1 \geq 2$, even if non-normally ordered terms are created, they still satisfy (vi).

Next, let us consider the terms arising from commuting $B(\eta)$ with the operator 
\begin{equation} \label{eq:Pi1-term}
\begin{split} 
N^{-k} \Pi^{(1)}_{\sharp,\flat} (\eta_{\natural_1}^{(j_1)}, &\dots , \eta_{\natural_k}^{(j_k)}; \eta^{(s)}_{\natural} (f_{\lozenge}))\\ &= N^{-k} \int b_{x_1}^{\flat_0} \prod_{i=1}^{k-1} a_{y_i}^{\sharp_i} a_{x_{i+1}}^{\flat_i} a_{y_k}^{\sharp_k} a^{\flat_k} (\eta^{(s)}_{\natural} (f_{\lozenge})) \prod_{i=1}^k \eta^{(j_i)}_{\natural_i} (x_i ; y_i) dx_i dy_i \end{split}\end{equation}
We argue similarly to the case in which $B(\eta)$ hits a $\Pi^{(2)}$-operator like (\ref{eq:L2}). In particular, if $B(\eta)$ hits the operator $b_{x_1}^{\flat_0}$, the operator (\ref{eq:Pi1-term}) is replaced by the sum of two terms, the first one being 
\[\begin{split}  - \frac{N-\cN}{N} N^{-p} \Pi^{(1)}_{\sharp,\wt{\flat}} (\eta_{\natural'_1}^{(m_1 + 1)}, &\eta_{\natural_2}^{(m_2)}, \dots, \eta_{\natural_k}^{(m_k)};\eta^{(s)}_{\natural} (f_\lozenge)) \qquad \text{or } \\ &- \frac{N+1-\cN}{N} N^{-p} \Pi^{(1)}_{\sharp,\wt{\flat}} (\eta_{\natural'_1}^{(m_1 + 1)}, \eta_{\natural_2}^{(m_2)}, \dots, \eta_{\natural_k}^{(m_k)} ;  \eta^{(s)}_{\natural} (f_\lozenge)) \end{split} \]
depending on whether $\flat_0 = \cdot$ or $\flat_0 = *$ (with $\wt{\flat} = (\bar{\flat}_0, \flat_1, \dots,  \flat_{k-1})$) and the second one being 
\[ N^{-(k+1)} \Pi^{(1)}_{\wt{\sharp}, \wt{\flat}} (\eta, \eta_{\natural_1}^{(m_1)}, \dots , \eta_{\natural_k}^{(m_k)}, \eta_{\natural}^{(s)} (f_{\lozenge})) \]
with $\wt{\sharp} = (\bar{\flat}_0, \sharp_1, \dots , \sharp_k)$ and $\wt{\flat} = (\bar{\flat}_0, \flat_1, \dots , \flat_{k})$. As we did in the analysis of (\ref{eq:first-L}) and (\ref{eq:second-L}), one can show that both these terms satisfy all properties (i), (ii), (iii), (v), (vi) (we will discuss the properties (iv) below).  

If instead $B(\eta)$ hits one of the factors $a_{y_r}^{\sharp_r} a_{x_{r+1}}^{\flat_r}$ for an $r=1, \dots , k-1$, the resulting two terms will have the form 
\begin{equation}\label{eq:replaa-Pi1} - \left[ N^{-r} \Pi^{(2)}_{\sharp', \flat'} (\eta^{(m_1)}_{\natural_1}, \dots, \eta_{\natural'_r}^{(m_r+1)}) \right] \left[ N^{-(k-r)} \Pi^{(1)}_{\sharp^{''},\flat^{''}} (\eta_{\natural_{r+1}}^{(m_{r+1})}, \dots, \eta_{\natural_k}^{(m_k)} ; \eta_{\natural}^{(s)} (f_{\lozenge})) \right] 
\end{equation}
and by
\begin{equation}\label{eq:replaa2-Pi1} - \left[ N^{-r} \Pi^{(2)}_{\sharp^{'''}, \flat^{'}} (\eta^{(m_1)}_{\natural_1}, \dots, \eta_{\natural'_r}^{(m_r)}) \right] \left[ N^{-(k-r)} \Pi^{(1)}_{\sharp^{''},\flat^{'''}} (\eta_{\natural^{'}_{r+1}}^{(m_{r+1}+1)}, \dots, \eta_{\natural_k}^{(m_k)} ;  
\eta_{\natural}^{(s)} (f_{\lozenge})) \right] \end{equation}
with $\sharp', \sharp^{''},\sharp^{'''}$ and $\flat',\flat^{''}, \flat^{'''}$ as defined after (\ref{eq:replaa2}). Proceeding similarly as we did in  (\ref{eq:replaa2}), we can show that these terms satisfy (i),(ii),(iii),(v),(vi).

Let us now consider the case that (\ref{eq:Pi1-term}) is commuted with the last pair of operators appearing in (\ref{eq:Pi1-term}). From the induction assumption, we know that this pair can only be $a_{y_k}^* a (\eta^{(2r)} (f))$ or $a_{y_k} a^* (\eta^{(2r+1)} (\bar{f}))$. In the first case, (\ref{eq:Pi1-term}) is replaced by
\begin{equation}\label{eq:restaa-Pi1l} - \Pi^{(2)}_{\sharp, \flat'} (\eta^{(j_1)}_{\natural_1}, \dots, \eta^{(j_k)}_{\natural_k}) \, b^* (\eta^{(2r+1)} (\bar{f})) - \Pi^{(2)}_{\sharp',\flat'} (\eta^{(j_1)}_{\natural_1} , \dots, \eta^{(j_{k-1})}_{\natural_{k-1}}, \eta^{(j_k+1)}_{\natural'_k}) \,  b (\eta^{(2r)} (f)) \end{equation}
In the second case, it is replaced by 
\begin{equation}\label{eq:restaa-Pi1ll} - \Pi^{(2)}_{\sharp', \flat'} (\eta^{(j_1)}_{\natural_1}, \dots, \eta^{(j_{k-1})}_{\natural_{k-1}}, \eta^{(j_k+1)}_{\natural'_k}) b^* (\eta^{(2r+1)} (\bar{f})) - \Pi^{(2)}_{\sharp',\flat'} (\eta^{(j_1)}_{\natural_1} , \dots, \eta^{(j_k)}_{\natural_k}) b (\eta^{(2r+2)} (f)) \end{equation}
In (\ref{eq:restaa-Pi1l}), (\ref{eq:restaa-Pi1ll}), we used the notation $\flat' = (\flat_0, \dots , \flat_{k-1})$, $\sharp' = (\sharp_1, \dots , \bar{\sharp}_k)$ (as usual, the precise form of $\natural'_k$ is not important). {F}rom the expression (\ref{eq:restaa-Pi1l}), (\ref{eq:restaa-Pi1ll}), we see that also in this case, (i), (ii), (iii), (v), (vi) are satisfied.

As for (iv), from the induction assumption we know that there is exactly one term, in the expansion for $\text{ad}_{B(\eta)}^{(n)} (b(f))$, given by (\ref{eq:iv1}) if $n$ is even and by (\ref{eq:iv2}) if $n$ is odd. Let us take, for example, (\ref{eq:iv1}). If we commute the zero-order $\Pi^{(1)}$-operator $b (\eta^{(n)} (f))$ in (\ref{eq:iv1}) with $B(\eta)$, we obtain exactly  the term in (\ref{eq:iv2}), with $n$ replaced by $(n+1)$ (together with a second term, containing a $\Pi^{(1)}$-operator of order one). Similarly, if we take (\ref{eq:iv2}) and we commute the $\Pi^{(1)}$-operator $b^* (\eta^{(n)} (\bar{f})$ with $B(\eta)$, we get (\ref{eq:iv1}), with $n$ replaced by $(n+1)$. Clearly, looking at the terms above, it is clear that there can be only one term with this form. This shows that also in the expansion for $\text{ad}_{B(\eta)}^{(n+1)} (b(f))$, there is exactly one term of the form given in (iv). 

Finally, let us count the number of terms in the expansion for $\text{ad}^{(n+1)}_{B(\eta)} (b(f))$. By the inductive assumption, the expansion for $\operatorname{ad}_{B(\eta)}^{(n)}(b(f))$ contains exactly $2^n n!$ terms. By $(ii)$, each of these terms is a product of exactly $(n+1)$ operators, each of them being either $(N-\cN)$, $(N+1-\cN)$, a field operator $b^{\sharp}_x$ or a quadratic factor $a_y^\sharp a_x^\flat$ commuting with the number of particles operator. By \eqref{2.2.Betacommutators}, the commutator of $B(\eta)$ with each such factor gives a sum of two terms. Therefore, by the product rule, $\operatorname{ad}_{B(\eta)}^{(n+1)}(b(f))$ contains $ 2^n(n!)\times 2(n+1) = 2^{(n+1)} (n+1)!$ summands. This concludes the proof of the lemma.   
\end{proof}

{F}rom Lemma \ref{lm:indu}, we immediately obtain a convergent series expansion for the conjugation of the fields $b(f)$ and $b^*(f)$ with the unitary operator $\exp (B(\eta))$.
\begin{lemma}\label{lm:conv-series}
Let $\eta \in L^2 (\bR^3 \times \bR^3)$ be symmetric, with $\| \eta \|_2$ sufficiently small. Then we have 
\begin{equation}\label{eq:conv-serie}
\begin{split} e^{-B(\eta)} b(f) e^{B (\eta)} &= \sum_{n=0}^\infty \frac{(-1)^n}{n!} \text{ad}_{B(\eta)}^{(n)} (b(f)) \\
e^{-B(\eta)} b^* (f) e^{B (\eta)} &= \sum_{n=0}^\infty \frac{(-1)^n}{n!} \text{ad}_{B(\eta)}^{(n)} (b^* (f)) \end{split} \end{equation}
where the series on the r.h.s. are absolutely convergent. 
\end{lemma}
\begin{proof}
{F}rom (\ref{eq:BCH}) we have \begin{equation}\label{eq:conv+err} \begin{split}
e^{-B(\eta)} b (f) e^{B(\eta)} = &\sum_{n=1}^{m-1} (-1)^n \frac{\text{ad}^{(n)}_{B(\eta)} (b(f))}{n!} \\ &+ \int_0^{1} ds_1 \int_0^{s_1} ds_2 \dots \int_0^{s_{m-1}} ds_m \, e^{-s_m B(\eta)} \text{ad}^{(m)}_{B(\eta)} (b(f)) e^{s_m B(\eta)} \end{split} \end{equation}
To prove (\ref{eq:conv-serie}), we show that the norm of the error term converges to zero, as $m \to \infty$. 
By Lemma (\ref{lm:indu}), $\text{ad}^{(n)}_{B(\eta)} (b(f)$ is given by a sum of $2^n n!$ terms of the form 
	\begin{equation} \label{eq:Lmbds}
	\Lambda_1 \dots \Lambda_i \, \frac{1}{N^{k}} \Pi^{(1)}_{\sharp,\flat} (\eta^{(j_1)}_{\natural_1},\dots, \eta^{(j_k)}_{\natural_k} ; \eta^{(\ell)}(f_{\lozenge}))  
	\end{equation}
with $ i, k, \ell \in\NN$, $j_1,\dots,j_k\in\NN \backslash \{ 0 \}$ and where each $\Lambda_r$ is either $(N-\cN)/N$, $(N+1-\cN)/N$ or an operator of the form
\[ \frac{1}{N^p} \Pi^{(2)}_{\sharp,\flat} (\eta_{\natural_1}^{(m_1)}, \dots , \eta_{\natural_p}^{(m_p)} ) \]
On $\cF^{\leq N}$, we have the bounds $\| (N-\cN)/N \| \leq 1$ and $\| (N+1-\cN)/N \| \leq 2$. Lemma \ref{lm:Pi-bds} implies that
\[ N^{-p} \left\| \Pi^{(2)}_{\sharp,\flat} (\eta_{\natural_1}^{(m_1)}, \dots , \eta_{\natural_p}^{(m_p)} ) \right\| \leq (12)^p (2\| \eta \|_2)^{m_1 + \dots + m_p} \]
and that 
\[ N^{-k} \left\| \Pi^{(1)}_{\sharp,\flat} (\eta^{(j_1)}_{\natural_1},\dots, \eta^{(j_k)}_{\natural_k} ; \eta^{(\ell)}(f_{\lozenge}))   \right\| \leq (12)^k \sqrt{N} \| f \|_2  (2\| \eta \|_2)^{\ell+j_1 + \dots + j_k} \]
Here we used the fact that, if a kernel $\eta^{(j)}$ is associated with a normally ordered pairs of creation and annihilation operators, then $\| \eta^{(j)} \|_\text{HS} \leq \| \eta \|_\text{HS}^j$. If instead $\eta^{(j)}$ is associated with a non-normally ordered pair, then point (vi) in Lemma \ref{lm:indu} implies that $j \geq 2$. Hence,  
\[ \begin{split} \int \left| \eta^{(j)} (x;x) \right| dx & = \int \left| \int \eta (x;y) \eta^{(j-1)} (y;x) dy \right| dx \\ &\leq \left( \int |\eta (x;y)|^2 dxdy\right)^{1/2} \left(\int |\eta^{(j-1)} (x;y)|^2 dx dy\right)^{1/2} 
\\ &\leq \| \eta \|_\text{2} \| \eta^{(j-1)} \|_\text{2} \leq \| \eta \|_\text{2}^j  \end{split} \]
Therefore, if the term (\ref{eq:Lmbds}) contains $\Pi^{(2)}$-operators of order $p_1, \dots, p_j \in \bN \backslash \{ 0 \}$, we can bound
\[ \begin{split} \Big\| \Lambda_1 \dots \dots \Lambda_i \, & \frac{1}{N^{k}} \Pi^{(1)}_{\sharp,\flat} (\eta^{(j_1)}_{\natural_1},\dots, \eta^{(j_k)}_{\natural_k} ; \eta^{(\ell)}(f_{\lozenge}))  \Big\| \\ &\leq 12^{p_1 + \dots + p_j + k} \sqrt{N} (2\| \eta \|_2)^{m} \leq \sqrt{N} \| f \|_2 C^m \| \eta\|^m  \end{split} \]
and therefore, since $\text{ad}^{(m)}_{B(\eta)} (b(f))$ is the sum of $2^m m!$ terms,
\begin{equation}\label{eq:smalleta} \| \text{ad}^{(m)}_{B(\eta)} (b(f)) \| \leq \sqrt{N} \| f \|_2 (2C \| \eta \|_2)^m m!\end{equation}
This proves, first of all, that the series on the r.h.s. of (\ref{eq:conv-serie}) converges absolutely, if $\| \eta \|_2 \leq (4C)^{-1}$. Under this condition, (\ref{eq:smalleta}) also implies that the error term on the r.h.s. of (\ref{eq:conv+err}) converges to zero, as $m \to \infty$, since
\[ \begin{split} &\left\| \int_0^1 ds_1 \dots \int_0^{s_{m-1}} ds_m e^{-s_m B(\eta)} \text{ad}_{B(\eta)} (b(f))  e^{s_m B(\eta)} \right\|  \leq \sqrt{N} \| f \|_2 (2C\| \eta \|)^m \end{split} \]
\end{proof}

\section{Fluctuation Dynamics}
\label{sec:fluc}

In this section, we are going to define the fluctuation dynamics describing the evolution of orthogonal excitations of the Bose-Einstein condensate.

Instead of comparing the solution of the many-body Schr\"odinger equation (\ref{eq:schr0}) directly with the solution of the Gross-Pitaevskii equation (\ref{eq:GPtd}), it is convenient to introduce a modified, $N$-dependent, Gross-Pitaevskii equation. To this end, we fix $\ell > 0$ and we consider the ground state $f_{\ell}$ of the Neumann problem 
\begin{equation}\label{eq:scatl} \left(-\Delta + \frac{1}{2} V \right) f_{\ell} = \lambda_{\ell} f_\ell \end{equation}
on the ball $|x| \leq N\ell$ (we omit the $N$-dependence in the notation for $f_\ell$ and for $\lambda_\ell$; notice that $\lambda_\ell$ scales as $N^{-3}$), with the normalization $f_\ell (x) = 1$ if $|x| = N \ell$.  We extend $f_\ell$ to $\bR^3$ by setting $f_\ell (x) = 1$ for all $|x| > N\ell$. We also define $w_\ell = 1-f_\ell$ (so that $w_\ell (x) = 0$ if $|x| > N \ell$). By scaling, we observe that $f_\ell (N.)$ satisfies the equation 
\begin{equation}\label{eq:scatlN} \left( -\Delta + \frac{N^2}{2} V (N.) \right) f_\ell (N.) = N^2 \lambda_\ell f_\ell (N.) \end{equation}
on the ball $|x| \leq \ell$ ($\ell > 0$ will be kept fixed, independent of $N$). With this choice, we expect that $f_\ell$ will be close, in the limit of large $N$, to the solution of the zero-energy scattering equation (\ref{eq:0en}). This is confirmed by the next lemma, where we collect some important properties of $f_\ell$. Most of the these results are taken from Lemma~A.1 of \cite{ESY0}.

\begin{lemma} \label{3.0.sceqlemma}
Let $V \in L^3 (\bR^3)$ be a non-negative, spherically symmetric potential with $V(x) = 0$ for all $|x| > R$. 
Fix $\ell > 0$ and let $f_\ell$ denote the solution of \eqref{eq:scatl}. 
\begin{enumerate}
\item [i)] We have 
\[ \lambda_\ell = \frac{3a_0}{N^3 \ell^3} \left(1 + \mathcal{O} (a_0 / N\ell) \right) \]
\item[ii)] We have $0\leq f_\ell, w_\ell \leq 1$ and \begin{equation}\label{eq:Vfa0} \int dx \, V(x) f_\ell (x) = 8\pi a_0 + \mathcal{O}(N^{-1}).
\end{equation}    
\item[iii)] There exists a constant $C>0 $, depending on the potential $V$, such that 
	\begin{equation}\label{3.0.scbounds1} 
	w_\ell(x)\leq \frac{C}{|x|+1} \quad\text{ and }\quad |\nabla w_\ell(x)|\leq \frac{C}{|x|^2+1}. 
	\end{equation}
for all $|x| \leq N \ell$.   
\end{enumerate}        
\end{lemma}
\begin{proof}
Statement (i), the fact that $0 \leq f_\ell, w_\ell \leq 1$, and statement (iii) follow from Lemma A.1 in \cite{ESY0}. We have to show (\ref{eq:Vfa0}). To this end, we adapt the proof of Lemma 5.1 (iv) of \cite{ESY2}. With $ r=|x|$, we may write $ m(r) = r f_\ell(r) $. We find that, for all $ r\in (R,N\ell]$,
        \begin{equation}
        m(r) = \lambda_\ell^{-\frac12}\sin(\lambda_\ell^{\frac12}(r-N\ell)) + N\ell\cos (\lambda_\ell^{\frac12}(r-N\ell)).
        \end{equation}
   By expanding up to the order $ \mathcal O (\lambda_\ell^2)$ we obtain
        \begin{equation}
        m(r) = r-a_0+\mathcal{O}(N^{-1}), \quad m'(r) = 1 + \mathcal{O}(N^{-1}).
        \end{equation}
Hence
        \begin{equation}
        \begin{split}
        \int dx\, V(x) f_\ell(x) &= 4\pi \int_0^R dr\, r V(r) m(r) \\ &=8\pi \int_0^R dr\, (rm''(r)+\lambda_\ell r^2 f_\ell(r))\\
        &= 8\pi \int_0^R dr\, rm''(r) + O(N^{-3}) 
        \\
        & = 8\pi (Rm'(R)-m(R)) + \mathcal{O}(N^{-1}) = 8\pi a_0 + \mathcal{O}(N^{-1}).
        \end{split}
        \end{equation}
   \end{proof}
   
Next, we introduce next the modified Gross-Pitaevskii equation 
\begin{equation}\label{eq:GPmod}
i\partial_t \wtph_t = -\Delta \wtph_t + \left( N^3 V(N.) f_\ell (N.) * |\wtph_t|^2 \right) \wtph_t 
\end{equation}
with initial data $\wtph_{t=0} = \ph$ describing the Bose-Einstein condensate at time $t=0$. While in Theorem \ref{thm:main2} the notation $\ph$ is already used to indicate the initial condensate wave function, in the proof of Theorem \ref{thm:main} we will choose $\ph = \phi_\text{GP}$ to be the minimizer of the Gross-Pitaevskii functional (\ref{eq:GPen-functr}). In both cases, we assume that $\ph \in H^4 (\bR^3)$. 

Notice that, in contrast with the initial data $\ph$, the solution $\wtph_t$ depends on $N$. With (\ref{eq:Vfa0}), one can show that $\wtph_t$ converges towards the solution of the original Gross-Pitaevskii equation (\ref{eq:GPtd}), as $N \to \infty$. This fact and some other important properties of the solutions of (\ref{eq:GPtd}) and (\ref{eq:GPmod}) are listed in the next proposition, whose proof can be found in Theorem 3.1 of \cite{BDS}, with the only difference that, in \cite{BDS}, the modified Gross-Pitaevskii equation was defined through the solution $f$ of the zero energy scattering equation, while here we work with the Neumann ground state $f_\ell$. The only relevant consequence is the fact that, here, the integral of $f_\ell$ against $V$ is not exactly equal to $8\pi a_0$; the error, however, is of order $N^{-1}$ by (\ref{eq:Vfa0}). 
\begin{prop}\label{prop:phph}
Let  $V \in L^3 (\bR^3)$ be a non-negative, spherically symmetric, compactly supported potential. 
Let $\varphi\in H^1(\bR)$ with $ \lt{\varphi}=1$. 
	\begin{enumerate}
	\item[i)] Well-Posedness. For any $\ph \in H^1 (\bR^3)$, with $\| \ph \|_2 = 1$, there exist unique global solutions $t \to \ph_t$ and $t \to \wtph_t$ in $C(\bR, H^1(\bR^3))$ of the Gross-Pitaevskii equation (\ref{eq:GPtd}) and, respectively, of the modified Gross-Pitaevskii equation (\ref{eq:GPmod}) with initial datum $ \varphi$. We have $\| \ph_t \|_2 = \| \wtph_t \|_2 = 1$ for all $t \in \bR$. Furthermore, there exists a constant $C>0$ such that
\[ \| \ph_t \|_{H^1} , \| \wtph_t \|_{H^1} \leq C \]
\item[ii)] Propagation of higher regularity. If $ \varphi\in H^m(\bR)$ for some $m \geq 2$, then $\ph_t,\wtph_t \in H^m(\bR)$ for every $t\in\bR$. Moreover, there exist constants $C>0$, depending on $m$ and on $\|\varphi\|_{H^m}$, and $c>0$, depending on $m$ and on $\| \varphi\|_{H^1}$, such that, for all $t\in \bR$,
	    \begin{equation} 
	    \hnn{\pt}{m}, \hnn{\wtph}{m}\leq Ce^{c|t|}. 
	    \end{equation}
	\item[iii)] Regularity of time derivatives. Suppose $\varphi\in H^4(\bR)$. Then there exist $C>0$, depending on $\| \varphi\|_{H^4}$, and $c>0$, depending on $ \| \varphi\|_{H^1}$, such that, for all $t \in \bR$, 
\[	\| \dot{\wtph}\|_{H^2}, \| \ddot{\wtph} \|_{H^2} \leq Ce^{c|t|} . \]
	\item[iv)] Comparison of Dynamics. Suppose $\varphi\in H^2(\bR)$. Then there exists a constant $c>0$, depending on $\| \varphi\|_{H^2}$, such that for all $t\in \bR$,
	    \begin{equation} 
	    \| \ph_t-\wtph_t \|_2 \leq C N^{-1} \exp(c \exp(c |t|). 
	    \end{equation}
	\end{enumerate}
    \end{prop}

To compare the many-body evolution $\psi_{N,t}$ with products of the solution $\wtph_t$ of the modified Gross-Pitaevskii equation (\ref{eq:GPtd}), we are going to define a unitary map (already discussed in Section \ref{sec:intro}, after (\ref{eq:UNdef})) that was first introduced in \cite{LNSS,LNS} in the mean-field setting. To this end, we remark that every $\psi_N \in L^2_s (\bR^{3N})$ has a unique representation of the form
\begin{equation}\label{eq:psi-repr} \psi_N = \sum_{n=0}^N \psi_N^{(n)} \otimes_s \wtph_t^{\otimes (N-n)} 
\end{equation}
where $\psi_N^{(n)} \in L^2_{\perp \wtph_t} (\bR^{3})^{\otimes_s n}$ is symmetric with respect to permutations and orthogonal to $\wtph_t$, in each of its coordinate, and where, for $\psi_N^{(n)} \in L^2_\perp (\bR^3)^{\otimes_s n}$ and $\psi_N^{(k)} \in L^2_\perp (\bR^3)^{\otimes_s k}$, $\psi_N^{(n)} \otimes_s \psi_N^{(k)}$ denotes the symmetrized product defined by 
  \begin{equation}\label{symprod}
    \begin{split}
    \psi_N^{(k)} \otimes_s &\psi_N^{(n)} (x_1, \dots, x_{k+n}) 
    \\ &=\frac{1}{\sqrt{ k! n! (k+n)! }}\sum_{\sigma\in S_{k+n}}  \psi_N^{(k)} (x_{\sigma(1)},  \dots, x_{\sigma_{(k)}})\psi_N^{(n)} (x_{\sigma(k+1)}, \dots, x_{\sigma_{(k+n)}}). 
    \end{split}
    \end{equation}
Using the representation (\ref{eq:psi-repr}), we define $U_{N,t} : L^2_s (\bR^{3N}) \to \cF^{\leq N}_{\perp \wtph_t}$ by setting  
\begin{equation}\label{eq:UNtdef} 
U_{N,t} \psi_N = \{ \psi_N^{(0)}, \psi_N^{(1)}, \dots , \psi_N^{(N)} \}.
\end{equation} 
In terms of creation and annihilation operators, the map $U_{N,t}$ is given by 
\[ U_{N,t} \psi_N = \bigoplus_{n=0}^N (1-|\wtph_t \rangle \langle \wtph_t|)^{\otimes n} \frac{a (\wtph_t)^{N-n}}{\sqrt{(N-n)!}} \psi_N \, .  \]
Here, and frequently in the sequel, we identify $\psi_N \in L^2_s (\bR^{3N})$ with the Fock space vector $\{ 0, \dots , 0, \psi_N, 0 , \dots \} \in \cF$. {F}rom (\ref{eq:psi-repr}) and by the requirement of orthogonality, it is easy to check that $\| \psi_N \|^2  = \sum_{n=0}^N \| \psi^{(n)}_N \|^2$. Hence, $U_{N,t} : L^2_s (\bR^{3N}) \to \cF^{\leq N}_{\perp \wtph_t}$ is a unitary map, with inverse   
\[ U_{N,t}^* \{ \psi_N^{(0)}, \psi_N^{(1)}, \dots, \psi^{(N)}_N \} = \sum_{n=0}^N \frac{a^* (\wtph_t)^{N-n}}{\sqrt{(N-n)!}} \psi_N^{(n)} \]
The action of $U_{N,t}$ on creation and annihilation operators is determined by the following rules (see \cite{LNSS,LNS}): 
 \begin{equation}\label{2.1.UNconjugation}
    \begin{split}
    U_{N,t} a^* (\wtph_t) a (\wtph_t) U_{N,t}^*  & = N-\cN \\
    U_{N,t} a^*(f)a (\wtph_t) U_{N,t}^* &= a^*(f) \sqrt{N-\cN} = \sqrt{N} \, b^* (f) \\
    U_{N,t} a^* (\wtph_t) a(g) U_{N,t}^* &= \sqrt{N-\cN} a(g) = \sqrt{N} \, b (g) \\
    U_{N,t} a^* (f) a(g)U_{N,t}^* &= a^*(f)a(g)
    \end{split}
    \end{equation}
for all $ f,g\in L^2_{\perp \wtph_t} (\bR^3)$. Here we used modified creation and annihilation operators, as defined in (\ref{eq:bb-fie}).

With $U_{N,t}$ we factor out the condensate and we focus on its orthogonal excitations. Observe, however, that $U_{N,t}$ does not remove correlations, which are known to play a crucial role in the Gross-Pitaevskii regime (see, for example, \cite{ESY2,EMS,CH2}). To remove correlations from the excitation vectors, we are going to use a generalized Bogoliubov transformation, as introduced in Section \ref{sec:Bog}. We define  
\begin{equation}\label{eq:ktdef} k_t (x;y) = - N w_\ell (N (x-y)) \wtph_t (x) \wtph_t (y) \end{equation}
{F}rom Lemma \ref{3.0.sceqlemma}, it follows that $k_t \in L^2 (\bR^3 \times \bR^3)$, with $L^2$-norm bounded uniformly in $N$. Hence, $k_t$ is the integral kernel of a Hilbert-Schmidt operator on $L^2 (\bR^3)$, which we denote again with $k_t$. We define a new Hilbert-Schmidt operator setting 
\begin{equation}\label{eq:etat} \eta_t = (1-|\wtph_t \rangle \langle \wtph_t|) \, k_t \, (1-|\bar{\wtph}_t \rangle \langle \bar{\wtph}_t|) 
\end{equation}
Also in this case, we will denote by $\eta_t$ both the Hilbert-Schmidt operator defined in (\ref{eq:etat}) and its integral kernel. Note that $\eta_t \in (q_{\wtph_t} \otimes q_{\wtph_t}) L^2 (\bR^3 \times \bR^3)$, where $q_{\wtph_t} = 1- |\wtph_t \rangle \langle \wtph_t|$. Let us write $\eta_t = k_t + \mu_t$, with the Hilbert-Schmidt operator
\begin{equation}\label{eq:mut} \mu_t = |\wtph_t \rangle \langle \wtph_t| \, k_t \, |\bar{\wtph}_t \rangle \langle \bar{\wtph}_t| - |\wtph_t \rangle \langle \wtph_t| \, k_t - k_t | \bar{\wtph}_t \rangle \langle \bar{\wtph}_t| 
\end{equation}
In the next lemma we collect some important properties of the operators $\eta_t$, $k_t$, $\mu_t$. The proof is a simple generalization of the proof of Lemma 3.3 and Lemma 3.4 in \cite{BDS}; we omit the details.
 \begin{lemma}\label{lm:eta}
    Let $\wtph_t$ be the solution of \eqref{eq:GPmod} with initial datum $ \varphi\in H^4 (\bR)$. Let $w_\ell =1-f_\ell$ with $f_\ell$ the ground state solution of the Neumann problem \eqref{eq:scatl}. Let $k_t, \eta_t,\mu_t$ be defined as in (\ref{eq:ktdef}), (\ref{eq:etat}), (\ref{eq:mut}). Then there exist constants $C,c>0$ depending only on $\| \varphi \|_{H^4}$ (in many cases, these constants actually 
   depend only on lower Sobolev norms of $\ph$) and on $V$ such that the following bounds hold true, for all $t \in \bR$. 
\begin{enumerate}
\item[i)] We have 
\begin{equation}\label{eq:etato0}
\lt{\et} \leq C, \quad  \| \eta_t^{(n)} \|_2 \leq \| \eta_t \|^n_2 \leq C^n \quad \text{ and } \quad  \lim_{\ell \to 0} \, \sup_{t \in \bR} \lt{\et} = 0
\end{equation}	    
and also
\begin{equation*}
\| \nabla_j \eta_t \|_2 \leq C \sqrt{N}, \; \| \nabla_j \mu_t \|_2 \leq C, \; \| \nabla_j \eta_t^{(n)} \|_2 \leq C \| \eta_t \|_2^{n-2}, \; \| \Delta_j \eta_t^{(n)} \|_2 \leq C \| \eta_t \|^{n-2}_2  \end{equation*}
for $j=1,2$ and for all $n \geq 2$. Here $\nabla_1 \eta_t$ and $\nabla_2 \eta_t$ denote the kernels $\nabla_x \eta_t (x;y)$ and $\nabla_y \eta_t (x;y)$ ($\Delta_1 \eta_t$ and $\Delta_2 \eta_t$ are defined similarly). 
Decomposing $\cosh_{\eta_t} = 1 + p_{\eta_t}$ and $\sinh_{\eta_t} = \eta_t + r_{\eta_t}$, we obtain
\begin{equation} \| \text{sinh}_{\eta_t} \|_2 , \| p_{\eta_t} \|_2 , \| r_{\eta_t} \|_2, \| \nabla_j p_{\eta_t} \|_2 , \| \nabla_j r_{\eta_t} \|_2 \leq C \end{equation}
\item[ii)] For a.e. $x,y \in \bR^3$ and $n \in \bN$, $n \geq 2$, we have the pointwise bounds 
\begin{equation}\label{eq:point-eta} \begin{split} 
|\eta_t (x;y)| &\leq \frac{C}{|x-y|+N^{-1}} |\wtph_t (x)| |\wtph_t (y)| \\ |\eta_t^{(n)} (x;y)| &\leq C \| \eta_t \|_2^{n-2} |\wtph_t (x)| |\wtph_t (y)| \\
|\mu_t (x;y)|, |p_{\eta_t} (x;y)| , |r_{\eta_t} (x;y)| &\leq C |\wtph_t (x)| |\wtph_t (y)| \end{split} 
\end{equation}
\item[iii)] We have 
\[ 
\sup_x \int |\eta_t (x;y)|^2 dy , \; \sup_x \int |k_t (x;y)|^2 dy ,\;  \sup_x \int |\mu_t (x;y)|^2 dy \leq C \| \wtph_t \|_{H^2} \leq C e^{c |t|}\]
and 
\[  
\sup_x \int |\eta_t^{(n)} (x;y)|^2 dy
 \leq C \| \eta_t \|_2^{n-2} \| \wtph_t \|_{H^2} \leq C \| \eta_t \|_2^{n-2} e^{c|t|} \]
 for all $n \geq 2$. Therefore
 \[  
\sup_x \int  |p_{\eta_t} (x;y)|^2 dy , \sup_x \int |r_{\eta_t} (x;y)|^2 dy , \sup_x \int |\text{sinh}_{\eta_t} (x;y)|^2 dy  \leq C e^{c|t|} \]
\item[iv)] For $j=1,2$ and $n \geq 2$, we have
\[ \| \partial_t \eta_t \|_2, \| \partial_t^2 \eta_t \|_2 \leq C e^{c|t|} , \quad \| \partial_t \eta_t^{(n)} \|_2 \leq C n e^{c|t|} \| \eta_t \|^{n-1}_2 \] and also \[ \| \partial_t \nabla_j \eta_t \|_2 \leq C \sqrt{N} e^{c|t|}, \quad \| \partial_t \nabla_j \mu_t \|_2 \leq C e^{c|t|}, \quad \| \partial_t \nabla_j \eta^{(n)}_t \|_2 \leq C n \| \eta_t \|^{n-2} e^{c|t|}  \] Therefore
\[ \begin{split} &\| \partial_t p_{\eta_t} \|_2, \| \partial_t r_{\eta_t} \|_2, \| \partial_t \text{sinh}_{\eta_t} \|_2, \| \nabla_j \partial_t p_{\eta_t} \|_2 , \| \nabla_j \partial_t r_{\eta_t} \|_2  \leq C e^{c|t|}  \end{split} \]
\item[v)] For a.e. $x,y \in \bR^3$, we have the pointwise bounds
\[ \begin{split} |\partial_t \eta_t (x;y)| \leq \; &C \left[ 1 + \frac{1}{|x-y|+N^{-1}} \right]  \\ &\hspace{2cm} \times  \left[ |\dot{\wtph}_t (x)|| \wtph_t (y)| + |\wtph_t (x)|| \dot{\wtph}_t (y)| + |\wtph_t (x)| |\wtph_t (y)|  \right]  \end{split} \]
Moreover, for $n \geq 2$, we have
\[ \begin{split} |\partial_t \eta_t^{(n)} (x;y)| \leq \; &C n e^{c|t|} \| \eta_t \|_2^{n-2} \left[ |\dot{\wtph}_t (x)|| \wtph_t (y)| + |\wtph_t (x)|| \dot{\wtph}_t (y)| + |\wtph_t (x)| |\wtph_t (y)| \right]  \end{split} \] 
Therefore 
\[ \begin{split} |\partial_t \mu_t (x;y)|, &|\partial_t r_{\eta_t} (x;y)|, |\partial_t p_{\eta_t} (x;y)|\\ & \leq C e^{c|t|} \left[ |\dot{\wtph}_t (x)| |\wtph_t (y)|  + |\wtph_t (x)| |\dot{\wtph}_t (y)| + |\wtph_t (x)| |\wtph_t (y)| \right] \end{split} \]
\item[vi)] Finally, we find
\[ \sup_x \int |\partial_t \eta_t (x;y)|^2 dy , \; \sup_x \int |\partial_t k_t (x;y)|^2 dy , \; \sup_x \int |\partial \mu_t (x;y)|^2 dy \leq C e^{c|t|} \]
Furthermore, for all $n \geq 2$, 
\[ \sup_x \int |\partial_t \eta^{(n)}_t (x;y)| dy \leq C n e^{c|t|} \| \eta_t \|_2^{n-2} \]
and therefore
\[ \sup_x \int |\partial_t p_{\eta_t} (x;y)|^2 dy , \; \sup_x \int |\partial_t r_{\eta_t} (x;y) |^2 dy , \; \sup_x \int |\partial_t \sinh_{\eta_t} (x;y)|^2 dy \leq C e^{c|t|} \]
\end{enumerate}
\end{lemma}

We model correlations in the solution $\psi_{N,t}$ of the many-body Schr\"odinger equation (\ref{eq:schr0}) by means of the generalized Bogoliubov transformation $\exp (B(\eta_t)) : \cF_{\perp \wtph_t}^{\leq N} \to \cF_{\perp \wtph_t}^{\leq N}$ with the integral kernel $\eta_t \in (q_{\wtph_t} \otimes q_{{\wtph}_t}) L^2 (\bR^3 \times \bR^3)$ defined in (\ref{eq:etat}). We define therefore the fluctuation dynamics 
\begin{equation}\label{eq:WNt} \cW_{N,t} = e^{-B(\eta_t)} \, U_{N,t} \, e^{-iH_N t} \, U^*_{N,0} \, e^{B (\eta_0)} 
\end{equation} 
Then $\cW_{N,t} : \cF_{\perp \ph}^{\leq N} \to \cF_{\perp \wtph_t}^{\leq N}$ is a unitary operator. Clearly, $\cW_{N,t}$ depends on the length parameter $\ell$ (the radius of the ball in (\ref{eq:scatl})), through the modified Gross-Pitaevskii equation (\ref{eq:GPmod}) and also through the kernel $\eta_t$ defined in (\ref{eq:ktdef}), (\ref{eq:etat}). While $\cW_{N,t}$ is well-defined for any value of $\ell > 0$, we will have to choose $\ell > 0$ small, to make sure that $\| \eta_t \|_2$ is sufficiently small; this will allow us to expand the action of the generalized Bogoliubov transformation $\exp (B(\eta_t))$ appearing in (\ref{eq:WNt}) using the series expansion (\ref{eq:conv-serie}) (because, by (\ref{eq:etato0}), smallness of $\ell$ implies that $\| \eta_t \|_2$ is small, uniformly in $t$). 

For $\xi \in \cF_{\perp \ph}^{\leq N}$, the operator $\cW_{N,t}$ is defined so that
\[ e^{-iH_N t} \, U_{N,0}^* \, e^{B(\eta_0)} \xi = U_{N,t}^* \, e^{B(\eta_t)} \, \cW_{N,t} \xi \,. \] 
It allows us to describe the many-body evolution of initial data of the form \begin{equation}\label{eq:psiN0} \psi_N = U_{N,0}^* e^{B(\eta_0)} \xi,\end{equation} and to express the evolved state again in the form \begin{equation}\label{eq:psiNt0} \psi_{N,t} = e^{-iH_N t} \psi_N = U_{N,t}^* e^{B(\eta_t)} \xi_t,\end{equation} where $\xi_t = \cW_{N,t} \, \xi$. As we will see below, a vector of the form (\ref{eq:psiN0}) exhibits Bose-Einstein condensation in the one-particle state $\ph$ if and only if the expectation of the number of particles operator $\langle \xi, \cN \xi \rangle$ is small, compared with the total number of particles $N$. Hence, to prove Theorems \ref{thm:main} and \ref{thm:main2}, we will have to show first that every initial $\psi_N \in L^2_s (\bR^{3N})$ satisfying (\ref{eq:assN}) can be written in the form (\ref{eq:psiN0}) for a $\xi \in \cF_{\perp \ph}^{\leq N}$ with $\langle \xi, \cN \xi \rangle \ll N$ and then that the bound on the expectation of the number of particles is approximately preserved by $\cW_{N,t}$. In fact, it turns out that to control the growth of the expectation of $\cN$ along the fluctuation dynamics, it is not enough to have a bound on $\langle \xi, \cN \xi\rangle$; instead, we will also need a bound on the energy of $\xi$ (this is why we need to assume $b_N \to 0$, in (\ref{eq:assN})). 

To control the growth of the number of particles with respect to the fluctuation dynamics it is important to compute the generator of $\cW_{N,t}$. A simple computation shows that 
\[ i\partial_t \cW_{N,t} = \cG_{N,t} \cW_{N,t} \]
with the time-dependent generator
\begin{equation}\label{eq:GNt} \cG_{N,t} = (i\partial_t e^{-B(\eta_t)}) e^{B(\eta_t)} + e^{-B(\eta_t)} \left[ (i\partial_t U_{N,t}) U^*_{N,t} + U_{N,t} H_N U_{N,t}^* \right] e^{B(\eta_t)} \end{equation}
Notice, that $\cG_{N,t}$ maps $\cF_{\perp \wtph_t}^{\leq N}$ into $\cF^{\leq N}$, but not into $\cF^{\leq N}_{\perp \wtph_t}$. This is due to the fact that the space $\cF_{\perp \wtph_t}^{\leq N}$ depends on time (and thus $\cG_{N,t}$ must have a component which allows $\cW_{N,t}$ to move to different spaces). We will mostly be interested in the expectation of $\cG_{N,t}$ for states in $\cF_{\perp \wtph_t}^{\leq N}$, but at some point (when we will consider the variation of the expectation of $\cG_{N,t}$) it will be important to remember the component of $\cG_{N,t}$ mapping out of $\cF_{\perp \wtph_t}^{\leq N}$.

In the next proposition, we collect important properties of the generator $\cG_{N,t}$.
\begin{theorem}\label{thm:gene}
Let $V \in L^3 (\bR^3)$ be non-negative, spherically symmetric and compactly supported. Let $\cW_{N,t}$ be defined as in (\ref{eq:WNt}) with the length parameter $\ell > 0$ sufficiently small and using the solution of the modified Gross-Pitaevskii equation (\ref{eq:GPmod}), with an initial data $\ph \in H^4 (\bR^3)$. Let
	\begin{equation}\label{eq:CNt}
	\begin{split}
	C_{N,t} =& \; \frac{1}{2} \left\langle \wtph_t, \left( [N^3V(N.) (N-1-2N f_{\ell} (N.))] * |\wtph_t|^2 \right)\wtph_t \right\rangle \\
	& + \int dxdy\, \left|\nabla_x k_t (x;y) \right|^2+ \frac{1}{2}\int dx dy\, N^2 V(N(x-y)) |k_t (x;y)|^2 \\
	&+ \operatorname{Re} \int dxdy\,  N^3V(N(x-y))\bar{\wtph_t} (x)\bar{\wtph_t} (y)k_t (x;y).
	\end{split}
	\end{equation}
Then there exist constants $C, c >0$ such that, in the 
sense of quadratic forms on $\cF_{\perp \wtph_t}^{\leq N}$,	  
\begin{equation}\label{eq:gene-bds}
\begin{split}
\frac{1}{2} \cH_N - C e^{c|t|} (\cN+1) \leq (\Gnt - C_{N,t}) &\leq 2\cH_N + C e^{c|t|} (\cN + 1) \\
\pm i\left[\np,\Gnt\right] &\leq \HN + Ce^{c|t|}(\np+1),\\	     \pm \delt (\Gnt-C_{N,t}) &\leq \HN + Ce^{c|t|}(\np+1),\\	      
\pm \re[ a^* (\delt \wtph_t) a(\wtph_t),\Gnt] &\leq \HN + Ce^{c|t|}(\np+1). 
\end{split}
\end{equation}
where $\cH_N$ is the Fock space Hamiltonian 
\begin{equation}\label{eq:HN-thm} \cH_N = \int dx \, \nabla_x a_x^* \nabla_x a_x + \frac{1}{2} \int dx dy \, N^2 V(N (x-y)) a_x^* a_y^* a_y a_x \end{equation}
Note that, on $\cF_{\perp \wtph_t}^{\leq N}$, we have $[a^* (\delt\wtph_t)a(\wtph_t),\Gnt ] = a^* (\delt\wtph_t)a(\wtph_t)\Gnt$.
\end{theorem}
The proof of Theorem \ref{thm:gene} is given in the next section. From the technical point of view, it represents the main part of our paper. In Section \ref{sec:main}, we show then how to use the properties of $\cG_{N,t}$ established in Theorem \ref{thm:gene} to complete the proof of Theorems \ref{thm:main} and \ref{thm:main2}.

\section{Analysis of the Generator of Fluctuation Dynamics}
\label{sec:gene}

In this section we study the properties of the generator 
\begin{equation}\label{eq:GNt2} \cG_{N,t} = (i\partial_t e^{-B(\eta_t)}) e^{B(\eta_t)} + e^{-B(\eta_t)} \left[ (i\partial_t U_{N,t}) U^*_{N,t} + U_{N,t} H_N U_{N,t}^* \right] e^{B(\eta_t)} \end{equation}   of the fluctuation dynamics (\ref{eq:WNt}); the goal is to prove Theorem \ref{thm:gene}.

As forms on $\cF_{\perp \wtph_t}^{\leq N} \times \cF_{\perp \wtph_t}^{\leq N}$, we find (see Lemma 6 in \cite{LNS}) 
\begin{equation}\label{eq:delUU} \begin{split} (i\partial_t U_{N,t}) U^*_{N,t} = \; &- \langle i\partial_t \wtph_t, \wtph_t \rangle (N-\cN) - \sqrt{N} \left[ b(i\partial_t \wtph_t) + b^* (i\partial_t \wtph_t) \right]  \end{split} \end{equation}
Using (\ref{2.1.UNconjugation}) to compute 
$U_{N,t} H_N U_{N,t}^*$ a lengthy but straightforward computation (see Appendix B of \cite{LNS}) shows that  
\[ (i\partial_t U_{N,t} ) U^*_{N,t} + U_{N,t} H_N U_{N,t}^* = \sum_{j=0}^4 \cL_{N,t}^{(j)} \]
where 
\begin{equation}\label{eq:L0-4} 
\begin{split} \cL^{(0)}_{N,t} = & \frac 12 \product{\wtph_t}{ [N^3V(N.) (1-2 f_\ell(N.)) \ast|\wtph_t|^2 ]\wtph_t} (N-\np)\\
    &-\frac12\product{\wtph_t}{ [N^3V(N.)\ast|\wtph_t|^2]\wtph_t } (\np+1) \frac{(N-\np)}{N}  \\
    \cL^{(1)}_{N,t} = &\sqrt{N} \,b ([N^3V(N.)w_\ell(N.)\ast|\wtph_t|^2] \wtph_t) - \frac{\cN+1}{\sqrt{N}} b([N^3V(N.)\ast|\wtph_t|^2] \wtph_t) +h.c. \\
    \cL^{(2)}_{N,t} = &\int dx\, \nabla_x a_x^* \nabla_x a_x \\
&+\int dx dy  \, N^3 V(N(x-y)) |\wtph_t (y)|^2 \left( b_x^* b_x - \frac{1}{N} a_x^* a_x \right)  \\
&+\int dx dy  \, N^3 V(N(x-y)) \wtph_t (x) \bar{\wtph}_t (y) \left( b_x^* b_y - \frac{1}{N} a_x^* a_y \right)  \\
    & +   \frac12 \left[\int dx dy\, N^3V(N(x-y)) \wtph_t (x)\wtph_t(y)b_x^*b_y^*+h.c. \right]   \\
    \cL^{(3)}_{N,t}  = &\int dx dy \, N^{5/2} V(N(x-y)) \wtph_t (y) b_x^* a_y^* a_x + \text{h.c.} \\
    \cL^{(4)}_{N,t} = &\frac1{2}\int dx dy\, N^2V(N(x-y)) a_x^* a_y^* a_y a_x   
\end{split} 
\end{equation}
The generator (\ref{eq:GNt2}) of the fluctuation dynamics is therefore given by
\[ \cG_{N,t}  = (i\partial_t e^{-B(\eta_t)}) e^{B(\eta_t)} + \sum_{j=0}^4 e^{-B(\eta_t)} \cL_{N,t}^{(j)} e^{B(\eta_t)} \]
In the next subsections, we will study separately the six terms contributing to $\cG_{N,t}$. Before doing so, however, we collect some preliminary results, which will be useful for our analysis. 

\medskip

\emph{Notation and Conventions.} For the rest of this section we employ the short-hand notation $\eta_x$ $k_x$, $\mu_x$ for the wave functions $\eta_x (y) = \eta_t (x;y)$, $k_x (y) = k_t (x;y)$ and $\mu_x (y) = \mu_t (x;y)$. We will always assume that $\sup_{t \in \bR} \| \eta_t \|_2$ is sufficiently small, so that we can use the expansions obtained in Lemma \ref{lm:conv-series}. Finally, by $C$ and $c$ we denote generic constants which only depend on fixed parameters, but not on $N$ or $t$, and which may vary from one line to the next. 

\subsection{Preliminary results}

In this subsection we show some simple but important auxiliary results 
which will be used throughout the rest of Section \ref{sec:gene}. Recall the operators 
\[ \begin{split} \Pi^{(2)}_{\sharp,\flat} (j_1, \dots , j_n) &= \int b_{x_1}^{\flat_0} \prod_{i=1}^{n-1} a^{\sharp_i}_{y_i} a^{\flat_i}_{x_{i+1}} b_{y_n}^{\sharp_n} \prod_{i=1}^n j_i (x_i ; y_i) dx_i dy_i \\
\Pi^{(1)}_{\sharp,\flat} (j_1, \dots , j_n; f) &= \int b_{x_1}^{\flat_0} \prod_{i=1}^{n-1} a^{\sharp_i}_{y_i} a^{\flat_i}_{x_{i+1}} 
a_{y_n}^{\sharp_n} a^{\flat_n} (f) \prod_{i=1}^n j_i (x_i ;y_i) dx_i dy_i 
\end{split} \]
introduced in Section \ref{sec:fock}. For each $i \in \{ 1, \dots , n\}$, we recall in particular the condition that either $\sharp_i = *$ and $\flat_i = \cdot$ or $\sharp_i = \cdot$ and $\flat_i = *$. 

In the next lemma, we consider commutators of these operators with the number of particles operator $\cN$ and with operators of the form $a^* (g_1) a (g_2)$. 

\begin{lemma}\label{lm:prel1}
Let $n\in\bN$, $f,g_1,g_2 \in L^2 (\bR^3)$, $j_1,\dots, j_{n}\in L^2 (\bR^3 \times \bR^3)$.
\begin{itemize}
\item[i)] We have 
\[ \begin{split} 
\left[ \cN, \Pi^{(2)}_{\sharp,\flat} (j_1, \dots, j_n) \right] &= \kappa_{\flat_0, \sharp_n} \Pi^{(2)}_{\sharp, \flat} (j_1, \dots, j_n) 
 \qquad \text{for all } \quad \sharp,\flat \in \{ \cdot , * \}^n \\
\left[ \cN, \Pi^{(1)}_{\sharp,\flat} (j_1, \dots, j_n ; f) \right] &= \nu_{\flat_0} \Pi^{(1)}_{\sharp, \flat} (j_1, \dots, j_n ; f) \quad \text{for all} \quad \sharp \in \{ \cdot , * \}^n, \flat \in \{ \cdot, * \}^{n+1} \, . \end{split} \]
Here $\kappa_{\flat_0, \sharp_n} = 2$, if $\flat_0 = \sharp_n = *$, $\kappa_{\flat_0, \sharp_n} = -2$ if $\flat_0 = \sharp_n = \cdot$, and $\kappa_{\flat_0, \sharp_n} = 0$ otherwise, while $\nu_{\flat_0} = 1$ if $\flat_0 = *$ and $\nu_{\flat_0} = -1$ if $\flat_0 = \cdot$. 
\item[ii)] The commutator 
\[ \left[ a^* (g_1) a(g_2), \Pi^{(2)}_{\sharp,\flat} (j_1, \dots , j_n) \right] \] 
can be written as the sum of $2n$ terms, all having the form
\[ \Pi^{(2)}_{\sharp,\flat} (j_1, \dots , j_{i-1}, h_i , j_{i+1},\dots , j_n) \]
for some $i \in \{1,\dots , n\}$. Here $h_i \in L^2 (\bR^3 \times \bR^3)$ has (up to a possible sign) one of the following forms:
\begin{equation}\label{eq:hell-poss}  h_i (x;y) = g_1 (x)  j_i (\bar{g}_2) (y), \qquad h_i (x;y) = g_1 (y) j_i (\bar{g}_2) (y) \end{equation}
or the same, but with $g_1$ and $\bar{g}_2$ exchanged. Here $j_i (g) (x) = \int j_i (x;z) g(z) dz$. Notice that 
\begin{equation}\label{eq:hkbds-l2}  \| h_i \|_2 \leq \| g_1 \|_2 \| g_2 \|_2 \| j_i \|_2 
\end{equation}
and
\begin{equation}\label{eq:hkbds-pw} 
\begin{split} |h_i (x;y)| &\leq \max \Big\{ |g_1 (x)| \| j_i (.;y) \|_2 \| g_2 \|_2 , |g_1 (y)| \| j_i (x;.) \|_2 \| g_2 \|_2 ,  \\ &\hspace{3cm}  |g_2 (x)| \| j_i (.;y) \|_2 \| g_1 \|_2 , |g_2 (y)| \| j_i (x;.) \|_2 \| g_1 \|_2 \Big\} \end{split} \end{equation}
\item[iii)] The commutator 
\begin{equation}\label{eq:commPi1} \left[ a^* (g_1) a(g_2), \Pi^{(1)}_{\sharp,\flat} (j_1, \dots , j_n;f) \right] 
\end{equation}
can be written as the sum of $2n+1$ terms. $2n$ of them have the form
\[ \Pi^{(1)}_{\sharp,\flat} (j_1, \dots , j_{i-1} , h_i, j_{i+1} , \dots , j_n ; f) \]
where $h_i$ is (up to a possible sign) one of the kernels appearing in (\ref{eq:hell-poss}) (or the same with $g_1$ and $\bar{g}_2$ exchanged), and satisfying the bounds in (\ref{eq:hkbds-l2}), (\ref{eq:hkbds-pw}). The remaining term in the expansion for (\ref{eq:commPi1}) has the form
\begin{equation}\label{eq:Pi1-repl} \Pi^{(1)}_{\sharp,\flat} (j_1, \dots , j_n ; k) \end{equation}
where $k \in L^2 (\bR^3)$ is (up to a possible sign) one of the functions
\begin{equation}\label{eq:hellk-poss} k(x) = \langle g_1, f \rangle \, g_2 (x) , \quad k(x) = \langle g_2, f \rangle \, g_1 (x) 
\end{equation} or one of their complex conjugated functions. In any event, we have 
\[ \| k \|_2 \leq \| g_1 \|_2 \| g_2 \|_2 \| f \|_2 \]
and 
\[ |k(x)| \leq \| f \|_2 \max \{ \| g_1 \|_2 |g_2 (x)|, \| g_2 \|_2 |g_1 (x)| \} \]
\item[iv)] If $f \in L^2 (\bR^3)$ and/or $j_1, \dots, j_n \in L^2 (\bR^3 \times \bR^3)$ depend on time $t \in \bR$, we have
\[\begin{split} \partial_t \Pi^{(2)}_{\sharp,\flat} (j_1, \dots , j_n) = \; &\sum_{i=1}^n \Pi^{(2)}_{\sharp,\flat} (j_1, \dots , j_{i-1}, \partial_t j_i , j_{i+1}, \dots , j_n) \\ 
\partial_t \Pi^{(1)}_{\sharp,\flat} (j_1, \dots , j_n;f) = \; &\Pi^{(1)}_{\sharp,\flat} (j_1, \dots , j_n ; \partial_t f)  \\ &+ \sum_{i=1}^n \Pi^{(1)}_{\sharp,\flat} (j_1, \dots , j_{i-1}, \partial_t j_i , j_{i+1}, \dots , j_n ; f) . \end{split} \]
\end{itemize}
\end{lemma}

\begin{proof}
Part (i) follows from $(\cN+1) b_x = b_x \cN$ and $\cN b_x^* = b_x^* (\cN+1)$. Part (iv) follows easily from the Leibniz rule. To prove part (ii), we apply Leibniz rule:
\begin{equation}\label{eq:commaaPi2} \begin{split}
\Big[ a^* (g_1) a(g_2) , &\Pi^{(2)}_{\sharp,\flat} (j_1, \dots , j_n) \Big] \\ = \; &\int [a^* (g_1) a(g_2) , b^{\flat_0}_{x_1}] \prod_{i=1}^n a_{y_i}^{\sharp_i} a_{x_{i+1}}^{\flat_i} b_{y_n}^{\sharp_n} \prod_{i=1}^n j_i (x_i; y_i) dx_i dy_i \\ &+ \sum_{m=1}^{n-1} \int b^{\flat_0}_{x_1} \prod_{i=1}^{m-1} a_{y_i}^{\sharp_i} a_{x_{i+1}}^{\flat_i} \left[ a^* (g_1) a(g_2), a_{y_m}^{\sharp_m} a_{x_{m+1}}^{\flat_m} \right] \\ &\hspace{3cm} \times  \prod_{i=m+1}^{n-1} a_{y_i}^{\sharp_i} a_{x_{i+1}}^{\flat_i} b_{y_n}^{\sharp_n} \prod_{i=1}^n j_i (x_i ; y_i) dx_i dy_i \\ &+  \int  b^{\flat_0}_{x_1} \prod_{i=1}^n a_{y_i}^{\sharp_i} a_{x_{i+1}}^{\flat_i} [a^* (g_1) a(g_2) , b_{y_n}^{\sharp_n} ] \, \prod_{i=1}^n j_i (x_i; y_i) dx_i dy_i \end{split} \end{equation}
Using the commutation relations 
 \begin{equation}\label{5.1.eq2.43}
    \begin{split}
    [a^*(g_1)a(g_2), b_x ] &= -g_1(x) b(g_2) ,\\ [a^*(g_1)a(g_2), b^*_x ] &=\bar  g_2(x)b^*(g_1) \\
    [a^*(g_1)a(g_2), a^*_x a_y ] &=[a^*(g_1)a(g_2), a_y a^*_x] = \bar g_2(x) a^*(g_1)a_y - g_1(y) a^*_xa(g_2) 
    \end{split}
    \end{equation}
we conclude that on the r.h.s. of (\ref{eq:commaaPi2}) we have $2n$ terms, each of them being a $\Pi^{(2)}$-operator (with the same indices $\sharp, \flat$ as the $\Pi^{(2)}$ operator on the l.h.s. of (\ref{eq:commaaPi2})). Furthermore, from (\ref{5.1.eq2.43}) it is clear that for each $\Pi^{(2)}$ operator on the r.h.s. of (\ref{eq:commaaPi2}), only one $j$-kernel will differ from the $j$-kernels of the $\Pi^{(2)}$ operator on the l.h.s. of (\ref{eq:commaaPi2}). In the first term on the r.h.s. of (\ref{eq:commaaPi2}), we only have to replace the $j_1$ kernel (either with  $g_1 (x_1) j_1 (\bar{g}_2) (y_1)$ or with $\bar{g}_2 (x_1) j_1 (g_1) (y_1)$, depending on $\flat_0 \in \{ \cdot , * \}$). Similarly, in the last term on the r.h.s. of (\ref{eq:commaaPi2}), only the $j_n$ kernel has to be changed. In the $m$-th term in the sum, on the other hand, the commutator leads to the sum of two $\Pi^{(2)}$-operators, one where the kernel $j_m$ is changed and one where the kernel $j_{m+1}$ is replaced. {F}rom (\ref{5.1.eq2.43}), it is easy to check that the new kernel can only have one of the forms listed in (\ref{eq:hell-poss}). The bounds (\ref{eq:hkbds-l2}), (\ref{eq:hkbds-pw}) follow easily from the explicit formula in (\ref{eq:hell-poss}). Part (iii) can be shown similarly; the only difference is that, in this case, the commutator can hit the last pair $a_{y_n}^{\sharp_n} a^{\flat_n} (f)$ instead of the $b^{\sharp_n}_{y_n}$ appearing in the $\Pi^{(2)}$-operator. 
\end{proof}

It follows from Lemma \ref{lm:prel1} that 
\begin{equation}\label{eq:cv-ser}
\begin{split} 
 [\cN , e^{-B(\eta)} b(f) e^{B(\eta)}] &= \sum_{n=0}^\infty \frac{(-1)^n}{n!} \, [ \cN , \text{ad}_{B(\eta)}^{(n)} (b(f)) ] \\
[a^* (g_1) a(g_2), e^{-B(\eta)} b(f) e^{B(\eta)}] &= \sum_{n=0}^\infty \frac{(-1)^n}{n!} \, [ a^* (g_1) a(g_2) , \text{ad}_{B(\eta)}^{(n)} (b(f)) ]  \\
\partial_t (e^{-B(\eta)} b(f) e^{B(\eta)}) &= \sum_{n=0}^\infty \frac{(-1)^n}{n!} \partial_t \text{ad}_{B(\eta)}^{(n)} (b(f))
\end{split} 
\end{equation}
where the series on the r.h.s. are absolutely convergent. 

In the next subsections we are going to study what happens to the operators $\cL^{(j)}_{N,t}$ defined in (\ref{eq:L0-4}), when they are conjugated with the generalized Bogoliubov transformation $e^{B(\eta_t)}$. The general strategy is to expand $e^{-B(\eta_t)} \cL_{N,t}^{(j)} e^{B(\eta_t)}$ using Lemma \ref{eq:conv-serie}, and then use Lemma \ref{lm:indu} to express every nested commutator. Therefore, we will have to bound expectations of operators of the form 
\[ \Lambda_1 \dots \Lambda_i N^{-k} \Pi^{(1)}_{\sharp, \flat} (\eta_{t,\natural_1}^{(j_1)}, \dots , \eta_{t,\natural_k}^{(j_k)} ; \eta^{(s)} (g)) \]
or of products of such operators. To this end, the next lemma will be frequently used.
\begin{lemma}\label{lm:prelim}
Let $g \in L^2 (\bR^3)$, $i_1,i_2,k_1,k_2,\ell_1,\ell_2 \in \bN$ and $j_1, \dots , j_{k_1}, m_1, \dots , m_{k_2} \in \bN \backslash \{ 0 \}$. Suppose that, for $s=1, \dots , i_1$, $s' = 1, \dots , i_2$, $\Lambda_s$, $\Lambda'_{s'}$ is either a factor $(N-\cN)/N$, a factor $(N-\cN+1)/N$ or an operator of the form 
\begin{equation}\label{eq:Pi2-prel} N^{-p} \, \Pi^{(2)}_{\sharp, \flat} (\eta^{(q_1)}_{t,\natural_1} , \dots , \eta_{t, \natural_p}^{(q_p)} ) \end{equation}
\begin{itemize}
\item[i)] Assume that the operator \[ \Lambda_1 \dots \Lambda_{i_1} N^{-k_1} \Pi^{(1)}_{\sharp,\flat} (\eta_{t,\natural_1}^{(j_1)}, \dots , \eta_{t,\natural_k}^{(j_k)} ; \eta^{(\ell_1)}_{t,\lozenge} (g))\]
appears in the expansion of $\text{ad}^{(n)}_{B(\eta_t)} (b(g))$ for some $n \in \bN$ (as discussed in Lemma \ref{lm:indu}). Then
\[ \left\| (\cN+1)^{-1/2} \Lambda_1 \dots  \Lambda_i N^{-k} \Pi^{(1)}_{\sharp,\flat} (\eta_{t,\natural_1}^{(j_1)}, \dots , \eta_{t,\natural_k}^{(j_k)} ; \eta^{(\ell_1)}_{t,\lozenge} (g)) \xi \right\| \leq C^n \| \eta_t \|^n \| g \| \| \xi \| \]
If moreover, at least one of the $\Lambda_s$ operators has the form (\ref{eq:Pi2-prel}) or if $k \geq 1$, we also have
\begin{equation}\label{eq:prel-i2} \begin{split} &\left\| (\cN+1)^{-1/2} \Lambda_1 \dots  \Lambda_i N^{-k} \Pi^{(1)}_{\sharp,\flat} (\eta_{t,\natural_1}^{(j_1)}, \dots , \eta_{t,\natural_k}^{(j_k)} ; \eta^{(\ell_1)}_{t,\lozenge} (g)) \xi \right\| \\ &\hspace{6cm} \leq C^n N^{-1/2} \| \eta_t \|^n \| g \| \| (\cN+1)^{1/2} \xi \| \end{split} \end{equation} 
\item[ii)] Let $r: L^2 (\bR^3) \to L^2 (\bR^3)$ be a bounded linear operator. We use the notation $(\eta^{(s)} r)_x (y) := (\eta^{(s)} r) (x;y)$ (if $s=0$, $(\eta^{(s)} r)_x (y) = r_x (y) = r(x;y)$, as a distribution). Assume that the operator 
\[ \Lambda_1 \dots \Lambda_{i_1} N^{-k_1} \Pi^{(1)}_{\sharp,\flat} (\eta_{t,\natural_1}^{(j_1)}, \dots , \eta_{t,\natural_{k_1}}^{(j_{k_1})} ; (\eta^{(\ell_1)}_{t,\lozenge} r)_x) \]
appears in the expansion of $\text{ad}^{(n)}_{B(\eta_t)} (b (r_x))$ for some $n \in \bN$. Then
\begin{equation}\label{eq:prelim-ii} \begin{split}  &\left\| \Lambda_1 \dots  \Lambda_i N^{-k} \Pi^{(1)}_{\sharp,\flat} (\eta_{t,\natural_1}^{(j_1)}, \dots , \eta_{t,\natural_k}^{(j_k)} ; (\eta^{(\ell_1)}_{t,\lozenge} r)_x ) \xi \right\| \\ &\hspace{3cm} \leq \left\{ \begin{array}{ll} C^n \| \eta_t \|^{n-1} \| (\eta_t r)_x \| \| (\cN+1)^{1/2} \xi \| \qquad &\text{if } \ell_1 \geq 1 \\ C^n \| \eta_t \|^n \| a (r_x) \xi \| \qquad &\text{if } \ell_1 = 0 \end{array} \right. \end{split} \end{equation}
\item[iii)] Suppose that the operators 
\[ \begin{split} &\Lambda_1 \dots \Lambda_{i_1} N^{-k_1} \Pi^{(1)}_{\sharp,\flat} (\eta_{t,\natural_1}^{(j_1)}, \dots , \eta_{t,\natural_{k_1}}^{(j_{k_1})} ; (\eta^{(\ell_1+1)}_{t,\lozenge} r)_x), \\  &\Lambda'_1 \dots \Lambda'_{i_2} N^{-k_2} \Pi^{(1)}_{\sharp,\flat} (\eta_{t,\natural'_1}^{(m_1)}, \dots , \eta_{t,\natural'_{k_2}}^{(m_{k_2})} ; \eta^{(\ell_2)}_{x,\lozenge'})
\end{split} \]
appear in the expansion of $\text{ad}^{(n)}_{B(\eta_t)} (b ((\eta_t r)_x))$ and respectively of $\text{ad}^{(k)}_{B(\eta_t)} (b_x)$ for some $n,k \in \bN$. Then 
\begin{equation}\label{eq:prelim-iii} \begin{split} 
&\Big\| (\cN+1)^{-1/2} \Lambda_1 \dots  \Lambda_{i_1} N^{-{k_1}} \Pi^{(1)}_{\sharp, \flat} (\eta^{(j_1)}_{t,\natural_1}, \dots , \eta^{(j_{k_1})}_{t,\natural_{k_1}} ; (\eta^{(\ell_1+1)}_{t,\lozenge} r)_x ) \\ &\hspace{1cm} \times \Lambda'_1 \dots \Lambda'_{i_2} N^{-k_2} \Pi^{(1)}_{\sharp', \flat'} (\eta^{(m_1)}_{t,\natural'_1} , \dots \eta^{(m_{k_2})}_{t,\natural'_{k_2}} ; \eta^{(\ell_2)}_{x,\lozenge'} ) \xi \Big\| \\ &\hspace{1.5cm} \leq \left\{ \begin{array}{ll} C^{n+k} \| \eta_t \|^{n+k-1}  \| (\eta_t r)_x \| \| \eta_x \| \| (\cN+1)^{1/2} \xi \| \quad &\text{if } \ell_2 > 0 
\\ C^{n+k} \| \eta_t \|^{n+k} \| (\eta_t r)_x \| \| a_x \xi \| \quad &\text{if } \ell_2 = 0 \end{array} \right. 
\end{split} \end{equation}
Similarly, if the operators 
\[ \begin{split} &\Lambda_1 \dots \Lambda_{i_1} N^{-k_1} \Pi^{(1)}_{\sharp,\flat} (\eta_{t,\natural_1}^{(j_1)}, \dots , \eta_{t,\natural_{k_1}}^{(j_{k_1})} ;  (\eta^{(\ell_1)}_{t,\lozenge} \partial_t \eta_{t,\wt{\lozenge}})_x ), \\  &\Lambda'_1 \dots \Lambda'_{i_2} N^{-k_2} \Pi^{(1)}_{\sharp,\flat} (\eta_{t,\natural'_1}^{(m_1)}, \dots , \eta_{t,\natural'_{k_2}}^{(m_{k_2})} ; \eta^{(\ell_2)}_{x,\lozenge'})
\end{split} \]
appear in the expansion of $\text{ad}^{(n)}_{B(\eta_t)} (b ( \partial_t \eta_t))$ and respectively of $\text{ad}^{(k)}_{B(\eta_t)} (b_x)$ for some $n,k \in \bN$, we have
\begin{equation}\label{eq:prelim-iiib} \begin{split} 
&\Big\| (\cN+1)^{-1/2} \Lambda_1 \dots  \Lambda_{i_1} N^{-{k_1}} \Pi^{(1)}_{\sharp, \flat} (\eta^{(j_1)}_{t,\natural_1}, \dots , \eta^{(j_{k_1})}_{t,\natural_{k_1}} ; (\eta^{(\ell_1)}_{t,\lozenge} \partial_t \eta_{t,\wt{\lozenge}})_x) \\ &\hspace{1cm} \times \Lambda'_1 \dots \Lambda'_{i_2} N^{-k_2} \Pi^{(1)}_{\sharp', \flat'} (\eta^{(m_1)}_{t,\natural'_1} , \dots \eta^{(m_{k_2})}_{t,\natural'_{k_2}} ; \eta^{(\ell_2)}_{x,\lozenge'} ) \xi \Big\| \\ &\hspace{1.5cm} \leq \left\{ \begin{array}{ll} C^{n+k} \| \eta_t \|^{n+k-1}  \| (\partial_t \eta_t)_x \| \| \eta_x \| \| (\cN+1)^{1/2} \xi \| \quad &\text{if } \ell_2 > 0 
\\ C^{n+k} \| \eta_t \|^{n+k} \| (\partial_t \eta_t)_x \| \| a_x \xi \| \quad &\text{if } \ell_2 = 0 \end{array} \right. 
\end{split} \end{equation}
\item[iv)] Suppose that the operators 
\[ \begin{split} &\Lambda_1 \dots \Lambda_{i_1} N^{-k_1} \Pi^{(1)}_{\sharp,\flat} (\eta_{t,\natural_1}^{(j_1)}, \dots , \eta_{t,\natural_{k_1}}^{(j_{k_1})} ; \eta^{(\ell_1)}_{y,\lozenge}), \\  &\Lambda'_1 \dots \Lambda'_{i_2} N^{-k_2} \Pi^{(1)}_{\sharp',\flat'} (\eta_{t,\natural'_1}^{(m_1)}, \dots , \eta_{t,\natural'_{k_2}}^{(m_{k_2})} ; \eta^{(\ell_2)}_{x,\lozenge'})
\end{split} \]
appear in the expansion of $\text{ad}^{(k)}_{B(\eta_t)} (b_y)$ and respectively of $\text{ad}^{(n)}_{B(\eta_t)} (b_x)$ for some $n,k \in \bN$. For $\alpha \in \bN$, let 
\[ \begin{split} \text{D} = \; &\Big\| (\cN+1)^{(\alpha-1)/2} \Lambda_1 \dots \Lambda_{i_1} N^{-k_1} \Pi^{(1)}_{\sharp, \flat} (\eta_{t,\natural_1}^{(j_1)}, \dots , \eta_{t,\natural_{k_1}}^{(j_{k_1})} ; \eta_{y,\lozenge}^{(\ell_1)}) \\ &\hspace{2.5cm} \times \Lambda'_1 \dots \Lambda'_{i_2} N^{-k_2} \Pi^{(1)}_{\sharp, \flat} (\eta_{t,\natural'_1}^{(m_1)}, \dots , \eta_{t,\natural'_{k_1}}^{(m_{k_2})} ; \eta_{x,\lozenge'}^{(\ell_2)}) \xi \Big\|  \end{split} \]
Then, if $\ell_1 > 0$, we have, for every $\alpha \in \bN$,  
\begin{equation}\label{eq:prelim-iv1}  
\text{D}  \leq \left\{ \begin{array}{ll} C^{n+k} \| \eta \|^{n+k-2} \| \eta_x \| \| \eta_y \| \| (\cN+1)^{(\alpha+1)/2} \xi \| &\quad \text{if } \ell_2 \geq 1 \\  C^{n+k} \| \eta \|^{n+k-1} \| \eta_y \| \| a_x (\cN+1)^{\alpha/2} \xi \| &\quad \text{if } \ell_2 = 0 
\end{array} \right. \end{equation}
If instead $\ell_1 = 0$, we distinguish three cases. For $\ell_2 > 1$, we obtain 
\begin{equation}\label{eq:prelim-iv2}
\begin{split}
\text{D} \leq \; &C^{n+k} \| \eta_t \|^{n+k-2} \Big\{ \| \eta_y \| \| \eta_{x} \|(\| (\cN+1)^{(\alpha-1)/2}  \xi \|+n/N   \| (\cN+1)^{(\alpha+1)/2}  \xi \|) \\ &\hspace{3cm} + \| \eta_t \| \| \eta_x \| \| a_y (\cN+1)^{\alpha/2}  \xi \| \Big\} \end{split}\end{equation} 
If $\ell_1 = 0$ and $\ell_2 = 1$, we find 
\begin{equation}\label{eq:prelim-iv3a}
\begin{split} 
\text{D} \leq &\; C^{n+k} \| \eta_t \|^{n+k-2} \Big\{ \left[ n \| \eta_x \| \| \eta_y \| + \| \eta_t \| |\eta_t (x;y)| \right]  \| (\cN+1)^{(\alpha-1)/2} \xi \| \\ &\hspace{3cm} + \| \eta_t \| \| \eta_x \| \| a_y (\cN+1)^{\alpha/2} \xi \| \Big\} \end{split} 
\end{equation} 
If $\ell_1 = 0$ and $\ell_2 = 1$ and we additionally assume that that $k+n \geq 2$ (since $\ell_1 \leq k$, $\ell_2 \leq n$ from Lemma \ref{lm:indu}, this assumption only excludes the case $k=\ell_1=0$, $n=\ell_2 =1$), we find the improved estimate 
\begin{equation}\label{eq:prelim-iv3}
\begin{split} 
\text{D} &  \leq C^{n+k} \| \eta_t \|^{n+k-2} \Big\{ N^{-1} \left[ n \| \eta_x \| \| \eta_y \| + \| \eta_t \| |\eta_t (x;y)| \right] \|  (\cN+1)^{(\alpha+1)/2} \xi \| \\ &\hspace{3.5cm} 
+ \| \eta_t \| \| \eta_x \| \| a_y (\cN+1)^{\alpha/2} \xi \| \Big\}
\end{split} \end{equation}
Finally, let $\ell_1 = \ell_2 = 0$. Then
\begin{equation}\label{eq:prelim-iv4a} \begin{split} 
\text{D} &\leq C^{n+k} \| \eta_t \|^{n+k-1} \Big\{
n N^{-1} \| \eta_{y} \| \| a_x (\cN+1)^{\alpha/2} \xi \| \\ &\hspace{4cm} + \| \eta_t \| \| a_x a_y (\cN+1)^{(\alpha-1)/2} \xi \| \Big\} \end{split} \end{equation}
If, however, $\ell_1 = \ell_2 = 0$ and, additionally, $k+n \geq 1$ (excluding the case $n=\ell_1 = k = \ell_2 = 0$), we find the improved bound  
\begin{equation}\label{eq:prelim-iv4} \begin{split} 
\text{D} &\leq C^{n+k} \| \eta_t \|^{n+k-1} \left\{
n N^{-1} \| \eta_{y} \| \| a_x \xi \| + N^{-1/2} \| \eta_t \| \| a_x a_y (\cN+1)^{\alpha/2} \xi \| \right\} \end{split} \end{equation}
\end{itemize}
\end{lemma}

\begin{proof}
Let us start with part i). If $\Lambda_1$ is either the operator $(N-\cN)/N$ or $(N-\cN+1)/N$, then, on $\cF^{\leq N}$, 
\begin{equation}\label{eq:Lambd1} 
\begin{split} 
&\left\| (\cN+1)^{-1/2} \Lambda_1 \dots  \Lambda_i N^{-k} \Pi^{(1)}_{\sharp,\flat} (\eta_{t,\natural_1}^{(j_1)}, \dots , \eta_{t,\natural_k}^{(j_k)} ; \eta^{(\ell_1)}_{t,\lozenge} (g)) \xi \right\| \\ &\hspace{2cm} \leq 2 \left\| (\cN+1)^{-1/2} \Lambda_2 \dots  \Lambda_i N^{-k} \Pi^{(1)}_{\sharp,\flat} (\eta_{t,\natural_1}^{(j_1)}, \dots , \eta_{t,\natural_k}^{(j_k)} ; \eta^{(\ell_1)}_{t,\lozenge} (g)) \xi \right\| \end{split} 
\end{equation}
If instead $\Lambda_1$ has the form (\ref{eq:Pi2-bg}) for a $p \geq 1$, we apply Lemma \ref{lm:Pi-bds} and we find (using part vi) in Lemma \ref{lm:indu})
\begin{equation}\label{eq:Lambd2} \begin{split} &\left\| (\cN+1)^{-1/2}  \Lambda_1 \dots \Lambda_{i_1} N^{-k} \Pi^{(1)}_{\sharp,\flat} (\eta_{t,\natural_1}^{(j_1)}, \dots , \eta_{t,\natural_k}^{(j_k)} ; \eta^{(\ell_1)}_{t,\lozenge} (g)) \xi \right\|
\\ &\hspace{2cm} \leq  C^p  \| \eta_t \|^{\bar{p}} \| (\cN+1)^{-1/2}  \Lambda_2 \dots  \Lambda_i N^{-k} \Pi^{(1)}_{\sharp,\flat} (\eta_{t,\natural_1}^{(j_1)}, \dots , \eta_{t,\natural_k}^{(j_k)} ; \eta^{(\ell_1)}_{t,\lozenge} (g)) \xi \| 
\end{split} \end{equation}
where we used the notation $\bar{p} = q_1 + \dots + q_p$ for the total number of $\eta_t$-kernels appearing in (\ref{eq:Pi2-prel}). Iterating the bounds (\ref{eq:Lambd1}) and (\ref{eq:Lambd2}), we conclude that 
\begin{equation}\label{eq:bd1-Lam} 
\begin{split} &\| (\cN+1)^{-1/2}  \Lambda_1
 \dots \Lambda_{i_1} N^{-k} \Pi^{(1)}_{\sharp,\flat} (\eta_{t,\natural_1}^{(j_1)}, \dots , \eta_{t,\natural_k}^{(j_k)} ; \eta^{(\ell_1)}_{t,\lozenge} (g)) \xi \| \\ &\hspace{.2cm} \leq C^{r+p_1 + \dots + p_s} \| \eta_t \|^{\bar{p}_1 + \dots + \bar{p}_s} \| (\cN+1)^{1/2} N^{-k} \Pi^{(1)}_{\sharp,\flat} (\eta_{t,\natural_1}^{(j_1)}, \dots , \eta_{t,\natural_k}^{(j_k)} ; \eta^{(\ell_1)}_{t,\lozenge} (g))
 \xi \| \end{split}  \end{equation}
if $r$ of the operators $\Lambda_1, \dots , \Lambda_{i_1}$ have either the form $(N-\cN)/N$ or the form $(N-\cN+1)/N$, and the other $s=i_1-r$ are $\Pi^{(2)}$-operators of the form (\ref{eq:NPi2}) of order $p_1, \dots , p_s$ , containing $\bar{p}_1, \dots \bar{p}_s$ $\eta_t$-kernels. Again with Lemma \ref{lm:Pi-bds}, we obtain 
\begin{equation}\label{eq:prel-i} \begin{split}  &\| (\cN+1)^{-1/2}  \Lambda_1
 \dots \Lambda_{i_1} N^{-k_1} \Pi^{(1)}_{\sharp,\flat} (\eta_{t,\natural_1}^{(j_1)}, \dots , \eta_{t,\natural_{k_1}}^{(j_{k_1})} ; \eta^{(\ell_1)}_{t,\lozenge} (g)) \xi \| \\ &\hspace{2cm} \leq C^{r+p_1 +\dots +p_s+j_1 + \dots + j_{k_1}+l_1} \| \eta_t \|^{\bar{p}_1 + \dots + \bar{p}_s+j_1 +\dots +j_{k_1}+l_1} \| g \| \| \xi \| \\ &\hspace{2cm} \leq C^{n} \| \eta_t \|^n \| g \| \| \xi \| \, . \end{split} 
\end{equation}
This shows the first bound in part i). Now, assume that at least one of the $\Lambda_m$ operators, for $m \in \{ 1, \dots , i_1 \}$, has the form (\ref{eq:Pi2-prel}). Since, for $\Psi \in \cF^{\leq N}$, 
\[ \begin{split} &\| (\cN+1)^{-1/2} N^{-p} \Pi^{(2)}_{\sharp,\flat} (\eta_{t,\natural_1}^{(q_1)}, \dots , \eta_{t,\natural_p}^{(q_p)}) \Psi \| \\ &\hspace{2cm} \leq C^p \| \eta_t \|^{q_1 + \dots + q_p} N^{-p} \| (\cN+1)^{p-1/2} \Psi \| \\ &\hspace{2cm} \leq C^p \| \eta_t \|^{q_1 + \dots + q_p} N^{-1/2} \| \Psi \| \end{split} \]
for any $p \geq 1$, in this case we can improve (\ref{eq:prel-i}) to
\[ \begin{split}  &\| (\cN+1)^{-1/2}  \Lambda_1
 \dots \Lambda_{i_1} N^{-k_1} \Pi^{(1)}_{\sharp,\flat} (\eta_{t,\natural_1}^{(j_1)}, \dots , \eta_{t,\natural_{k_1}}^{(j_{k_1})} ; \eta^{(\ell_1)}_{t,\lozenge} (g)) \xi \| \\ &\hspace{6cm} \leq C^{n} N^{-1/2} \| \eta_t \|^n \| g \| \| (\cN+1)^{1/2} \xi \| \, . \end{split} 
\]
Similarly, if $k_1 \geq 1$, we have by Lemma \ref{lm:Pi-bds},   
\[ \begin{split} 
N^{-k_1} &\left\|  (\cN+1)^{-1/2}  \Pi^{(1)}_{\sharp,\flat} (\eta_{t,\natural_1}^{(j_1)}, \dots , \eta_{t,\natural_k}^{(j_{k_1})} ; \eta_{t,\natural_{k+1}}^{(\ell_1)} (g)) \xi \right\|  \\ &\hspace{2cm}\leq  N^{-k_1} C^{k_1} \| \eta_t \|^{j_1 + \dots + j_{k_1}+\ell_1} \|g \|  \|(N+1)^{k_1-1/2} \xi \| \\ &\hspace{2cm} \leq  C^k N^{-1/2} \| \eta_t \|^{j_1 + \dots + j_{k_1} +\ell_1} \| g \|  \|(N+1)^{1/2} \xi \| 
\end{split} \]
Hence, also in this case, the bound (\ref{eq:prel-i2}) holds true. 

If $\ell_1 \geq 1$, part ii) can be proven similarly to part i), noticing that 
\[ \| (\eta^{(\ell_1)}_{t,\lozenge} r)_x \| \leq \| \eta_t \|^{\ell_1 -1} \| (\eta_{t} r)_x \| \, . \]
If instead $\ell_1 = 0$, it follows from Lemma \ref{lm:indu}, part v), that the field operator associated with $(\eta^{(\ell_1)}_{t,\lozenge} r)_x = r_x$ (the one appearing on the right of $\Pi^{(1)}$) is an annihilation operator (acting directly on $\xi$). Hence, (\ref{eq:prelim-ii}) holds true also in this case.   

Let us now consider part iii). We can bound, first of all 
\[ \left\| (\cN+1)^{-1/2} \Lambda_1 \dots \Lambda_{i_1} N^{-k_1} \Pi^{(1)}_{\sharp,\flat} (\eta_{t,\natural_1}^{(j_1)} , \dots , \eta_{t,\natural_{k_1}}^{(j_{k_1})} ; (\eta_t^{(\ell_1 + 1)} r)_x) \Psi \right\| \leq C^n \| \eta_t \|^n \| (\eta_t r)_x \| \| \Psi \| \]
and 
\[ \begin{split} &\left\|  (\cN+1)^{-1/2} \Lambda_1 \dots \Lambda_{i_1} N^{-k_1} \Pi^{(1)}_{\sharp,\flat} (\eta_{t,\natural_1}^{(j_1)} , \dots , \eta_{t,\natural_{k_1}}^{(j_{k_1})} ; (\eta_{t,\lozenge}^{(\ell_1)} \partial_t \eta_{t,\wt{\lozenge}})_x) \Psi \right\| \\ &\hspace{8cm} \leq C^n \| \eta_t \|^n \| (\partial_t \eta_t)_x \| \| \Psi \| \end{split} \]
Choosing now \[ \Psi = \Lambda'_1 \dots \Lambda'_{i_2} N^{-k_2} \Pi^{(1)}_{\sharp', \flat'} (\eta_{t,\natural'_1}^{(m_1)} , \dots , \eta_{t,\natural'_{k_2}}^{(m_{k_2})} ; \eta_{x,\lozenge'}^{(\ell_2)})\xi,\] and proceeding as in part ii), distinguishing the cases $\ell_2 \geq 1$ and $\ell_2 = 0$, we obtain (\ref{eq:prelim-iii}) and (\ref{eq:prelim-iiib}).

Finally, we consider part iv). If $\ell_1 > 0$, we can proceed as in part iii) to show (\ref{eq:prelim-iv1}). So, let us focus on the case $\ell_1 = 0$. In this case, the field operator on the right of the first $\Pi^{(1)}$-operator (the one on the left) is an annihilation operator, $a_y$. To estimate $\text{D}$, we need to commute $a_y$ to the right, until it hits $\xi$. To commute $a_y$ through factors of $\cN$, we just use the pull-through formula $a_y \cN = (\cN+1) a_y$. When we commute $a_y$ through a pair of creation and/or annihilation operators associated with a kernel $\eta^{(j)}_t$ for a $j \geq 1$ (as the ones appearing in the $\Pi^{(2)}$-operators of the form (\ref{eq:Pi2-prel}) or in the operator $\Pi^{(1)}$-operator), we generate a creation or an annihilation operator with argument $\eta_{y}^{(j)}$, whose $L^2$-norm is uniformly bounded. At the same time, we spare a factor $N^{-1}$. For example, we have
\[ \left[ a_y , \int a^*_{x_i} a_{y_i} \eta^{(j)} (x_i ; y_i) dx_i dy_i \right] = a (\bar{\eta}^{(j)}_y) \]
At the end, we have to commute $a_y$ through the field operator with argument $\eta_{x,\lozenge'}^{(\ell_2)}$. The commutator is trivial if $\ell_2$ is even (because then the corresponding field operator is an annihilation operator; see Lemma \ref{lm:indu}, part v)). It is given by 
\begin{equation}\label{eq:comm-ell2} [ a_y, a^* (\eta_{x,\lozenge'}^{(\ell_2)}) ] = \eta_{t,\lozenge'}^{(\ell_2)} (x;y) \end{equation} if $\ell_2$ is odd. If $\ell_2 \geq 2$, we can bound $|\eta_{t,\lozenge'}^{(\ell_2)} (x;y)| \leq \| \eta_t \|^{\ell_2 - 2} \| \eta_{x} \| \| \eta_{y} \|$ and we obtain (taking into account the fact that there are at most $n$ pairs of fields with which $a_y$ has to be commuted)
\[ \begin{split} 
\text{D} \leq &C^{k+n}  \| \eta_t \|^{k+n-2} \Big\{ 
n N^{-1} \| \eta_y \| \| \eta_x \| \| (\cN+1)^{(\alpha+1)/2} \xi \| \\ &\hspace{3cm} + \| \eta_x \| \| \eta_y \| \| (\cN+1)^{(\alpha-1)/2} \xi \| + \| \eta_t \| \| \eta_x \| \| a_y (\cN+1)^{\alpha/2}\xi \| \Big\} \, .\end{split}  \]
If instead $\ell_2 = 1$, the r.h.s. of (\ref{eq:comm-ell2}) blows up, as $N \to \infty$. To make up for this singularity, we use the additional assumption $k +n \geq 2$. Combining this information with $\ell_1 = 0$, $\ell_2 = 1$, we conclude that either $k_1 > 0$ or $k_2 >0$ or there exists $i \in \bN$ such that either $\Lambda_i$ or $\Lambda'_i$ is a $\Pi^{(2)}$-operator of the form (\ref{eq:Pi2-prel}) with $p \geq 1$. This factor allows us to gain a factor $(\cN+1)/N$ in the estimate for the term arising from the commutator (\ref{eq:comm-ell2}). We conclude that, in this case, 
\[ \begin{split} 
\text{D} \leq \; &C^{k+n}  \| \eta_t \|^{k+n-2} \Big\{ 
n N^{-1} \| \eta_y \| \| \eta_x \| \| (\cN+1)^{(\alpha+1)/2} \xi \| \\ & \hspace{3cm} + N^{-1} |\eta_t (x;y)|  \| (\cN+1)^{(\alpha+1)/2} \xi \| + \| \eta_t \| \| \eta_x \| \| a_y (\cN+1)^{\alpha/2} \xi \| \Big\} \, . \end{split}  \]
Finally, let us consider the case $\ell_2 = 0$. Here we proceed as before, commuting $a_y$ to the right. The commutator produces at most $n$ factors, whose norm can be bounded similarly as before. We easily conclude that
\[ \text{D} \leq C^{k+n} \| \eta_t \|^{k+n-1} \left\{ nN^{-1} \| \eta_x \| \| a_y (\cN+1)^{\alpha/2} \xi \| + \| \eta_t \| \| a_x a_y (\cN+1)^{(\alpha-1)/2} \xi \| \right\} \]
If we impose the additional condition $k + n \geq 1$, we deduce that either $k_1 > 0$ or $k_2 >0$ or there exists $i \in \bN$ such that either $\Lambda_i$ or $\Lambda'_i$ is a $\Pi^{(2)}$-operator of the form (\ref{eq:Pi2-prel}) with $p \geq 1$. Similarly as we argued in the case $\ell_2 = 1$, when estimating the contribution with the two annihilation operators $a_x, a_y$ acting on $\xi$, we can therefore extract an additional factor $(\cN+1)/N$. Under this additional condition, we obtain  
\[\text{D} \leq C^{k+n} \| \eta_t \|^{k+n-1} \left\{ nN^{-1} \| \eta_x \| \| a_y \xi \| + N^{-1/2} \| \eta_t \| \| a_x a_y (\cN+1)^{(\alpha-1)/2} \xi \| \right\} \]
which proves (\ref{eq:prelim-iv4}). 
\end{proof}

\subsection{Analysis of $e^{-B(\eta_t)} \cL_{N,t}^{(0)} e^{B (\eta_t)}$}
\label{subsec:L0}

{F}rom the definition (\ref{eq:L0-4}), we can write
\[ \begin{split} \cL^{(0)}_{N,t} = \; &C_{N,t} - \langle \wtph_t , \left[ N^3 V(N.) w_\ell (N.) * |\wtph_t|^2 \right] \wtph_t \rangle \cN
\\ &+\frac{1}{2N} \langle \wtph_t , \left[ N^3 V(N.) * |\wtph_t|^2 \right] \wtph_t \rangle \cN + \frac{1}{2N} \langle \wtph_t , \left[ N^3 V(N.) * |\wtph_t|^2 \right] \wtph_t \rangle \cN^2 \end{split} \]
with   
\[ C_{N,t} = \frac{N}{2} \langle \wtph_t , \left[ N^3 V(N.) w_\ell (N.) * |\wtph_t|^2 \right] \wtph_t \rangle - \frac{1}{2}  \langle \wtph_t , \left[ N^3 V(N.) * |\wtph_t|^2 \right] \wtph_t \rangle \]
The properties of the other terms are described in the next proposition. 
\begin{prop}\label{prop:L0}
Under the same assumptions as in Theorem \ref{thm:gene}, there exist constants $C, c >0$ such that
\begin{equation}\label{eq:bds-L0} 
\begin{split}
\left| \left\langle \xi, e^{-B(\eta_t)}\left(\cL^{(0)}_{N,t} - C_{N,t} \right) e^{B(\eta_t)} \xi \right\rangle \right|&\leq C \langle \xi, (\cN+1) \xi \rangle \\ \left| \left\langle \xi, \left[ \cN, e^{-B(\eta_t)} \left( \cL^{(0)}_{N,t} - C_{N,t} \right) e^{B(\eta_t)} \right] \xi \right\rangle \right| &\leq C \langle \xi, (\cN+1) \xi \rangle \\
\left|  \left\langle \xi, \left[ a^* (g_1) a(g_2), e^{-B(\eta_t)} \left( \cL^{(0)}_{N,t} - C_{N,t} \right) e^{B(\eta_t)} \right] \xi \right\rangle \right| &\leq C \| g_1 \| \| g_2 \| \langle \xi, (\cN+1) \xi \rangle  \\
\left|  \partial_t \left\langle \xi, e^{-B(\eta_t)} \left( \cL^{(0)}_{N,t} - C_{N,t} \right) e^{B(\eta_t)} \xi \right\rangle \right|  &\leq C e^{c|t|} \langle \xi, (\cN+1) \xi \rangle 
\end{split} 
\end{equation}
for all $t \in \bR$, $g_1, g_2 \in L^2 (\bR^3)$, $\xi \in \cF^{\leq N}$.
\end{prop}

In order to show Proposition \ref{prop:L0}, we need to conjugate the number of particles operator $\cN$ with the generalized Bogoliubov transformation $e^{-B(\eta_t)}$. To this end, we make use of the following lemma, where, for later convenience, we consider conjugation of more general quadratic operators.
\begin{lemma}\label{lm:R}
Let $r: L^2 (\bR^3) \to L^2 (\bR^3)$ be a bounded linear operator. Consider the Fock-space operators 
\[ \begin{split} R_1 = \int dx dy \, r(y;x) \, b_x^* b_y \qquad \text{and} \qquad R_2 = \int dx dy \, r(y;x) \, a_x^* a_y \end{split} \]
mapping $\cF^{\leq N}$ in itself. Then we have the bounds
\begin{equation}\label{eq:R1R2} \begin{split} 
\left| \left\langle \xi_1 , e^{-B(\eta_t)} R_i e^{B(\eta_t)} \xi_2 \right\rangle \right|  &\leq C \| r \|_\text{op} \| (\cN+1)^{1/2} \xi_1 \| \|(\cN+1)^{1/2} \xi_2 \|  \\
\left| \left\langle \xi_1 , \left[ \cN, e^{-B(\eta_t)} R_i e^{B(\eta_t)} \right] \xi_2 \right\rangle \right|  &\leq C  \| r \|_\text{op} \, \| (\cN+1)^{1/2} \xi_1 \| \|(\cN+1)^{1/2} \xi_2 \|  \\
\left| \left\langle \xi_1 , \left[ a^* (g_1) a(g_2) , e^{-B(\eta_t)} R_i e^{B(\eta_t)} \right] \xi_2 \right\rangle \right|  &\leq C \| r \|_\text{op} \| g_1 \| \| g_2 \| \\ &\hspace{1cm} \times \| (\cN+1)^{1/2} \xi_1 \| \|(\cN+1)^{1/2} \xi_2 \|  
\end{split} \end{equation}
for $i=1,2$ and all $\xi_1, \xi_2 \in \cF^{\leq N}$. Furthermore, if $r = r_t$ is differentiable in $t$, we find 
\begin{equation}\label{eq:parR1R2} \left| \partial_t \left\langle \xi_1 , e^{-B(\eta_t)} R_i e^{B(\eta_t)} \xi_2 \right\rangle \right|  \leq C e^{c|t|} (\| r \|_\text{op} + \| \dot{r} \|_\text{op}) \, \| (\cN+1)^{1/2} \xi_1 \| \|(\cN+1)^{1/2} \xi_2 \| 
\end{equation}
for $i=1,2$ and all $\xi_1, \xi_2 \in \cF^{\leq N}$.
\end{lemma}
\begin{proof}
We consider first the operator $R_1$. By Lemma \ref{lm:conv-series}, we expand
\begin{equation}\label{eq:eR1e} \begin{split}
e^{-B(\eta_t)} R_1 e^{B(\eta_t)} &= \int dx  \, e^{-B(\eta_t)} b^* (r_x) b_x e^{B(\eta_t)} \\ &= \sum_{k,n \geq 0} \frac{(-1)^{k+n}}{k! n!} \int dx \, \text{ad}^{(n)}_{B(\eta_t)} (b^* (r_x)) \text{ad}^{(k)}_{B(\eta_t)} (b_x) \end{split} \end{equation}
with the notation $r_x (y) = r(x;y)$. According to Lemma \ref{lm:indu} the operator 
\[  \int dx \, \text{ad}^{(n)}_{B(\eta_t)} (b^* (r_x))
\, \text{ad}^{(k)}_{B(\eta_t)} (b_x) \]
is given by the sum of $2^{n+k} n!k!$ terms having the form
\begin{equation}\label{eq:typR}
\begin{split}  \text{E} := \; &\int dx  \, 
N^{-k_1} \Pi^{(1)}_{\sharp,\flat} (\eta_{t,\natural_1}^{(j_1)}, \dots , \eta_{t,\natural_{k_1}}^{(j_{k_1})} ; (\eta_{t,\lozenge}^{(\ell_1)} r)_x)^* \Lambda_{i_1}^* \dots \Lambda_1^*  \\ &\hspace{3cm} \Lambda'_1 \dots \Lambda'_{i_2} N^{-k_2} \Pi^{(1)}_{\sharp',\flat'}
(\eta_{t,\natural'_1}^{(m_1)}, \dots , \eta_{t,\natural'_{k_2}}^{(m_{k_2})} ; \eta_{x,\lozenge'}^{(\ell_2)} ) \end{split} \end{equation}
where $i_1, i_2, k_1, k_2, \ell_1 , \ell_2 \geq 0$, $j_1,\dots , j_{k_1}, m_1, \dots , m_{k_2} \geq 1$, and where each operator $\Lambda_i$ and $\Lambda'_i$ is either a factor $(N-\cN)/N$, a factor $(N+1-\cN)/N$ or a $\Pi^{(2)}$-operator of the form
\begin{equation}\label{eq:Pi2R} N^{-p} \Pi^{(2)}_{\underline{\sharp}, \underline{\flat}} (\eta_{t,\underline{\natural}_1}^{(q_1)}, \dots , \eta_{t,\underline{\natural}_p}^{(q_p)}) \end{equation}
for a $p \geq 1$ and powers $q_1,\dots , q_p \geq 1$. 
With Cauchy-Schwarz we find 
\begin{equation}\label{eq:bd-R} 
\begin{split}
| \langle \xi_1, \text{E} \xi_2 \rangle | \leq \; & \int dx \,\left\| \Lambda_1 \dots \Lambda_{i_1} 
N^{-k_1} \Pi^{(1)}_{\sharp, \flat} (\eta_{t,\natural_1}^{(j_1)}, \dots , \eta_{t,\natural_{k_1}}^{(j_{k_1})} ; (\eta_{t,\lozenge}^{(\ell_1)}r)_x ) \xi_1 \right\| \\ &\hspace{2cm} \times  \left\| \Lambda'_1 \dots \Lambda'_{i_2} N^{-k_2} \Pi^{(1)}_{\sharp',\flat'}
(\eta_{t,\natural'_1}^{(m_1)}, \dots , \eta_{t,\natural'_{k_2}}^{(m_{k_2})} ; \eta_{x,\lozenge'}^{(\ell_2)} ) \xi_2 \right\|\end{split} \end{equation}
for every $\xi_1,\xi_2 \in \cF^{\leq N}$. With Lemma \ref{lm:prelim}, part ii), we find that 
\begin{equation}\label{eq:fin-R}
\begin{split} 
\left| \langle \xi_1, \text{E} \, \xi_2 \rangle \right| &\leq C^{k+n} \| r \|_\text{op} \| \eta_t \|^{n+k} \| (\cN+1)^{1/2} \xi_1 \| \| (\cN+1)^{1/2} \xi_2 \| 
\end{split} \end{equation}
where we used the fact that
\[ \int dx \, \| a(r_x) \xi_1 \|^2 = \langle \xi_1, d\Gamma (r^2) \xi_1 \rangle  \leq \| r^2 \|_\text{op} \| \cN^{1/2} \xi_1 \|^2 \leq \| r \|_\text{op}^2 \| \cN^{1/2} \xi_1 \|^2 \]

From (\ref{eq:eR1e}), we conclude that, if $\sup_t \| \eta_t \|$ is small enough, 
\begin{equation}\label{eq:eR1e-bd} \left| \langle \xi_1 , e^{-B(\eta_t)} R_1 e^{B(\eta_t)} \xi_2 \rangle \right| \leq C \| r \|_\text{op} \| (\cN+1)^{1/2} \xi_1 \| \| (\cN+1)^{1/2} \xi_2 \|  \end{equation}
This proves the first bound in (\ref{eq:R1R2}), if $i=1$. The other two bounds in (\ref{eq:R1R2}) and the bound in (\ref{eq:parR1R2}) for $i=1$ can be proven similarly. To be more precise, we first expand the operator $e^{-B(\eta_t)} R_1
e^{B(\eta_t)}$ as in (\ref{eq:eR1e}), where the $(n,k)$-th term can be written as the sum of $2^{n+k} k!n!$ terms of the form (\ref{eq:typR}). Then we use Lemma \ref{lm:prel1} to express the commutator of (\ref{eq:typR}) with $\cN$ or with $a^* (g_1) a(g_2)$ or its time-derivative as a sum of at most $2(k+n+1)$ terms having again the form (\ref{eq:typR}), with just one of the $\eta_t$-kernels appropriately replaced. Finally, we proceed as above to show that the matrix elements of such a term can be bounded as in (\ref{eq:fin-R}). We omit further details.

Let us now consider the operator $R_2$. We start by writing
\[ \begin{split} e^{-B(\eta_t)} R_2 e^{B(\eta_t)} &= R_2 + \int_0^1 ds \,e^{-s B(\eta_t)} [R_2 , B(\eta_t) ] e^{s B(\eta_t)} \\ &= R_2 + \int_0^1 ds \int dx dy \, r(y;x) e^{-s B(\eta_t)} \left[ a_x^* a_y , B(\eta_t) \right] e^{s B(\eta_t)} \\ &= R_2 + \int_0^1 ds\int dx e^{-sB(\eta_t)} \left[ b ((\eta_t r)_x) b_x + \text{h.c.} \right] e^{sB(\eta_t)}  \end{split} \]
Expanding as in Lemma \ref{lm:conv-series} and then integrating over $s$, we find
\begin{equation}\label{eq:N-exp} \begin{split} 
e^{-B(\eta_t)} &R_2 e^{B(\eta_t)} \\ = &\; R_2 + \sum_{k,n \geq 0} \frac{(-1)^{k+n}}{k! n! (k+n+1)} \int dx \left[ \text{ad}_{B (\eta_t)}^{(n)} (b ((\eta_t r)_x)) \text{ad}^{(k)}_{B(\eta_t)} (b_x) + \text{h.c.} \right] \end{split} \end{equation}
With Lemma \ref{lm:indu}, we can write the operator
\begin{equation}\label{eq:exp0} \int dx \,\text{ad}_{B (\eta_t)}^{(n)} (b ((\eta_t r)_x)) \text{ad}^{(k)}_{B(\eta_t)} (b_x)  \end{equation}
as a sum of $2^{n+k} k!n!$ contributions of the form
\begin{equation}\label{eq:typadad} \begin{split} 
\text{E} = \int dx \, \Lambda_1 \dots \Lambda_{i_1} N^{-{k_1}} &\Pi^{(1)}_{\sharp,\flat} (\eta_{t,\natural_1}^{(j_1)}, \dots , \eta_{t,\natural_{k_1}}^{(j_{k_1})} ; (\eta_{t,\lozenge}^{(\ell_1+1)} r)_x) \\ &\times \Lambda'_1 \dots \Lambda'_{i_2} N^{-{k_2}} \Pi^{(1)}_{\sharp',\flat'} (\eta_{t,\natural'_1}^{(m_1)} , \dots , \eta_{t,\natural'_{k_2}}^{(m_{k_2})} ; \eta_{x,\lozenge'}^{(\ell_2)})\end{split} \end{equation}
where each $\Lambda_i$ and $\Lambda'_i$ is either $(N-\cN)/N$, $(N+1-\cN)/N$ or an operator of the form
\begin{equation}\label{eq:NPi2} N^{-p} \, \Pi^{(2)}_{\underline{\sharp},\underline{\flat}} (\eta_{t,\underline{\natural}_1}^{(q_1)}, \dots \eta_{t,\underline{\natural}_p}^{(q_p)} ) 
\end{equation}

{F}rom Lemma \ref{lm:prelim}, part iii), we obtain that 
\[ \begin{split} |\langle \xi_1 , \text{E} \xi_2 \rangle | \leq \; & \, \| (\cN+1)^{1/2} \xi_1 \| \\ & \times \int dx \, \Big\| (\cN+1)^{-1/2} 
\Lambda_1 \dots \Lambda_{i_1} N^{-{k_1}} \Pi^{(1)}_{\sharp,\flat} (\eta_{t,\natural_1}^{(j_1)}, \dots , \eta_{t,\natural_{k_1}}^{(j_{k_1})} ; (\eta_{t,\lozenge}^{(\ell_1+1)} r)_x) \\ &\hspace{3cm} \times \Lambda'_1 \dots \Lambda'_{i_2} N^{-{k_2}} \Pi^{(1)}_{\sharp',\flat'} (\eta_{t,\natural'_1}^{(m_1)} , \dots , \eta_{t,\natural'_{k_2}}^{(m_{k_2})} ; \eta_{x,\lozenge'}^{(\ell_2)}) \xi_2 \Big\|  \\ \leq \; &C^{n+k} \| r \|_\text{op} \, \| \eta_t \|^{k+n+1}  \| (\cN+1)^{1/2} \xi_1 \| \| (\cN+1)^{1/2} \xi_2 \| \end{split} \]
This implies that, if $\sup_t \| \eta_t \|$ is small enough, 
\[ \Big|\langle \xi_1, e^{-B(\eta_t)} R_2 e^{B(\eta_t)} \xi_2 \rangle \Big| \leq C \| r \|_\text{op}  \| (\cN+1)^{1/2} \xi_1 \| \| (\cN+1)^{1/2} \xi_2 \| \]
As in the analysis of $R_1$ above, also here one can show the other bounds in (\ref{eq:R1R2}) for the commutators of $e^{-B(\eta_t)} R_1 e^{B(\eta_t)}$ with $\cN$ and with $a^* (g_1) a(g_2)$ and for its time-derivative. 
\end{proof}

Next, we use Lemma \ref{lm:R} to show Prop. \ref{prop:L0}. 
\begin{proof}[Proof of Prop. \ref{prop:L0}]
To control $\cL_{N,t}^{(0)}$ we start by noticing that, with Young's inequality,
\begin{equation}\label{eq:you1} \begin{split} \left| \langle \wtph_t, \left[ N^3 V(N.) *|\wtph_t|^2 \right] \wtph_t \rangle \right| &\leq \int N^3 V(N (x-y)) |\wtph_t (x)|^2 |\wtph_t (y)|^2 dx dy \\ &\leq C \| \wtph_t \|_4^4 \leq C \| \wtph_t \|_{H^1}^4 \leq C \end{split}
\end{equation}
and
\begin{equation}\label{eq:youdot} \left| \partial_t  \langle \wtph_t, \left[ N^3 V(N.) *|\wtph_t|^2 \right] \wtph_t \rangle \right| \leq C \| \wtph_t \|^3_4 \| \dot\wtph_t \|_4 \leq C \| \wtph_t \|_{H^1}^3 \| \wtph_t \|_{H^3} \leq C e^{c|t|} \end{equation}
for constants $C,c > 0$. Similarly, we also have
\begin{equation}\label{eq:you2} \begin{split} \left| \langle \wtph_t, \left[ N^3 V(N.) w_\ell (N.) *|\wtph_t|^2 \right] \wtph_t \rangle \right|  &\leq C \\
\left|\partial_t  \langle \wtph_t, \left[ N^3 V(N.) w_\ell (N.) *|\wtph_t|^2 \right] \wtph_t \rangle \right|  &\leq C e^{c|t|} \,.
\end{split} \end{equation}

By (\ref{eq:you1}), (\ref{eq:youdot}), (\ref{eq:you2}), it is enough to show the four bound in (\ref{eq:bds-L0}) with $\cL_{N,t}^{(0)} - C_{N,t}$ replaced by $\cN$ and by $\cN^2 / N$. If we replace $\cL_{N,t}^{(0)} - C_{N,t}$ with $\cN$, the bounds in (\ref{eq:bds-L0}) follow from Lemma \ref{lm:R}. To prove that these bounds also hold for $\cN^2/N$, we use again Lemma \ref{lm:R}. Setting $\xi_2 =  e^{-B(\eta_t)} (\cN/N) e^{B(\eta_t)} \xi$, we have 
\[ \left| \langle \xi , e^{-B(\eta_t)} (\cN^2/N) e^{B(\eta_t)} \xi \rangle \right| = \left| \langle \xi , e^{-B(\eta_t)} \cN e^{B(\eta_t)} \xi_2 \rangle \right| \leq C \| (\cN+1)^{1/2} \xi \| \| (\cN+1)^{1/2} \xi_2 \|\]
Since, by Lemma \ref{lm:Npow}, 
\[\begin{split} \| (\cN+1)^{1/2} \xi_2 \|^2 &= N^{-2} \langle \xi, e^{-B(\eta_t)} \cN e^{B(\eta_t)} (\cN+1) e^{-B(\eta_t)} \cN e^{B(\eta_t)} \xi \rangle \\ &\leq 
N^{-2} \langle \xi, (\cN+1)^3 \xi \rangle \leq C \langle \xi, (\cN+1) \xi \rangle \end{split} \]
for all $\xi \in \cF^{\leq N}$, we have  
\[ \left| \langle \xi , e^{-B(\eta_t)} (\cN^2/N) e^{B(\eta_t)} \xi \rangle \right| \leq C \| (\cN+1)^{1/2} \xi \|^2 \]
Using Lemma \ref{lm:R} and Leibniz rule, we also find
\[ \begin{split} 
\left| \langle \xi , \big[ \cN, e^{-B(\eta_t)} (\cN^2/N) e^{B(\eta_t)} \big] \xi \rangle \right| &\leq C \| (\cN+1)^{1/2} \xi \|^2 \\
\left| \langle \xi , \big[ a^* (g_1) a(g_2) , e^{-B(\eta_t)} (\cN^2/N) e^{B(\eta_t)} \big] \xi \rangle \right| &\leq C \| g_1 \| \| g_2 \| \| (\cN+1)^{1/2} \xi \|^2 \\
\left| \langle \xi ,   \partial_t \big( e^{-B(\eta_t)} (\cN^2/N) e^{B(\eta_t)}\big) \xi \rangle \right| &\leq C e^{c|t|} \| (\cN+1)^{1/2} \xi \|^2  
\end{split} \]
This concludes the proof of the proposition. 
\end{proof}

\subsection{Analysis of $e^{-B(\eta_t)} \cL^{(1)}_{N,t} e^{B(\eta_t)}$}
\label{sec:L1}

We recall that 
\[ \cL^{(1)}_{N,t} = \sqrt{N} b(h_{N,t}) - \frac{\cN+1}{\sqrt{N}} b (\wt{h}_{N,t})  + \text{h.c.} \]
where we used the notation $h_{N,t} = (N^3 V(N.)w_\ell (N.) * 
|\wtph_t|^2)\wtph_t$ and $\wt{h}_{N,t} = (N^3 V(N.) *|\wtph_t|^2) \wtph_t$. 
We write 
\begin{equation}\label{eq:eL1e} e^{-B(\eta_t)} \cL^{(1)}_{N,t} e^{B(\eta_t)} = \sqrt{N} \left[ b (\cosh_{\eta_t} (h_{N,t})) + b^* (\sinh_{\eta_t} (\bar{h}_{N,t})) + \text{h.c.} \right] + \cE^{(1)}_{N,t} 
\end{equation}
In the next proposition we show that the operator $\cE^{(1)}_{N,t}$, defined in (\ref{eq:eL1e}), its commutator with $\cN$ and its time-derivative can all be controlled by the number of particles operator $\cN$ (while the first term on the r.h.s. of (\ref{eq:eL1e}) will cancel with contributions arising from conjugation of $\cL^{(3)}_{N,t}$).
\begin{prop}\label{prop:L1}
Under the same assumptions as in Theorem \ref{thm:gene}, there exist constants $C, c >0$ such that
\begin{equation}\label{eq:bds-L1} 
\begin{split} 
\left| \langle \xi, \cE^{(1)}_{N,t} \xi \rangle \right| &\leq C  \langle \xi, (\cN+1) \xi \rangle \\
\left| \langle \xi, \left[ \cN, \cE^{(1)}_{N,t} \right] \xi \rangle \right| &\leq C  \langle \xi, (\cN+1) \xi \rangle \\
\left| \langle \xi, \left[ a^* (g_1) a(g_2),  \cE^{(1)}_{N,t} \right] \xi \rangle \right| &\leq C  \| g_1 \| \| g_2 \| \langle \xi, (\cN+1) \xi \rangle \\
\left| \partial_t \langle \xi, \cE^{(1)}_{N,t} \xi \rangle \right| &\leq C e^{c|t|} \langle \xi, (\cN+1) \xi \rangle
\end{split} \end{equation}
for all $\xi \in \cF^{\leq N}$. 
\end{prop}

\begin{proof}
We start with the observation that
\begin{equation}\label{eq:bd-hNt} \begin{split} \| h_{N,t} \|, \| \wt{h}_{N,t} \| &\leq 
C \| \wtph_t \|_{H^1}^3 \leq C \\  
\| \partial_t h_{N,t} \|, \| \partial_t \wt{h}_{N,t} \| &\leq \| \wtph_t \|_{H^1}^2 \| \wtph_t \|_{H^3} \leq C e^{c|t|} \end{split} \end{equation}
uniformly in $N$ and for all $t \in \bR$. Recall that, by (\ref{eq:eL1e}), 
\begin{equation}\label{eq:cE1-bd} \begin{split}  \cE^{(1)}_{N,t}  = \; &\left[ e^{-B(\eta_t)} \cL^{(1)}_{N,t} e^{B(\eta_t)} - \sqrt{N} \left( b(\cosh_{\eta_t} (h_{N,t}) + b^* (\sinh_{\eta_t} (h_{N,t}) + \text{h.c.} \right) \right] \\ = \; &\sqrt{N} \left[ e^{-B(\eta_t)} b(h_{N,t}) e^{B(\eta_t)} -  \left( b(\cosh_{\eta_t} (h_{N,t}) + b^* (\sinh_{\eta_t} (h_{N,t}) \right) \right] +\text{h.c.}  \\ &+  N^{-1/2} e^{-B(\eta_t)} (\cN+1) b (\wt{h}_{N,t}) e^{B(\eta_t)} \end{split} \end{equation}
Set \[ D(g) =  e^{-B(\eta_t)} b(g) e^{B(\eta_t)} - b(\cosh_{\eta_t} (g)) - b^* (\sinh_{\eta_t} (g)) \]
We observe that Proposition \ref{prop:L1} follows if we prove that 
\begin{equation} \label{eq:bg} 
\begin{split} 
| \langle \xi_1, D(g) \xi_2 \rangle | &\leq C N^{-1/2} \| g \| \| (\cN+1)^{1/2} \xi_1 \| \| (\cN+1)^{1/2} \xi_2 \| \\
| \langle \xi_1, [\cN, D(g)] \xi_2 \rangle | &\leq C N^{-1/2} \| g \| \| (\cN+1)^{1/2} \xi_1 \| \| (\cN+1)^{1/2} \xi_2 \| \\
| \langle \xi_1, [a^* (g_1) a (g_2) , D(g)] \xi_2 \rangle | &\leq C N^{-1/2} \| g \| \| g_1 \| \| g_2 \| \| (\cN+1)^{1/2} \xi_1 \| \| (\cN+1)^{1/2} \xi_2 \| \\
| \langle \xi_1, \partial_t D(g) \xi_2 \rangle | &\leq C N^{-1/2} (\| g \| + \| \dot{g} \|) \| (\cN+1)^{1/2} \xi_1 \| \| (\cN+1)^{1/2} \xi_2 \|
\end{split} \end{equation}
for every, possibly time-dependent, $g \in L^2 (\bR^3)$. In fact, applying (\ref{eq:bg}) with $g= h_{N,t}$, we obtain the desired bounds for the first line on the r.h.s. of (\ref{eq:cE1-bd}). To bound the expectation of the operator on the second line on the r.h.s. of (\ref{eq:cE1-bd}), on the other hand, we apply (\ref{eq:bg}) with $g= \wt{h}_{N,t}$, $\xi_1 = \xi$ and $\xi_2 = e^{-B(\eta_t)} (\cN+1) e^{B(\eta_t)} \xi$. We find 
\begin{equation}\label{eq:cE1-fin} 
\begin{split}
&N^{-1/2} \left| \langle \xi , e^{-B(\eta_t)} (\cN+1) b (\wt{h}_{N,t}) e^{B(\eta_t)} \xi \rangle \right| \\ &\hspace{1cm} = N^{-1/2} \left| \langle \xi_2, e^{-B(\eta_t)} b (\wt{h}_{N,t}) e^{B(\eta_t)} \xi \rangle \right| \\ &\hspace{1cm} \leq N^{-1/2} \left| \langle \xi_2, \left[ b(\cosh_{\eta_t} (\wt{h}_{N,t})) + b^* (\sinh_{\eta_t} (\wt{h}_{N,t})) \right] \xi \rangle \right| \\ &\hspace{1.3cm} + C N^{-1} \| \wt{h}_{N,t} \| \| (\cN+1)^{1/2} \xi \| \| (\cN+1)^{1/2} \xi_2 \| \\
&\hspace{1cm}  \leq C N^{-1/2} \|  (\cN+1)^{1/2} \xi \| \| \xi_2 \| + C N^{-1}\| (\cN+1)^{1/2} \xi \| \| (\cN+1)^{1/2} \xi_2 \| \end{split}\end{equation}
where we used Lemma \ref{lm:bbds}, the fact that $\cosh_{\eta_t}, \sinh_{\eta_t}$ are bounded operators (uniformly in $t$ and $N$), and (\ref{eq:bd-hNt}). {F}rom Lemma \ref{lm:Npow}, we obtain
\[ \| \xi_2 \|^2 = \langle \xi, e^{-B(\eta_t)} (\cN+1)^2 e^{B(\eta_t)} \xi \rangle \leq C \langle \xi , (\cN+1)^2 \xi \rangle = C \| (\cN+1) \xi \|^2 \]
and, similarly,
\[ \begin{split} \| (\cN+1)^{1/2} \xi_2 \|^2 &= \langle \xi, e^{-B(\eta_t)} (\cN+1) e^{B(\eta_t)} (\cN+1) e^{-B(\eta_t)} (\cN+1) e^{B(\eta_t)} \xi \rangle \\ &\leq C \langle \xi , e^{-B(\eta_t)} (\cN+1)^3 e^{B(\eta_t)}  \xi \rangle \\ &\leq C \langle \xi,  (\cN+1)^3 \xi \rangle = C \| (\cN+1)^{3/2} \xi \|^2 \end{split} \]
Inserting the last two bounds in the r.h.s. of (\ref{eq:cE1-fin}), we conclude that 
\[  N^{-1/2} \left| \langle \xi , e^{-B(\eta_t)} (\cN+1) b (\wt{h}_{N,t}) e^{B(\eta_t)} \xi \rangle \right| \leq 
C \|  (\cN+1)^{1/2} \xi \|^2   \]
for all $\xi \in \cF^{\leq N}$. Similarly, we can control the commutator of the second line on the r.h.s. of (\ref{eq:cE1-bd}) with $\cN$ and with $a^* (g_1) a(g_2)$ and its time-derivative. 

We still have to show (\ref{eq:bg}). To this end, we use Lemma \ref{lm:conv-series} to expand 
\begin{equation}\label{eq:ser-bg} e^{-B(\eta_t)} b(g) e^{B(\eta_t)}  = \sum_{n \geq 0} \frac{(-1)^n}{n!} \text{ad}_{B(\eta_t)}^{(n)} (b(g)) \end{equation}
According to Lemma \ref{lm:indu}, the nested commutator $\text{ad}_{B(\eta_t)}^{(n)} (b(g))$ can be written as a sum of $2^n n!$ terms, having the form
\begin{equation}\label{eq:typadn} \Lambda_1 \dots \Lambda_i N^{-k} \Pi^{(1)}_{\sharp,\flat} (\eta^{(j_1)}_{t,\natural_1}, \dots , \eta^{(j_k)}_{t,\natural_k}  ; \eta^{(s)}_{t,\natural_{k+1}} (g_\Delta)) \end{equation}
where each $\Lambda_m$ is either $(N-\cN)/N$, $(N-\cN+1)/N$ or a $\Pi^{(2)}$-operator of the form
\begin{equation}\label{eq:Pi2-bg} N^{-p} \Pi^{(2)}_{\sharp',\flat'} (\eta^{(m_1)}_{t,\natural'_1}, \dots , \eta^{(m_p)}_{t,\natural'_p}) \end{equation}
Exactly one of these $2^n n!$ terms has the form
\begin{equation}\label{eq:chsh-term} \left\{ \begin{array}{ll} \frac{(N - \cN)^r}{N^r} \frac{(N+1-\cN)^r}{N^r} b (\eta_t^{(2r)} (g)) &\quad \text{if } n = 2r \text{ is even} \\
-\frac{(N - \cN)^{r+1}}{N^{r+1}} \frac{(N+1-\cN)^r}{N^r} b^* (\eta_t^{(2r+1)} (\bar{g})) &\quad \text{if } n = 2r +1 \text{ is odd} \end{array} \right. \end{equation}
All other terms are of the form (\ref{eq:typadn}), with either $k > 0$ or with at least one factor $\Lambda_i$ being of the form (\ref{eq:Pi2-bg}). Let us suppose that $n=2r$ is even. Then we write (\ref{eq:chsh-term}) as
\begin{equation}\label{eq:ch-term}
\begin{split}  \frac{(N - \cN)^r}{N^r} &\frac{(N+1-\cN)^r}{N^r} b (\eta_{t}^{(2r)} (g))  \\ &= b(\eta_{t}^{(2r)} (g)) + \left[ \frac{(N - \cN)^r}{N^r} \frac{(N+1-\cN)^r}{N^r} - 1 \right] b (\eta_{t}^{(2r)} (g)) \end{split} \end{equation}
Inserting the term $b (\eta^{(2r)}_t (g))$ on the r.h.s. of (\ref{eq:ser-bg}) and summing over all $r \in \bN$, we reconstruct 
\[ \sum_{r \geq 0} \frac{1}{(2r)!} b (\eta_t^{(2r)} (g)) = b (\cosh_{\eta_t} (g)) \]
On the other hand, the contribution of the second term on the r.h.s. of (\ref{eq:ch-term}) has matrix elements bounded by 
\begin{equation}\label{eq:ch-term2}
\begin{split} &\left| \langle \xi_1 , \left[ \frac{(N - \cN)^r}{N^r} \frac{(N+1-\cN)^r}{N^r} - 1 \right] b (\eta_{t}^{(2r)} (g)) \xi_2 \rangle \right| \\ &\hspace{2cm} \leq  \left\| \left[ \frac{(N - \cN)^r}{N^r} \frac{(N+1-\cN)^r}{N^r} - 1 \right] \xi_1 \right\| \| b (\eta^{(2r)}_t (g)) \xi_2 \| \\ &\hspace{2cm} \leq 2r N^{-1/2}  \| \eta_t \|^{2r} \| g \|  \| (\cN+1)^{1/2} \xi_1 \| \| (\cN+1)^{1/2} \xi_2 \| \end{split} \end{equation}
since $1-(1-x)^r \leq r x$ for all $0 \leq x \leq 1$. Similarly, the contribution (\ref{eq:chsh-term}) with $n=2r+1$ odd can be shown to reconstruct the operator $b^* (\sinh_{\eta_t} (\bar{g}))$, up to an error that can be estimated as in (\ref{eq:ch-term2}).

As for the other terms of the form (\ref{eq:typadn}), excluding (\ref{eq:chsh-term}), we can bound their matrix elements using part i) of Lemma \ref{lm:prelim}. We obtain
\begin{equation}\label{eq:case-Pi2} \begin{split} & \left| \langle \xi_1, \Lambda_1 \dots \Lambda_i N^{-k} \Pi^{(1)}_{\sharp,\flat} (\eta_{t,\natural_1}^{(j_1)}, \dots , \eta_{t,\natural_k}^{(j_k)} ; \eta_{t,\natural_{k+1}}^{(s)}) \xi_2 \rangle \right| \\ & \hspace{2cm} \leq \| (\cN+1)^{1/2} \xi_1 \| \left\| (\cN+1)^{-1/2} \, N^{-k} \Pi^{(1)}_{\sharp,\flat} (\eta_{t,\natural_1}^{(j_1)}, \dots , \eta_{t,\natural_k}^{(j_k)} ; \eta_{t,\natural_{k+1}}^{(s)} (g_\Delta)) \xi_2 \right\| \\ & \hspace{2cm} \leq C^n \| \eta_t \|^n N^{-1/2} 
 \| g \| \| (\cN+1)^{1/2} \xi_1 \| \| (\cN+1)^{1/2} \xi_2 \|  
\end{split} \end{equation}
We conclude that  
\begin{equation}\label{eq:ebe-} \begin{split} &\left| \langle \xi_1 , \left\{ e^{-B(\eta_t)} b(g) e^{B(\eta_t)} - b(\cosh_{\eta_t} (g)) - b^* (\sinh_{\eta_t}) (\bar{g})) \right\} \xi_2 \rangle  \right| \\ &\hspace{3cm} \leq N^{-1/2} \| g \| \| (\cN+1)^{1/2} \xi_1 \| \| (\cN+1)^{1/2}\xi_2 \|  \sum_{n \geq 2} n C^n \| \eta_t \|^n  \\ &\hspace{3cm} \leq C N^{-1/2} \| g \| 
\| (\cN+1)^{1/2} \xi_1 \| \| (\cN+1)^{1/2} \xi_2 \|  \end{split} \end{equation}
if the parameter $\ell > 0$ in the definition (\ref{eq:etat}) of the kernel $\eta_t$ is small enough. 

Since, by Lemma \ref{lm:prel1}, part i), the commutator of every term of the form (\ref{eq:typadn}) with $\cN$ is again a term of the same form, just multiplied with a constant $\kappa \in \{0, \pm 1 , \pm 2\}$, we conclude that
\begin{equation}
\label{eq:nebe-} \begin{split} &\left| \langle \xi_1 , \left[ \cN ,  \left\{ e^{-B(\eta_t)} b(g) e^{B(\eta_t)} - b(\cosh_{\eta_t} (g)) - b^* (\sinh_{\eta_t}) (\bar{g})) \right\} \right] \xi_2 \rangle  \right| \\ &\hspace{3cm} \leq C N^{-1/2} \| g \| \| (\cN+1)^{1/2} \xi_1 \| \| (\cN+1)^{1/2} \xi_2 \|  \end{split} \end{equation}
Since, again by Lemma \ref{lm:prel1}, part ii) and iii), the commutator of every term of the form (\ref{eq:typadn}) with $a^* (g_1) a(g_2)$ can be written as a sum of at most $2n$ terms having again the form (\ref{eq:typadn}), just with one of the $\eta_t$-kernels or with the function $\eta_{t,\natural_{k+1}}^{(s)} (g_\Delta)$ appearing in the $\Pi^{(1)}$-operator replaced according to (\ref{eq:hell-poss}) and (\ref{eq:hellk-poss}), we also find that
\begin{equation}\label{eq:aaebe-}  \begin{split} &\left| \langle \xi_1 , \left[ a^* (g_1) a(g_2) ,  \left\{ e^{-B(\eta_t)} b(g) e^{B(\eta_t)} - b(\cosh_{\eta_t} (g)) - b^* (\sinh_{\eta_t}) (\bar{g})) \right\} \right] \xi_2 \rangle  \right| \\ &\hspace{3cm} \leq C N^{-1/2} \| g \| \| g_1 \| \| g_2 \| \| (\cN+1)^{1/2} \xi_1 \| \| (\cN+1)^{1/2} \xi_2 \|  \end{split} \end{equation}
Finally, since by Lemma \ref{lm:prel1}, part iv), the time-derivative of each term of the form (\ref{eq:typadn}) can be written as a sum of at most $(n+1)$ terms having again the form (\ref{eq:typadn}), but with one of the $\eta_t$-kernels or the function $\eta_{t,\natural_{k+1}}^{(s)} (g_\Delta)$ appearing in the $\Pi^{(1)}$-operator replaced by their time-derivative, we get (since $\| \dot{\eta}_t \| \leq C e^{c|t|}$) 
\begin{equation}\label{eq:parebe-}  \begin{split} &\left|\partial_t \langle \xi_1 , \left[ e^{-B(\eta_t)} b(g) e^{B(\eta_t)} - b(\cosh_{\eta_t} (g)) - b^* (\sinh_{\eta_t}) (\bar{g})) \right] \xi_2 \rangle  \right| \\ &\hspace{3cm} \leq C N^{-1/2} e^{c|t|}  (\| g \| + \| \dot{g} \|) \| (\cN+1)^{1/2} \xi_1 \| \| (\cN+1)^{1/2} \xi_2 \|  \end{split} \end{equation}

\end{proof}

\subsection{Analysis of $e^{-B(\eta_t)} \cL_{N,t}^{(2)} e^{B(\eta_t)}$}
\label{sec:L2}

Recall that 
\begin{equation}\label{eq:L2-b} \begin{split} \cL_{N,t}^{(2)} = \; &\cK + \int dx dy \, N^3 V(N(x-y))|\wtph_t (y)|^2 \left[ b_x^* b_x - \frac{1}{N} a_x^* a_x \right]  \\ &
+ \int dx dy \, N^3 V(N (x-y)) \wtph_t (x) \bar{\wtph}_t (y) \left[ b_x^* b_y - \frac{1}{N} a_x^* a_y \right] \\ &+ \frac{1}{2} \int dx dy N^3 V(N(x-y)) \left[ \wtph_t (x) \wtph_t (y) b_x^* b_y^* + \text{h.c.} \right] \end{split} \end{equation}
with the notation 
\[ \cK = \int dx \, \nabla_x a_x^* \nabla_x a_x \]
for the kinetic energy operator. 

In the next two subsections we consider first the conjugation of the kinetic energy operator and then of the rest of $\cL_{N,t}^{(2)}$ with $e^{B(\eta_t)}$.

\subsubsection{Analysis of $e^{-B(\eta_t)} \cK e^{B(\eta_t)}$}

We write 
\begin{equation}\label{eq:conj-K} \begin{split} 
e^{-B(\eta_t)} \cK e^{B(\eta_t)} = &\; \cK + \int |\nabla_x k_t (x;y)|^2 dx dy \\ &+ 
\int dx dy \, (\Delta w_\ell) (N(x-y)) \, \left[ \wtph_t (x) \wtph_t (y) b_x^* b_y^* + \text{h.c.} \right]  \\ &+ \cE_{N,t}^{(K)} \end{split} 
\end{equation}
In the next proposition, we collect important properties of the error term $\cE_{N,t}^{(K)}$ defined in (\ref{eq:conj-K}). 
\begin{prop}\label{prop:cEK}
Under the same assumptions as in Theorem \ref{thm:gene}, there exist constants $C, c >0$ such that
\begin{equation}\label{eq:cEK} \begin{split} 
\left| \langle \xi, \cE^{(K)}_{N,t} \xi \rangle \right| &\leq C e^{c|t|} \| (\cH_N+\cN+1)^{1/2} \xi \| \| (\cN+1)^{1/2} \xi \| \\
\left| \langle \xi, \left[ \cN , \cE^{(K)}_{N,t} \right] \xi \rangle \right| &\leq C e^{c|t|} \| (\cH_N+\cN+1)^{1/2} \xi \| \| (\cN+1)^{1/2} \xi \| \\
\left| \langle \xi, \left[ a^* (g_1) a(g_2) , \cE^{(K)}_{N,t} \right]  \xi \rangle \right| &\leq C e^{c|t|} \| g_1 \|_{H^1} \| g_2 \|_{H^1} \| (\cH_N+\cN+1)^{1/2} \xi \| \| (\cN+1)^{1/2} \xi \| \\
\left| \partial_t \langle \xi, \cE^{(K)}_{N,t} \xi \rangle \right| &\leq C e^{c|t|} \| (\cH_N+\cN+1)^{1/2} \xi \| \| (\cN+1)^{1/2} \xi \| \end{split} \end{equation}
where we used the notation $\cH_N = \cK + \cV_N$, with 
\begin{equation}\label{eq:cVN} \cV_N = \frac{1}{2} \int dx dy N^2 V(N(x-y)) a_x^* a_y^* a_y a_x \end{equation}
\end{prop}

\begin{proof}
We write
\[ \begin{split} e^{-B(\eta_t)} \cK e^{B(\eta_t)} - \cK &= \int_0^1 e^{-sB(\eta_t)} \left[\cK, B (\eta_t) \right] e^{s B (\eta_t)} \\ &= \int_0^1 ds \int dx \, e^{-sB(\eta_t)} \left[ \nabla_x a^*_x \nabla_x a_x , B(\eta_t) \right] e^{sB(\eta_t)} \end{split} \]
{F}rom (\ref{2.2.Betacommutators}), we find 
\[ \begin{split} e^{-B(\eta_t)} & \cK e^{B(\eta_t)} - \cK \\ &= \int_0^1 ds \int dx \, \left[ e^{-sB(\eta_t)} b(\nabla_x \eta_{x}) \nabla_x b_x e^{sB(\eta_t)} + \text{h.c.} \right] \\ &= \sum_{k,n \geq 0} \frac{(-1)^{k+n}}{k! n! (k+n+1)} \int dx \, \left[ \text{ad}^{(n)}_{B (\eta_t)} (b(\nabla_x \eta_{x})) \text{ad}_{B(\eta_t)}^{(k)} (\nabla_x b_x) + \text{h.c.} \right] \end{split} \]
{F}rom the sum on the r.h.s. we extract the term with $k=n=0$ and also the term with $n=0,k=1$. We obtain
\begin{equation}\label{eq:sum'} \begin{split} e^{-B(\eta_t)} &\cK e^{B(\eta_t)} - \cK \\ =\; & \int dx \, \left[ b(\nabla_x \eta_{x}) \nabla_x b_x + \text{h.c.} \right] \\ &+  \int dx \, b(\nabla_x \eta_{x}) b^* (\nabla_x \eta_{x}) -\frac{1}{N} \int dx \, b(\nabla_x \eta_{x}) \cN   b^* ( \nabla_x \eta_{x}) \\
&-\frac{1}{2N} \int dx dz dy \left[ \eta_t (z,y) b(\nabla_x \eta_x) b_y^* a_z^* \nabla_x a_x + \text{h.c.} \right] \\ &+ \sum_{k,n}^* \frac{(-1)^{k+n}}{k! n! (k+n+1)} \int dx \left[ \text{ad}^{(n)}_{B (\eta_t)} (b(\nabla_x \eta_{x}) ) \text{ad}^{(k)}_{B(\eta_t)} (\nabla_x b_x) + \text{h.c.} \right] \end{split} \end{equation}
where $\sum^*$ denotes the sum over all indices $k,n \geq 0$, excluding the two pairs $(k,n) = (0,0)$ and $(k,n) = (1,0)$. We discuss now the terms on the r.h.s. of (\ref{eq:sum'}) separately. 

The first term on the r.h.s. of (\ref{eq:sum'}) can be decomposed as in (\ref{eq:mut}), giving  
\begin{equation}\label{eq:deco-mu} \int dx \, b(\nabla_x \eta_{x}) \nabla_x b_x = \int dx \, b(\nabla_x k_{x}) \nabla_x b_x + \int dx \, b(\nabla_x \mu_{x}) \nabla_x b_x \end{equation}
The second term on the r.h.s. of (\ref{eq:deco-mu})  contributes to the error $\cE^{(K)}_{N,t}$. Its expectation is bounded by
\[ \begin{split} \left| \int dx \langle \xi, b(\nabla_x \mu_{x}) \nabla_x b_x \xi \rangle \right| &\leq \| (\cN+1)^{1/2} \xi \| \int dx \, \| \nabla_x \mu_x \| \| \nabla_x b_x  \xi \|\\ &\leq \| \nabla_x \mu \| \| (\cN+1)^{1/2} \xi \| \| \cK^{1/2} \xi \| \leq C \| (\cN+1)^{1/2} \xi \| \| \cK^{1/2} \xi \| 
\end{split} \]
The expectation of the commutator of this term with $\cN$ and with $a^* (g_1) a(g_2)$ and also its time-derivative can be bounded similarly, using the formula  
\[ [a^* (g_1) a(g_2), b(\nabla_x \mu_{x}) \nabla_x b_x ] = \langle g_1, \nabla_x \mu_{x} \rangle b(g_2) \nabla_x b_x + b(\nabla_x \mu_{x}) \nabla g_1 (x) b(g_2) \]
and the fact that $\| \partial_t \nabla_x \mu_{t} \| < Ce^{c|t|}$, uniformly in $N$. 

As for the first term on the r.h.s. of (\ref{eq:deco-mu}), we integrate by parts and we use the definition (\ref{eq:ktdef}), to write 
\begin{equation}\label{eq:deco-k} \begin{split} \int dx \, b(\nabla_x k_{x}) \nabla_x b_x = \; &\int dx dy \, N^3 (\Delta w_\ell) (N(x-y)) \bar{\wtph}_t (x) \bar{\wtph}_t (y) \, b_x b_y \\ &+ 2 \int dx dy \, N^2 (\nabla w_\ell) (N(x-y)) (\nabla\wtph_t) (x) \wtph_t (y) \, b_x b_y \\ &+ \int dx dy \, N w_\ell (N(x-y)) (\Delta \wtph_t) (x) \wtph_t (y) b_x b_y \end{split} \end{equation}
The first term on the r.h.s. of (\ref{eq:deco-k}) is exactly the (hermitian conjugate of the) contribution that we isolated on the second line of (\ref{eq:conj-K}); it does not enter the error term $\cE^{(K)}_{N,t}$. The second and third terms on the r.h.s. of (\ref{eq:deco-k}), on the other hand, are included in $\cE^{(K)}_{N,t}$. The expectation of the third term is bounded by 
\begin{equation}\label{eq:third-Kbb} \begin{split} &\left| \int dx dy \, N w_\ell (N(x-y)) (\Delta\wtph_t) (x) \wtph_t (y) \, \langle \xi, b_x b_y \xi \rangle \right| \\ &\hspace{1cm} \leq    \int dx \, |\Delta\wtph_t (x)|  \| b^* (N w_\ell (N(x-.)) \wtph_t) \xi \| \, \| b_x \xi \| \\ &\hspace{1cm} \leq \sup_x \| N w_\ell (N(x-.)) \wtph_t \| \| \Delta \wtph_t \| \| (\cN+1)^{1/2} \xi \|^2 \leq C e^{c|t|} \| (\cN+1)^{1/2} \xi \|^2 \end{split} \end{equation}
To bound the expectation of the second term on the r.h.s. of (\ref{eq:deco-k}), we integrate by parts. We find
\[ \begin{split} \int dx dy \, N^2 (\nabla w_\ell) (N(x-y)) &(\nabla\wtph_t) (x) \wtph_t (y) \, \langle \xi, b_x b_y \xi \rangle \\ = \; &- \int dx dy N w_\ell (N(x-y)) (\Delta\wtph_t) (x) \wtph_t (y) \langle \xi, b_x b_y \xi \rangle  \\ &- \int dx dy N w_\ell (N(x-y)) (\nabla \wtph_t) (x) \wtph_t (y) \langle \xi, b_y \nabla_x b_x \xi \rangle \end{split} \]
Proceeding as in (\ref{eq:third-Kbb}), we conclude that 
\[ \begin{split} &\left| \int dx dy \, N^2 (\nabla w_\ell) (N(x-y)) (\nabla\wtph_t) (x) \wtph_t (y) \, \langle \xi, b_x b_y \xi \rangle \right| \\ &\hspace{.5cm} \leq  \sup_x \| N w_\ell (N(x-.)) \wtph_t \| \, \left[ \| \Delta \wtph_t \| \| (\cN+1)^{1/2} \xi \|^2 + \| \nabla \wtph_t \| \| (\cN+1)^{1/2} \xi \|  \| \cK^{1/2} \xi \|\right] \\ &\hspace{.5cm} \leq C e^{c|t|} \left[ \| (\cN+1)^{1/2} \xi \|^2 + \|  (\cN+1)^{1/2} \xi \| \| \cK^{1/2} \xi \| \right] \end{split} \]
Notice that the last estimate and the estimate (\ref{eq:third-Kbb}) for the third term on the r.h.s. of (\ref{eq:deco-k}) continue to hold, if we replace the operator whose expectation we are bounding, with its commutator with $\cN$ or with $a^* (g_1) a(g_2)$ or with its time-derivative. 

Now, let us consider the second term on the r.h.s. of (\ref{eq:sum'}). We observe that
\begin{equation}\label{eq:K-line2}
\begin{split}  \int dx \, b(\nabla_x \eta_{x}) b^* (\nabla_x \eta_{x}) = & \| \nabla_x \eta_{x} \|^2-\frac{\cN}N \| \nabla_x \eta_{x} \|^2\\&+ \int dx dy dz \nabla_x \eta_t (x;z) \nabla_x \bar{\eta}_t (y;x)\left( b^*_z b_y -\frac1N a^*_z a_y\right)\end{split} 
\end{equation}
Denoting by $D$ the operator with the integral kernel 
\begin{equation}\label{eq:Dde} D(z;y) = \int dx \, \nabla_x \eta_t (z;x) \nabla_x\bar{\eta}_t (x;y) \end{equation}
we have
\begin{equation}\label{eq:sec-nab} \left| \int dx dy dz \nabla_x \eta_t (x;z) \nabla_x  \bar{\eta}_t (y;x) \langle \xi, b^*_z b_y \xi \rangle \right| \leq \left| \langle \xi, d\Gamma (D) \xi \rangle \right| \leq \| D \|_2\| \cN^{1/2} \xi \|^2 
\end{equation}
Since, by Lemma \ref{lm:eta}, $\| D \|_2 \leq  C$, we obtain  
\[ \left| \int dx dy dz \, \nabla_x \eta_t (x;z) \nabla_x \bar{\eta}_t (y;x) \langle \xi, b^*_z b_y \xi \rangle \right|  \leq C \| \cN^{1/2} \xi \|^2 \]
and similarly for the $ a^*_za_y$ term. As for the first term on the r.h.s. of (\ref{eq:K-line2}), we decompose $\eta_t = k_t + \mu_t$. Since $\| \nabla_x \mu_t \|$ is finite, uniformly in $N$ and in $t$, we find
\[ \left| \int dx \| \nabla_x \eta_{x} \|^2 - \int dx dy \, |\nabla_x k_t (x;y)|^2 \right| \leq C  \]
The second term on the r.h.s. of (\ref{eq:K-line2}) can be controlled using $ N^{-1}\| \nabla_x \eta_{x} \|^2\leq C$. Furthermore, one can show that 
\[ \begin{split} 
\int dx \,  \langle \xi, [ \cN, b (\nabla_x \eta_{x})  b^* (\nabla_x \eta_{x})]  \xi \rangle  &= 0 \\ 
\left| \int dx \,  \langle \xi, [ a^* (g_1) a(g_2), b (\nabla_x \eta_{x})  b^* (\nabla_x \eta_{x})]  \xi \rangle \right| &\leq C \| g_1 \| \| g_2 \| \| (\cN+1)^{1/2} \xi \|^2 \end{split} \]
and
\[ \begin{split} 
&\Big| \partial_t \, \Big[ \int dx \,  \langle \xi, \partial_t [b (\nabla_x \eta_{x})  b^* (\nabla_x \eta_{x})]  \xi \rangle -   \int dx dy |\nabla_x k_t (x;y)|^2 \Big] \Big|  \leq C e^{K|t|} \| (\cN+1)^{1/2} \xi \|^2 
\end{split} \]
Here we used the formula
\[ \begin{split} & \left[ a^* (g_1) a(g_2) , \int dx \,  b (\nabla_x \eta_{x})  b^* (\nabla_x \eta_{x}) \right] \\ &\hspace{2cm} = \int dx \, \langle \nabla_x \eta_{x} , g_1 \rangle b(g_2) b^* (\nabla_x \eta_{x}) + \int dx , \langle g_2 , \nabla_x \eta_{x} \rangle  b (\nabla_x \eta_{x}) b^* (g_1) \end{split} \]
for the commutator with $a^* (g_1) a (g_2)$ and the bounds in Proposition \ref{prop:phph} for $\partial_t \wtph_t$. 

The third term on the r.h.s. of (\ref{eq:sum'}) can be controlled similarly. 


To control the fourth term on the r.h.s. of (\ref{eq:sum'}) we proceed as follows. First of all, we commute the annihilation operator $b(\nabla_x \eta_{x})$ to the right of the two creation operators $b_y^* a_z^*$. Using (\ref{eq:comm-b}), we find 
\begin{equation}\label{eq:comm01} \begin{split} \frac{1}{2N} \int dx dy dz &\, \eta_t (z;y) b(\nabla_x \eta_{x}) b_y^* a_z^* \nabla_x a_x \\
= \; &\frac{1}{2N} \int dx dy dz \, \eta_t (z;y) b_y^* a_z^* a(\nabla_x \eta_{x}) \nabla_x b_x \\
&+\frac{1}{N} \int dx dy dz \, \eta_t (z;y) \nabla_x \eta_t (x;y) \left(1 - \frac{\cN}{N} - \frac{1}{2N} \right) a_z^* \nabla_x a_x \\
&-\frac{1}{2N^2} \int dx dydz \, \eta_t (z;y) a_y^* a (\nabla_x \eta_{x}) a_z^* \nabla_x a_x \end{split} \end{equation}
To bound the expectation of the last term, we use the additional $N^{-1}$ factor to compensate for $\| \nabla_x \eta_t \| \simeq N^{1/2}$. We find
\[ \begin{split} &\left| \frac{1}{2N^2} \int dx dydz \, \eta_t (z;y) \langle \xi , a_y^* a (\nabla_x \eta_{x}) a_z^* \nabla_x a_x \xi \rangle \right| \\ &\hspace{1cm} \leq \frac{1}{2N^2} \left[ \int dx dy dz |\eta_t (y;z)|^2  \| \nabla_x a_x \xi \|^2 \right]^{1/2} \left[ \int dx dy dz \| a_z a^* (\nabla_x \eta_{x}) a_y \xi \|^2 \right]^{1/2} \\  
&\hspace{1cm} \leq \frac{\| \eta_t \| \| \nabla_x \eta_t \|}{2N^2}  \| \cK^{1/2} \xi \| \|(\cN+1)^{3/2} \xi \|\\ &\hspace{1cm} \leq C N^{-1/2} \|  \cK^{1/2} \xi \| \|(\cN+1)^{1/2} \xi \|
\end{split} \]
Similarly, the expectation of the second term on the r.h.s. of (\ref{eq:comm01}) is bounded by
\[ \begin{split} &\left| \frac{1}{N} \int dx dy dz \, \eta_t (z;y) \nabla_x \eta_t (x;y) \left\langle \xi, \left(1 - \frac{\cN}{N} - \frac{1}{2N} \right) a_z^* \nabla_x a_x \xi\right\rangle \right| 
 \\ &\hspace{1cm} \leq \frac{1}{N} \left[ \int dx dy dz |\eta_t (z;y)|^2 \| \nabla_x a_x \xi \|^2 \right]^{1/2}
 \left[ \int dx dy dz |\nabla_x \eta_t (x;y)||^2 \| a_z \xi \|^2 \right]^{1/2}  \\ &\hspace{1cm} \leq \frac{\| \eta_t \| \| \nabla_x \eta_t \|}{N} \| (\cN+1)^{1/2} \xi \| \| \cK^{1/2} \xi \| \\ &\hspace{1cm} \leq C N^{-1/2} \| (\cN+1)^{1/2} \xi \| \| \cK^{1/2} \xi \|
 \end{split} \]
We are left with the first term on the r.h.s. of (\ref{eq:comm01}). Here, we decompose 
\begin{equation}\label{eq:comm01-2} \begin{split} 
\frac{1}{2N} \int dx dy dz \, &\eta_t (z;y) b_y^* a_z^* a(\nabla_x \eta_{x}) \nabla_x b_x \\ = \; &  
\frac{1}{2N} \int dx dy dz \, \eta_t (z;y) b_y^* a_z^* a(\nabla_x k_{x}) \nabla_x b_x \\ &+ \frac{1}{2N} \int dx dy dz \, \eta_t (z;y) b_y^* a_z^* a(\nabla_x \mu_{x}) \nabla_x b_x =: \text{M}_1 + \text{M}_2 
\end{split} \end{equation}
Since $\nabla_x \mu_t \in L^2 (\bR^3 \times \bR^3)$, with norm bounded uniformly in $N$ and $t$, we easily find
\[ |\langle \xi, M_2 \xi \rangle| \leq C N^{-1/2} \| (\cN+1)^{1/2} \xi \| \| \cK^{1/2} \xi \| \]
To control the term $\text{M}_1$, on the other hand, we integrate by parts. We obtain
\begin{equation}\label{eq:comm01-3} \begin{split} 
\text{M}_1 = \; & \frac{1}{2N} \int dx dy dz dw \, \eta_t (z;y) (-\Delta_x k_{t}) (x;w) b_y^* a_z^* a_w b_x \\ = \; &\frac{N^2}{2} \int dx dy dz dw \, \eta_t (z;y) (\Delta w_\ell) (N(x-w)) \wtph_t (x) \wtph_t (w) b_y^* a_z^* a_w b_x \\ &+ \frac{N}{2} \int dx dy dz dw \eta_t (z;y) (\nabla w_{\ell}) (N(x-w)) \nabla \wtph_t (x) \wtph_t (w) b_y^* a_z^* a_w b_x \\ &+ \frac{1}{2} \int dx dy dz dw \, \eta_t (z;y) w_\ell (N(x-w)) \Delta \wtph_t (x) \wtph_t (w) b_y^* a_z^* a_w b_x \\
= \;& \text{M}_{11} + \text{M}_{12} + \text{M}_{13} 
\end{split} \end{equation}
Since $|(\nabla w_\ell) (Nx)| \leq C / (N^2 |x|^2)$, we have
\[ \begin{split} \left| \langle \xi, \text{M}_{12} \xi \rangle \right| &\leq C N^{-1} \int dx dy dz dw \, |\eta_t (z;y)| \frac{|\nabla \wtph_t (x)||\wtph_t (w)|}{|x-w|^2} \| a_z b_y \xi \| \| a_w b_x \xi \| \\ 
& \leq CN^{-1} \left[ \int dx dy dz dw \, \frac{|\nabla \wtph_t (x)|^2|\wtph_t (w)|^2}{|x-w|^2} \| a_z b_y \xi \|^2   \right]^{1/2}  \\ &\hspace{5cm} \times \left[ \int dx dy dz dw 
\, \frac{|\eta_t (y;z)|^2}{|x-w|^2}  \| a_w b_x \xi \|^2 \right]^{1/2} \\
&\leq C N^{-1} \| \eta_t \| \| (\cN+1) \xi \| \| (\cN+1)^{1/2} (\cK + \cN)^{1/2} \xi \| \\ &\leq C \| (\cN+1)^{1/2}\xi \| \| (\cK+\cN)^{1/2} \xi \| \end{split} \]
where we used Hardy's inequality $|x|^{-2} \leq C (1-\Delta)$. The expectation of $\text{M}_{13}$ can be bounded analogously. Let us focus now on the term $\text{M}_{11}$. Here we use the fact that $f_\ell = 1-w_\ell$ solves the Neumann problem (\ref{eq:scatl}) to write
\begin{equation}\label{eq:comm01-4} \begin{split} \text{M}_{11} = \; &-\frac{N^2}{2} \int dx dy dz dw \, \eta_t (z;y) V(N(x-w)) f_\ell (N(x-w)) \wtph_t (x) \wtph_t (w) b_y^* a_z^* a_w b_x \\ &+ N^2 \lambda_\ell \int dx dy dz dw \, \eta_t (z;y) f_\ell (N(x-w)) \chi (|x-w| \leq \ell)  \wtph_t (x) \wtph_t (w) b_y^* a_z^* a_w b_x \\ =: \; &\text{M}_{111} + \text{M}_{112} \end{split} \end{equation}
Since, by Lemma \ref{3.0.sceqlemma}, $\lambda_\ell \leq CN^{-3}$ and $0 \leq f_\ell \leq 1$, it is easy to check that 
\[ |\langle \xi , \text{M}_{112} \xi \rangle| \leq C \| (\cN+1)^{1/2} \xi \|^2 \]
As for the first term on the r.h.s. of (\ref{eq:comm01-4}), it can be estimated by 
\[ \begin{split} 
\left|  \langle \xi, \text{M}_{111} \xi \rangle \right|  \leq \; & \int dx dy dz dw |\eta_t (z;y)| N^2 V(N (x-w)) |\wtph_t (w)| |\wtph_t (x)| \| a_z b_y \xi \| \| a_w b_x \xi \|
\\  \leq \; & \left[ \int dx dy dz dw |\eta_t (z;y)|^2 N^2 V(N (x-w))  \| a_w b_x \xi \|^2 \right]^{1/2}\\ &\hspace{1.5cm} \times  \left[ \int dx dy dz dw \, N^2 V(N (x-w)) |\wtph_t (w)|^2 |\wtph_t (x)|^2 \| a_z b_y \xi \|^2 \right]^{1/2} \\  \leq \; & C N^{-1/2} \| \eta_t \| \| \cV_N^{1/2}  \xi \| \| (\cN+1) \xi \| \leq C \| \cV_N^{1/2} \xi \| \| (\cN+1)^{1/2} \xi \| \end{split} \]  
where we used the fact that $0\leq f_\ell \leq 1$ and the notation (\ref{eq:cVN}). 

Summarizing, we have shown that the expectation of the fourth term on the r.h.s. of (\ref{eq:sum'}) can be bounded by
\begin{equation}\label{eq:case6-pro} \left| \frac{1}{2N} \int dx dy dz \, \eta_t (y;z) \langle \xi, b(\nabla_x \eta_{x}) b_y^* a_z^* \nabla_x a_x \xi \rangle  \right| \leq C \| (\cN+1)^{1/2} \xi \| \| (\cK+\cN+\cV_N+1)^{1/2} \xi \| \end{equation}
Also in this case, it is also easy to check that the same estimate holds true for the expectation of the commutator of the fourth term on the r.h.s. of (\ref{eq:sum'}) with $\cN$ and with $a^* (g_1) a(g_2)$ and for the expectation of its time-derivative. 

Finally, we have to deal with the last term on the r.h.s. of (\ref{eq:sum'}). According to Lemma \ref{lm:indu}, the operator
\[ \int dx \, \text{ad}^{(n)}_{B(\eta_t)} (b(\nabla_x \eta_{x})) \text{ad}^{(k)}_{B(\eta_t)} (\nabla_x b_x) \]
is given by the sum of $2^{n+k} n! k!$ terms, all having the form
\begin{equation}\label{eq:typbnab}
\begin{split} 
\text{E} := \int dx \, \Lambda_1 \dots \Lambda_{i_1} &N^{-k_1} \Pi^{(1)}_{\sharp,\flat} (\eta_{t,\natural_1}^{(j_1)}, \dots , \eta_{t,\natural_{k_1}}^{(j_{k_1})} ; \nabla_x \eta^{(\ell_1 + 1)}_{x,\lozenge}) \\ &\times \Lambda'_1\dots \Lambda'_{i_2} N^{-k_2} \Pi^{(1)}_{\sharp',\flat'} (\eta^{(m_1)}_{t,\natural'_1} , \dots , \eta^{(m_{k_2})}_{t,\natural'_{k_2}} ; \nabla_x \eta_{x,\lozenge'}^{(\ell_2)} ) \end{split} 
\end{equation}
with $k_1, k_2, \ell_1, \ell_2 \geq 0$, $j_1, \dots , j_{k_1}, m_1, \dots , m_{k_2} \geq 1$, and where each operator $\Lambda_i$ or $\Lambda'_i$ is either a factor $(N-\cN)/N$, $(N+1-\cN)/N$ or a $\Pi^{(2)}$-operator of the form
\begin{equation}\label{eq:Pi2-nab} N^{-p} \Pi^{(2)}_{\underline{\sharp}, \underline{\flat}} (\eta_{t,\underline{\natural}_1}^{(q_1)}, \dots , \eta_{t,\underline{\natural}_p}^{(q_p)}) \end{equation}
with $p, q_1, \dots , q_p \geq 1$. Here we used the fact that $\eta^{(\ell_1)}_{\natural} (\nabla_x \eta_{x,\lozenge}) = \nabla_x \eta^{(\ell_1+1)}_{x,\lozenge'}$ for an appropriate choice of $\lozenge' \in \{ \cdot, * \}^{\ell_1 + 1}$. 

We study the expectation of a term of the form  (\ref{eq:typbnab}), distinguishing several cases, depending on the values of $\ell_1, \ell_2 \in \bN$.

{\it Case 1}: $\ell_1 \geq 1$, $\ell_2 \geq 2$. In this case, $\nabla_x \eta^{(\ell_1 + 1)}_{t,\lozenge}, \nabla_x \eta^{(\ell_2)}_{t,\lozenge} \in L^2 (\bR^3 \times \bR^3)$, with norm bounded uniformly in $N$ and $t$. Hence, with Lemma \ref{lm:Pi-bds}, we can bound 
\[ |\langle\xi, \text{E} \xi \rangle| \leq C^{k+n} \| \eta_t \|^{k+n-\ell_1 - \ell_2} \| \nabla_x \eta_t^{(\ell_1+1)} \| \| \nabla_x \eta_t^{(\ell_2)} \| \| (\cN+1)^{1/2} \xi \|^2 \]
Now we observe that, for example, 
\[ \| \nabla_x \eta_t^{(\ell_2)} \| \leq \| \nabla_x \eta_t^{(2)} \| \| \eta_t^{(\ell_2-2)} \| \leq \| \nabla_x \eta_t^{(2)} \| \| \eta_t \|^{\ell_2-2} \leq C \| \eta_t \|^{\ell_2-2} \]
Similarly, $\| \nabla_x \eta_t^{(\ell_1+1)} \| \leq C \| \eta_t \|^{\ell_1 - 1}$. Hence, in this case,
\[ |\langle \xi, \text{E} \xi \rangle| \leq C^{k+n} \| \eta_t \|^{k+n-3} \| (\cN+1)^{1/2} \xi \|^2 \, . \]

{\it Case 2}: $\ell_1 \geq 1$, $\ell_2 = 1$. In this case we integrate by parts, writing
\[ \begin{split} \langle \xi ,\text{E} \xi \rangle = \; &\int dx \, \langle \xi, \Lambda_1 \dots \Lambda_{i_1} N^{-k} \Pi^{(1)}_{\sharp,\flat} (\eta_{t,\natural_1}^{(j_1)}, \dots , \eta_{t,\natural_{k_1}}^{(j_{k_1})} ; -\Delta_x \eta_{x,\lozenge}^{(\ell_1+1)}) \\
&\hspace{3cm} \times \Lambda'_1 \dots \Lambda'_{i_2} N^{-k_2} \Pi^{(1)}_{\sharp',\flat'} (\eta^{(m_1)}_{t,\natural'_1}, \dots , \eta^{(m_{k_2})}_{t,\natural'_{k_2}} ; \eta_{x,\lozenge'}) \xi \rangle \end{split}  \]
Since, by Lemma \ref{lm:eta}, $\| \Delta_x \eta^{(2)}_t \| \leq C e^{c|t|}$, we conclude by Lemma \ref{lm:Pi-bds} that, in this case, 
\[ |\langle \xi, \text{E} \xi \rangle | \leq C^{k+n} \| \eta_t \|^{k+n-1} \| \Delta_x \eta^{(2)}_t \| \| (\cN+1)^{1/2} \xi \|^2 \leq C^{k+n} e^{c|t|} \| \eta_t \|^{k+n-1} \| (\cN+1)^{1/2} \xi \|^2 \, . \]
 
{\it Case 3:} $\ell_1 \geq 1$, $\ell_2 = 0$. In this case, the second $\Pi^{(1)}$-operator in (\ref{eq:typbnab}) has the form
\[ \begin{split}  N^{-k_2} \Pi^{(1)}_{\sharp',\flat'} (\eta^{(m_1)}_{t,\natural'_1}, \dots , &\eta^{(m_{k_2})}_{t,\natural'_{k_2}} ; \nabla_x \delta_x) \\ &= N^{-k_2} \int b^{\flat_0}_{x_1} \prod_{j=1}^{k_2-1} a^{\sharp_j}_{y_j} a^{\flat_j}_{x_{j+1}} a_{y_{k_2}}^{\sharp_{k_2}} \nabla_x a_x \, \prod_{j=1}^{k_2} \eta^{(m_j)}_{t,\natural'_j} (x_j ; y_j) dx_j dy_j \end{split} \]
Here we used part v) of Lemma \ref{lm:indu} to conclude that the last field on the right, the one carrying the derivative, must be an annihilation operator (or possibly a $b$-operator). Repeatedly applying Lemma~\ref{lm:Abds} on pairs of creation and annihilation operators, but leaving the last annihilation operator $\nabla_x a_x$ untouched, we find 
\[ \begin{split} |\langle \xi, E \xi \rangle | &\leq C^{k+n} \| \eta_t \|^{k+n-\ell_1} \| (\cN+1)^{1/2} \xi \| \int dx \, \| \nabla_x \eta_{x}^{(\ell_1+1)} \|  \| \nabla_x a_x \xi \| \\
&\leq C^{k+n} \| \eta_t \|^{k+n-\ell_1} \| \nabla_x \eta_t^{(\ell_1+1)} \| \| (\cN+1)^{1/2} \xi \| \| \cK^{1/2} \xi \| \\ 
&\leq C^{k+n} \| \eta_t \|^{k+n-1} \| (\cN+1)^{1/2} \xi \| \| \cK^{1/2} \xi \| \end{split} \]

{\it Case 4:} $\ell_1 = 0$, $\ell_2 \geq 2$. Here we proceed as in Case 2, integrating by parts and moving the derivative over $x$ from $\nabla_x \eta_{x,\lozenge}$ (whose $L^2$ norm blows up) to $\nabla_x \eta^{(\ell_2)}_{x,\lozenge'}$ (using the fact that $\| \Delta_x \eta_t^{(2)} \| < \infty$). 

{\it Case 5:} $\ell_1 = 0$, $\ell_2 = 1$. In this case, by part v) of Lemma \ref{lm:indu}, the two $\Pi^{(1)}$-operators in (\ref{eq:typbnab}) have the form
\begin{equation}\label{eq:Pi1-1} \Pi^{(1)}_{\sharp,\flat} (\eta_{t,\natural_1}^{(j_1)}, \dots , \eta_{t,\natural_{k_1}}^{(j_{k_1})} ; \nabla_x \eta^{(\ell_1 + 1)}_{x,\lozenge})  = \int b_{x_1}^{\flat_0} \prod_{i=1}^{k_1} a_{y_j}^{\sharp_j} a_{x_{j+1}}^{\flat_j} a_{y_n}^{\sharp_n} a (\nabla_x \eta_x) \prod_{i=1}^{k_1} \eta^{(j_i)}_{t,\natural_i} (x_i ; y_i) dx_i dy_i  \end{equation}
and 
\begin{equation}\label{eq:Pi1-2} \Pi^{(1)}_{\sharp',\flat'} (\eta^{(m_1)}_{t,\natural'_1} , \dots , \eta^{(m_{k_2})}_{t,\natural'_{k_2}} ; \nabla_x \eta_{x,\lozenge'}^{(\ell_2)} ) = \int b_{x_1}^{\flat'_0} \prod_{j=1}^{k_2} a_{y_j}^{\sharp'_j} a_{x_{j+1}}^{\flat'_j} a_{y_n}^{\sharp'_n} a^* (\nabla_x \eta_x) \prod_{i=1}^{k_2} \eta^{(m_i)}_{t,\natural_i} (x_i ; y_i) dx_i dy_i  \end{equation}
Since $\| \nabla_x \eta_t \| \simeq N^{1/2}$ blows up as $N \to\infty$, to estimate (\ref{eq:typbnab}) in this case we first have to commute the annihilation operator $a(\nabla_x \eta_{x,\lozenge})$ in (\ref{eq:Pi1-1}) with the creation operator $a^* (\nabla_x \eta_{x,\lozenge'})$ in (\ref{eq:Pi1-2}). We proceed similarly as we did to bound the second term on the r.h.s. of (\ref{eq:sum'}) in the case $n=0$, $k=1$, starting in (\ref{eq:K-line2}). Here, however, we first have to commute the annihilation operator $a(\nabla_x \eta_{x,\lozenge})$ through the $\Lambda'_i$ operators and through the creation operators in (\ref{eq:Pi1-2}). 

If $\Lambda'_i = (N-\cN)/N$ or $\lambda'_i=(N+1-\cN)/N$, we just pull the annihilation operator $a(\nabla_x \eta_{x,\lozenge})$ through, using the fact that $a(\nabla_x \eta_{x,\lozenge}) \cN = (\cN+1) a(\nabla_x \eta_{x,\lozenge})$. On the other hand, to commute $a(\nabla_x \eta_{x,\lozenge})$ through the $\Lambda'_i$ operators having the form (\ref{eq:Pi2-nab}) and through the creation operators in (\ref{eq:Pi1-2}) (excluding the very last one on the right), we use the canonical commutation relations (\ref{eq:ccr}). The important observation here is the fact that every creation operator appearing in (\ref{eq:Pi2-nab}) and in (\ref{eq:Pi1-2}) is associated with an $\eta_t$-kernel; the commutator produces a new creation or annihilation operator, this time with a wave function whose $L^2$-norm remains bounded, uniformly in $N$. For example, we have  
\begin{equation}\label{eq:comm-aPi}
\begin{split}  \left[ a(\nabla_x \eta_x) , \int a^*_{x_i} a_{y_i} \eta^{(m_i)} (x_i ; y_i) dx_i dy_i \right] &= a (\nabla_x \eta^{(m_i+1)}_x) \end{split} \end{equation}
Since $m_i + 1 \geq 2$, $\| \nabla_x \eta^{(m_i+1)} \| \leq C$, uniformly in $N$. Similar formulas hold for commutators of $a(\nabla_x \eta_x)$ with a pair of not normally ordered creation and annihilation operators or with the product of two creation operators. 
In fact, not only the $L^2$-norm but even the $H^1$-norm of the wave function of the annihilation operator on the r.h.s. of (\ref{eq:comm-aPi}) is bounded, uniformly in 
$N$. This means that terms resulting from commutators like (\ref{eq:comm-aPi}) can be bounded integrating by parts and moving the derivative in (\ref{eq:Pi1-2}) to the argument of the annihilation operator in (\ref{eq:comm-aPi}). We conclude that 
$E = \text{F}_1 + \text{F}_2$, where
\[ \begin{split} \text{F}_1 = \; &\int dx \,  \Lambda_1 \dots \Lambda_{i_1} N^{-k_1} \int b^{\flat_0}_{x_1} \prod_{i=1}^{k_1} a_{y_j}^{\sharp_j} a_{x_{j+1}}^{\flat_j} a_{y_n}^{\sharp_n} \prod_{i=1}^{k_1} \eta_{t,\natural_i}^{(j_i)} (x_i ; y_i) dx_i dy_i \\ & \hspace{1cm} \times \Lambda'_1 \dots \Lambda'_{i_2} N^{-k_2} \int b^{\flat'_0}_{x'_1} \prod_{i=1}^{k_1} a_{y'_j}^{\sharp'_j} a_{x'_{j+1}}^{\flat'_j} a_{y'_n}^{\sharp'_n} \prod_{i=1}^{k_1} \eta_{t,\natural_i}^{(j_i)} (x'_i ; y'_i) dx'_i dy'_i   
\\ & \hspace{7cm} \times a(\nabla_x \eta_{x,\lozenge}) a^* (\nabla_x \eta_{x,\lozenge'})  \end{split} \]
while $\text{F}_2$, which contains the contribution of all commutators, is bounded by  
\[ |\langle \xi , \text{F}_2 \xi \rangle | \leq n C^{k+n} \| \eta_t \|^{k+n-1} \| (\cN+1)^{1/2} \xi \|^2 \]
To estimate $\text{F}_1$, we write it as $\text{F}_1 = \text{F}_{11} + \text{F}_{12}$, with
\begin{equation}\label{eq:f11} \begin{split} \text{F}_{11} = \; &\| \nabla_x \eta_t \|^2 \,  \Lambda_1 \dots \Lambda_{i_1} N^{-k_1} \int b^{\flat_0}_{x_1} \prod_{i=1}^{k_1} a_{y_j}^{\sharp_j} a_{x_{j+1}}^{\flat_j} a_{y_n}^{\sharp_n} \prod_{i=1}^{k_1} \eta_{t,\natural_i}^{(j_i)} (x_i ; y_i) dx_i dy_i \\ & \hspace{1cm} \times \Lambda'_1 \dots \Lambda'_{i_2} N^{-k_2} \int b^{\flat'_0}_{x'_1} \prod_{i=1}^{k_1} a_{y'_j}^{\sharp'_j} a_{x'_{j+1}}^{\flat'_j} a_{y'_n}^{\sharp'_n} \prod_{i=1}^{k_1} \eta_{t,\natural_i}^{(j_i)} (x'_i ; y'_i) dx'_i dy'_i  \end{split} \end{equation}
and 
\begin{equation}\label{eq:f12} \begin{split} \text{F}_{12} = \; &\int dx \, \Lambda_1 \dots \Lambda_{i_1} N^{-k_1} \int b^{\flat_0}_{x_1} \prod_{i=1}^{k_1} a_{y_j}^{\sharp_j} a_{x_{j+1}}^{\flat_j} a_{y_n}^{\sharp_n} \prod_{i=1}^{k_1} \eta_{t,\natural_i}^{(j_i)} (x_i ; y_i) dx_i dy_i \\ &\hspace{1cm} \times \Lambda'_1 \dots \Lambda'_{i_2} N^{-k_2} \int b^{\flat'_0}_{x'_1} \prod_{i=1}^{k_1} a_{y'_j}^{\sharp'_j} a_{x'_{j+1}}^{\flat'_j} a_{y'_n}^{\sharp'_n} \prod_{i=1}^{k_1} \eta_{t,\natural_i}^{(j_i)} (x'_i ; y'_i) dx'_i dy'_i \\ &\hspace{7cm} \times a^* (\nabla_x \eta_{x,\lozenge'}) a (\nabla_x \eta_{x,\lozenge})  \end{split} 
\end{equation}
The contribution $\text{F}_{11}$ can be estimated by
\begin{equation}\label{eq:f11-alpha} |F_{11}| \leq C^{k+n} \| \eta_t \|^{k+n-1} \| \nabla_x \eta_t \|^2 N^{-\alpha} \| (\cN+1)^{\alpha/2} \xi \|^2 \end{equation}
where $\alpha = k_1 + p_1 + \dots + p_r + k_2 + p'_1 + \dots + p'_{r'}$, if $r$ of the operators $\Lambda_1,\dots , \Lambda_{i_1}$ and $r'$ of the operators $\Lambda'_1, \dots, \Lambda'_{i_2}$ are $\Pi^{(2)}$-operators of the form (\ref{eq:Pi2-nab}), with orders $p_1, \dots , p_r > 0$ and, respectively, $p'_1,\dots ,p'_{r'} > 0$. Now observe that, since $\ell_2=1$, we must have $k \geq 1$. Since we are excluding here the case $n=0, k=1$, we must either have $n \geq 1$ and $k=1$, or $k \geq 2$. In both cases $k+n \geq 2$. According to Lemma \ref{lm:indu}, the total number of $\eta_t$-kernels in every term of the form (\ref{eq:typbnab}) is equal to $k+n+1 \geq 3$. This implies that there is at least one $\eta_t$-kernel, additional to the two $\eta_t$-kernels which produced the commutator $\| \nabla_x \eta_t \|^2$ in (\ref{eq:f11}). We conclude that, in (\ref{eq:f11-alpha}), we have $\alpha \geq 1$, and therefore, on $\cF^{\leq N}$,
\[ |\text{F}_{11}| \leq C^{k+n} \| \eta_t \|^{k+n-1} \| \nabla_x \eta_t \|^2 N^{-1} \| (\cN+1)^{1/2} \xi \|^2 \leq C^{k+n} \| \eta_t \|^{k+n-1} \|  (\cN+1)^{1/2} \xi \|^2 \]
since $\| \nabla_x \eta_t \|^2 \leq C N$ by Lemma \ref{lm:eta}. To control $\text{F}_{12}$ we notice that, with the operator $D$ defined in (\ref{eq:Dde}),  
\[ 0 \leq \int dx\, a^*(\nabla_x \eta_{x,\lozenge'}) a(\nabla_x \eta_{x,\lozenge}) = d\Gamma (D) \leq \| D \|_2 \cN \leq C \cN \]
This easily implies that 
\[ |\langle \xi, \text{F}_{12} \xi \rangle| \leq C^{k+n} \| \eta_t \|^{k+n-1} \| (\cN+1)^{1/2} \xi \|^2 \]
We conclude that, in this case,
\[ |\langle \xi , \text{E} \xi \rangle | \leq n C^{k+n} \| \eta_t \|^{k+n-1} \| (\cN+1)^{1/2} \xi \|^2 \]

{\it Case 6:} $\ell_1 = 0$, $\ell_2 = 0$. In this case, the term (\ref{eq:typbnab}) has the form
\begin{equation}\label{eq:case6} \begin{split} \text{E} = \; &\int dx \, \Lambda_1 \dots \Lambda_{i_1} N^{-k_1} \int b^{\flat_0}_{x_1} \prod_{i=1}^{k_1} a_{y_i}^{\sharp_i} a_{x_{i+1}}^{\flat_i} a_{y_n}^{\sharp_n} a (\nabla_x \eta_{x,\lozenge}) \prod_{i=1}^{k_1} \eta^{(j_i)} (x_i ; y_i) dx_i dy_i  \\ & \hspace{1cm} \times \Lambda'_1 \dots \Lambda'_{i_2} N^{-k_2} \int b^{\flat'_0}_{x'_1} \prod_{i=1}^{k_1} a_{y'_i}^{\sharp'_i} a_{x'_{i+1}}^{\flat'_i} a_{y'_n}^{\sharp'_n} \nabla_x a_x 
\prod_{i=1}^{k_2} \eta^{(m_i)} (x'_i ; y'_i) dx'_i dy'_i 
\end{split} \end{equation}
Notice that a term of this form (with $n=0$ and $k=1$) already appears in the fourth line of (\ref{eq:sum'})
and was studied starting in (\ref{eq:comm01}) (to be more precise, in this case the first $\Pi^{(1)}$-operator in (\ref{eq:typbnab}) is of order zero (for $n=0$, there is no other choice), and therefore the operator $a(\nabla_x \eta_{x,\lozenge})$ appearing in (\ref{eq:case6}) is replaced by $b(\nabla_x \eta_{x,\lozenge})$). We will bound (\ref{eq:case6}) following the same strategy used in (\ref{eq:comm01}). First we have to commute the operator $a(\nabla_x \eta_{x,\lozenge})$ in (\ref{eq:case6}) to the right, close to the 
$\nabla_x a_x$ operator. As already explained in Case 5, the annihilation and creation operators produced while commuting $a(\nabla_x \eta_{x,\lozenge})$ to the right will have wave function with $H^1$-norm bounded, uniformly in $N$. Integrating by parts over $x$, we obtain $\text{E} = \text{G}_1 + \text{G}_2$, with 
\[ \begin{split} \text{G}_1 = \; &\int dx \, \Lambda_1 \dots \Lambda_{i_1} N^{-k_1} \int b^{\flat_0}_{x_1} \prod_{i=1}^{k_1} a_{y_j}^{\sharp_j} a_{x_{j+1}}^{\flat_j} a_{y_n}^{\sharp_n} \prod_{i=1}^{k_1} \eta_{t,\natural_i}^{(j_i)} (x_i ; y_i) dx_i dy_i \\ & \hspace{.5cm} \times \Lambda'_1 \dots \Lambda'_{i_2} N^{-k_2} \int b^{\flat'_0}_{x'_1} \prod_{i=1}^{k_1} a_{y'_j}^{\sharp'_j} a_{x'_{j+1}}^{\flat'_j} a_{y'_n}^{\sharp'_n} \prod_{i=1}^{k_1} \eta_{t,\natural_i}^{(j_i)} (x'_i ; y'_i) dx'_i dy'_i \, a(\nabla_x \eta_{x,\lozenge}) \nabla_x a_x \end{split} \]
and 
\[ |\langle \xi , \text{G}_2 \xi \rangle | \leq n C^{k+n} \| \eta_t \|^{k+n-1} \| (\cN+1)^{1/2} \xi \|^2 \]
To bound $G_1$, we proceed exactly as we did starting in (\ref{eq:comm01-2}). Decomposing $\eta_t = \mu_t + k_t$, and using the fact that $\nabla_x \mu_t$ has bounded $L^2$-norm, uniformly in $N$, we conclude that $\text{G}_1 = \text{G}_{11} + \text{G}_{12}$, with
\[  \begin{split} \text{G}_{11} = \; & \Lambda_1 \dots \Lambda_{i_1} N^{-k_1} \int b^{\flat_0}_{x_1} \prod_{i=1}^{k_1} a_{y_j}^{\sharp_j} a_{x_{j+1}}^{\flat_j} a_{y_n}^{\sharp_n} \prod_{i=1}^{k_1} \eta_{t,\natural_i}^{(j_i)} (x_i ; y_i) dx_i dy_i \\ & \hspace{1cm} \times \Lambda'_1 \dots \Lambda'_{i_2} N^{-k_2} \int b^{\flat'_0}_{x'_1} \prod_{i=1}^{k_1} a_{y'_j}^{\sharp'_j} a_{x'_{j+1}}^{\flat'_j} a_{y'_n}^{\sharp'_n} \prod_{i=1}^{k_1} \eta_{t,\natural_i}^{(j_i)} (x'_i ; y'_i) dx'_i dy'_i  \\ &\hspace{1cm} \times  \int dx \,  (-\Delta_x k_t) (x;y)  \, a_x a_y  \end{split} \]
and 
\[ |\langle \xi , \text{G}_{12} \xi \rangle | \leq C^{k+n} \| \eta_t \|^{k+n-1} \|(\cN+1)^{1/2} \xi \| \| \cK^{1/2} \xi \| \]
By Cauchy-Schwarz, the term $G_{11}$ is bounded by
\begin{equation}\label{eq:G11} |\langle \xi , \text{G}_{11} \xi \rangle | \leq C^{k+n} \| \eta_t \|^{k+n-1} N^{-\alpha}\| (\cN+1)^{\alpha} \xi \|  \int dx dy |\Delta_x k_t (x;y)| \| a_x a_y \xi \| \end{equation}
where $\alpha = k_1 + p_1 +\dots +p_r +k_2 +p'_1+ \dots + p'_{r'}$, if $r$ of the operators $\Lambda_1,\dots , \Lambda_{i_1}$ and $r'$ of the operators $\Lambda'_1, \dots , \Lambda'_{i_2}$ are $\Pi^{(2)}$-operators of the form (\ref{eq:Pi2-nab}), with orders $p_1, \dots, p_r > 0$ and, respectively, $p'_1, \dots ,p'_r > 0$. The important observation now is that, since we excluded the case $k=n=0$, we have $k+n \geq 1$, and therefore every term of the form (\ref{eq:typbnab}) must have at least two $\eta_t$-kernels in it. This implies that, in (\ref{eq:G11}), $\alpha \geq 1$, and therefore that
\[ |\text{G}_{11}| \leq C^{k+n} \| \eta_t \|^{k+n-1} N^{-1/2} \| (\cN+1)^{1/2} \xi \| \int dx dy |\Delta_x k_t (x;y)| \| a_x a_y \xi \| \]
Proceeding as we did from (\ref{eq:comm01-3}) to (\ref{eq:case6-pro}), we conclude that 
\[ |\text{G}_{11}| \leq C^{k+n} \| \eta_t \|^{k+n-1} \| (\cN+1)^{1/2} \xi \| \| (\cH_N+ \cN + 1)^{1/2} \xi \| \]

Summarizing, we proved that the last term on the r.h.s. of (\ref{eq:sum'}) is a sum over all $(k,n) \not = (0,0), (1,0)$ of $2^{n+k} n!k!$ terms of the form (\ref{eq:typbnab}), each of them having expectation bounded by
\[ |\langle \xi, \text{E} \xi \rangle | \leq C^{k+n} e^{c|t|} \| \eta_t \|^{\max (0, k+n-3)} \| (\cN+1)^{1/2} \xi \| \| (\cH_N + \cN+1)^{1/2} \xi \| \]
Similarly, one can show that 
\[ \begin{split}
|\langle \xi, [\cN, \text{E}] \xi \rangle | &\leq C^{k+n} e^{c|t|}\| \eta_t \|^{\max (0, k+n-3)} \| (\cN+1)^{1/2} \xi \| \| (\cH_N + \cN+1)^{1/2} \xi \| \\
|\langle \xi, [a^* (g_1) a(g_2) , \text{E} ] \xi \rangle | &\leq C^{k+n} e^{c|t|}\| \eta_t \|^{\max (0, k+n-3)} \| g_1 \|_{H^1} \| g_2 \|_{H^1} \\ &\hspace{3cm} \times  \| (\cN+1)^{1/2} \xi \| \| (\cH_N + \cN+1)^{1/2} \xi \| \\
|\langle \xi, \partial_t [\text{E}] \xi \rangle | &\leq C^{k+n}e^{c|t|} \| \eta_t \|^{\max (0, k+n-3)} \| (\cN+1)^{1/2} \xi \| \| (\cH_N + \cN+1)^{1/2} \xi \| \end{split} \]

Inserting in (\ref{eq:sum'}) we conclude that, if $\sup_{t \in \bR} \| \eta_t \|$ is small enough, the operator $\cE^{(K)}_{N,t}$ defined in (\ref{eq:conj-K})
satisfies the bounds in (\ref{eq:cEK}).
\end{proof}

\subsubsection{Analysis of $e^{-B(\eta_t)} (\cL^{(2)}_{N,t} - \cK) e^{B(\eta_t)}$}

Recall that 
\begin{equation}\label{eq:L2-K} \begin{split} \cL^{(2)}_{N,t} - \cK = \; &\int dx (N^3 V(N.) *|\wtph_t|^2)(x) \left[ b_x^* b_x - N^{-1} a_x^* a_x \right] \\ &+ \int dx dy N^3 V(N (x-y)) \wtph_t (x) \bar{\wtph}_t (y) \left[ b_x^* b_y - N^{-1} a_x^* a_y \right] \\
&+ \frac{1}{2} \int dx dy \, N^3 V(N(x-y)) \left[ \wtph_t (x) \wtph_t (y) b_x^* b_y^* + \text{h.c.} \right] \end{split} 	\end{equation}
We define the error term $\cE^{(2)}_{N,t}$ through the equation 
\begin{equation}\label{eq:def-E2} \begin{split} 
e^{-B(\eta_t)} (\cL^{(2)}_{N,t} - \cK) e^{B(\eta_t)} = \; &\text{Re } \int dx dy \, N^3 V(N(x-y)) \bar{\wtph_t} (x) \bar{\wtph_t} (y) k_t (y;x) \\ &+ \frac{1}{2} \int dx dy \, N^3 V(N(x-y)) \left[ \wtph_t (x) \wtph_t (y) b_x^* b_y^* + \text{h.c.} \right] \\ &+\cE^{(2)}_{N,t}  \end{split} 
\end{equation}  
The properties of the error term $\cE^{(2)}_{N,t}$ are described in the next proposition.
\begin{prop}\label{prop:E2N}
Under the same assumptions as in Theorem \ref{thm:gene}, there exist constants $C, c >0$ such that
\begin{equation}\label{eq:bds-E2} \begin{split} 
\left| \left\langle \xi, \cE^{(2)}_{N,t} \xi \right\rangle \right| &\leq 
Ce^{c|t|}  \| (\cN+1)^{1/2} \xi \| \| (\cV_N + \cN+1)^{1/2} \xi \| \\ 
\left| \left\langle \xi, \left[ \cN ,  \cE^{(2)}_{N,t} \right]  \xi  \right\rangle \right| &\leq 
Ce^{c|t|} \| (\cN+1)^{1/2} \xi \| \| (\cV_N + \cN + 1)^{1/2} \xi \|   \\ 
\left| \left\langle \xi, \left[ a^* (g_1) a(g_2) , \cE^{(2)}_{N,t} \right]  \xi \ \right\rangle \right| &\leq C e^{c|t|} \| g_1 \|_{H^2} \| g_2 \|_{H^2} \| (\cN+1)^{1/2} \xi \| \| (\cV_N^{1/2} + \cN+ 1)^{1/2} \xi \| \\ 
\left| \partial_t \left\langle \xi,  \cE^{(2)}_{N,t} \xi \right\rangle \right| &\leq C e^{c|t|} \| (\cN+1)^{1/2} \xi \| \| (\cV_N^{1/2} + \cN+ 1)^{1/2} \xi \| 
\end{split} \end{equation}
for all $\xi \in \cF^{\leq N}$. 
\end{prop}

\begin{proof}
The conjugation of the first two terms on the r.h.s of (\ref{eq:L2-K}) can be controlled with Lemma \ref{lm:R}, taking $r$ to be multiplication operator with the convolution $N^3 V(N.) *|\wtph_t|^2$ in the first case 
(so that $\| r \|_\text{op} = \| N^3 V(N.) * |\wtph_t|^2 \|_\infty \leq C \| \wtph_t \|_\infty^2 \leq C e^{c|t|}$) and the operator with integral kernel $r (x;y) = N^3 V(N(x-y)) \wtph_t (x) \wtph_t (y)$ in the second case 
(then $\| r \|_\text{op} \leq \sup_x \int |r(x;y)| dy \leq C e^{c|t|}$, uniformly in $N$). Hence, to show Prop. \ref{prop:E2N} it is enough to prove the bounds (\ref{eq:bds-E2}), with $\cE^{(2)}_{N,t}$ replaced by 
\begin{equation}\label{eq:wtE2} \begin{split} 
\wt{\cE}^{(2)}_{N,t} = \; &\frac{1}{2} \int dx dy \, N^3 V(N(x-y)) \left[ \bar{\wtph}_t (x) \bar{\wtph}_t (y) e^{-B(\eta_t)} b_x b_y e^{B(\eta_t)} + \text{h.c.} \right]  \\ &- \text{Re } \int dx dy \, N^3 V(N(x-y)) \bar{\wtph}_t (x) \bar{\wtph}_t (y) k_t (x;y) \\ &- \frac{1}{2} \int dx dy N^3 V(N(x-y)) \left[ \bar{\wtph}_t (x) \bar{\wtph}_t (y) b_x b_y + \text{h.c.} \right] \end{split} \end{equation}
By Lemma \ref{lm:conv-series}, we can write
\begin{equation}\label{eq:ebbe} \begin{split} 
\int dx dy \, &N^3 V(N(x-y)) \bar{\wtph}_t (x) \bar{\wtph}_t (y) e^{-B(\eta_t)} b_x b_y e^{B(\eta_t)} \\ = \; & \sum_{n,k \geq 0} \frac{(-1)^{k+n}}{k!n!} \int dx dy N^3 V(N(x-y)) \bar{\wtph}_t (x) \bar{\wtph}_t (y) \text{ad}^{(n)}_{B(\eta_t)} (b_x) \text{ad}^{(k)}_{B(\eta_t)} (b_y) \\ =\;& \int dx dy N^3 V(N(x-y)) \bar{\wtph}_t (x) \bar{\wtph}_t (y) b_x b_y \\&- \int dx dy N^3 V(N(x-y)) \bar{\wtph}_t (x) \bar{\wtph}_t (y) b_x [B(\eta_t) , b_y] \\ &+ \sum_{n,k}^* \frac{(-1)^{k+n}}{k!n!} \int dx dy N^3 V(N(x-y)) \bar{\wtph}_t (x) \bar{\wtph}_t (y) \text{ad}^{(n)}_{B(\eta_t)} (b_x) \text{ad}^{(k)}_{B(\eta_t)} (b_y) \end{split} \end{equation}
where we isolated the terms with $(n,k) = (0,0)$ and $(n,k) = (0,1)$ and the sum $\sum^*$ runs over all other pairs $(n,k) \in \bN \times \bN$. The first term on the r.h.s. of (\ref{eq:ebbe}) (the one associated with $(k,n) = (0,0)$) is subtracted in (\ref{eq:wtE2}) and does not enter the error term $\wt{\cE}_{N,t}^{(2)}$. The second term on the r.h.s. of (\ref{eq:ebbe}), on the other hand, is given by
\[ \begin{split} \text{P} := \; & -\int dx dy \, N^3 V(N(x-y)) \bar{\wtph}_t (x) \bar{\wtph}_t (y) b_x [B(\eta_t) , b_y] \\ = &\; \frac{N-1-\cN}{N} \int dx dy \, N^3 V(N(x-y)) \bar{\wtph}_t (x) \bar{\wtph}_t (y) \,  b_x  b^* (\eta_{y})  
\\ &- \frac{1}{N} \int dx dy dw dz \, N^3 V(N (x-y)) \bar{\wtph}_t (x) \bar{\wtph}_t (y) \, \eta_t (z;w) \, b_x b^*_z a_w^* a_y
\end{split}\]
Commuting in both terms the annihilation field $b_x$ to the right, we find 
\begin{equation}\label{eq:Pbd} \begin{split} \text{P} = &\; 
\frac{N-1-\cN}{N} \frac{N-\cN}{N} \int dx dy \, N^3 V(N(x-y)) \bar{\wtph}_t (x) \bar{\wtph}_t (y) \, \eta_t (x;y)
\\ & +\frac{N-1-\cN}{N} \int dx dy dz\, N^3 V(N(x-y)) \bar{\wtph}_t (x) \bar{\wtph}_t (y) \, \left[ b^* (\eta_{y}) b_x - \frac1N a^* (\eta_{y}) a_x \right]  \\
&- 2\frac{N-\cN}{N^2} \int dx dy dz\, N^3 V(N(x-y)) \bar{\wtph}_t (x) \bar{\wtph}_t (y) \,  a^* (\eta_{y}) a_x   \\
&-\frac{N-\cN}{N^2} \int dx dy dz dw \, N^3 V(N(x-y)) \bar{\wtph}_t (x) \bar{\wtph}_t (y) \,\eta_t (z;w)  
a_w^* a_z^* a_x a_y \\ =: & \; \text{P}_1 + \text{P}_2 + \text{P}_3+\text{P}_4  \end{split} \end{equation}
Writing $\eta_t = k_t + \mu_t$, and using the pointwise bounds $|\mu_t (x;y)| \leq C |\wtph_t (x)| |\wtph_t (y)|$ and $|k_t (x;y)| \leq CN |\wtph_t (x)| |\wtph_t (y)|$ from Lemma \ref{lm:eta}, we obtain that 
\[ \Big| \langle \xi, \text{P}_1 \xi \rangle - \int dx dy N^3 V(N(x-y)) \bar{\wtph}_t (x) \bar{\wtph}_t (y) k_t (x;y) \Big| \leq C \| (\cN+1)^{1/2} \xi \|^2 \]

The expectation of the operator $\text{P}_2$, and analogously the expectation of the operator $\text{P}_3$, can be bounded by
\[ \begin{split} \left| \langle \xi , \text{P}_2 \xi \rangle \right| \leq  \; &\| (\cN+1)^{1/2} \xi \| \int dx dy N^3 V(N(x-y)) |\wtph_t (x)| |\wtph_t (y)| \| \eta_{y} \|  \| b_x \xi \|  \\
\leq \; &\| \wtph_t \|_\infty^2 \, \| (\cN+1)^{1/2} \xi \| \left[ \int dx dy N^3 V(N(x-y)) \| \eta_{y} \|^2 \right]^{1/2} \\ &\hspace{4cm} \times \left[  \int dx dy N^3 V(N(x-y)) \| b_x \xi \|^2 \right]^{1/2} \\ \leq \; &C e^{c|t|} \| \eta_t \| \| (\cN+1)^{1/2} \xi \|^2\end{split} \]
As for the last term on the r.h.s. of (\ref{eq:Pbd}), its expectation is estimated by 
\[ \begin{split} \left| \left\langle \xi, \text{P}_3 \xi \right\rangle \right| &\leq \| \eta_t \| \| (\cN+1) \xi \| \int dx dy N^2 V(N(x-y)) |\wtph_t (x)| |\wtph_t (y)| \| a_x a_y \xi \| 
\\ &\leq \| \eta_t \| \| (\cN+1) \xi \| \left[ \int dx dy N^2 V(N(x-y)) \| a_x a_y \xi \|^2 \right]^{1/2} \\ &\hspace{2cm} \times  \left[ \int dx dy N^2 V(N(x-y)) |\wtph_t (x)|^2 |\wtph_t (y)|^2 \right]^{1/2} \\ 
&\leq C \| \eta_t \| \| (\cN+1)^{1/2} \xi \| \| \cV_N^{1/2} \xi \| \end{split} \]
We conclude that
\begin{equation}\label{eq:P} \begin{split} \Big| \langle \xi, \text{P} \xi \rangle - &\int dx dy N^3 V(N(x-y))\bar{\wtph}_t (x) \bar{\wtph}_t (y) k_t (x;y) \Big| \\ &\leq    C e^{c|t|} \| \eta_t \| \| (\cN+1)^{1/2} \xi \| \| (\cV_N + \cN + 1)^{1/2} \xi \| 
\end{split} \end{equation}
Let us now consider the terms in the sum on the last line of (\ref{eq:ebbe}), where we excluded the pairs $(k,n) = (0,0)$ and $(k,n) = (0,1)$. By Lemma \ref{lm:indu}, the operator 
\begin{equation}\label{eq:term-nk}\int dx dy N^3 V(N(x-y)) \wtph_t (x) \wtph_t (y) \text{ad}^{(n)}_{B(\eta_t)} (b_x) \text{ad}^{(k)}_{B(\eta_t)} (b_y) \end{equation}
can be expressed as the sum of $2^{n+k} n!k!$ terms having the form
\begin{equation}\label{eq:typbb} 
\begin{split} \text{E} = \int dx dy N^3 V(N(x-y)) \wtph_t (x) &\wtph_t (y) \Lambda_1 \dots \Lambda_{i_1} N^{-k_1} \Pi^{(1)}_{\sharp,\flat} (\eta^{(j_1)}_{t, \natural_1} , \dots \eta^{(j_{k_1})}_{t,\natural_{k_1}} ; \eta_{x,\lozenge}^{(\ell_1)}) \\ & \times \Lambda'_1 \dots \Lambda'_{i_2} N^{-k_2} \Pi^{(2)}_{\sharp',\flat'} (\eta^{(m_1)}_{t, \natural'_1} , \dots \eta^{(m_{k_2})}_{t,\natural'_{k_2}} ; \eta_{y,\lozenge'}^{(\ell_2)})
\end{split} \end{equation}
where $k_1, k_2, i_1, i_2 \geq 0$, $j_1, \dots , j_{k_1}, m_1, \dots , m_{k_2} >0$ and where each $\Lambda_i$ and $\Lambda'_i$ is either a factor $(N-\cN)/N$ or $(N+1-\cN)/N$ or a $\Pi^{(2)}$-operator of the form
\begin{equation}\label{eq:Pi2bb} N^{-p} \, \Pi^{(2)}_{\underline{\sharp}, \underline{\flat}} (\eta_{t,\underline{\natural}_1}^{(q_1)} , \dots , \eta_{t,\underline{\natural}_{p}}^{(q_p)}) \end{equation}
With Lemma \ref{lm:prel1}, we obtain
\[ \begin{split} 
| \langle \xi, \text{E} \xi \rangle | \leq \; &\| (\cN+1)^{1/2} \xi \| \int dx dy  N^3 V(N(x-y)) \wtph_t (x) \wtph_t (y) \\ & \hspace{1cm} \times  \Big\| (\cN+1)^{-1/2} \Lambda_1 \dots \Lambda_{i_1} N^{-k_1} \Pi^{(1)}_{\sharp,\flat} (\eta^{(j_1)}_{t,\natural_1}, \dots , \eta^{(j_{k_1})}_{t,\natural_{k_1}} ; \eta_{x,\lozenge}^{(\ell_1)}) \\  & \hspace{2cm} \times \Lambda'_1 \dots \Lambda'_{i_2} N^{-k_2} \Pi^{(1)}_{\sharp,\flat} (\eta^{(m_1)}_{t,\natural'_1}, \dots , \eta^{(m_{k_2})}_{t,\natural'_{k_2}} ; \eta_{y,\lozenge'}^{(\ell_2)}) \xi \Big\|  
\\ \leq \; &C^{k+n} \| \eta_t \|^{n+k-2}\| (\cN+1)^{1/2} \xi \| \int dx dy  N^3 V(N(x-y)) \wtph_t (x) \wtph_t (y) \\ &\hspace{1cm} \times \Big\{ n \| \eta_x \| \| \eta_y \| \| (\cN+1)^{1/2} \xi \| + \| \eta_t \| \| \eta_y \| \| a_x \xi \| \\ &\hspace{2cm} + C e^{c|t|} \| \eta_t \| \| (\cN+1)^{1/2} \xi \| + N^{-1/2} \| \eta_t \|^2 \| a_x a_y \xi \|  \Big\} \end{split} \]
where (in the last term in the parenthesis) we used the pointwise bound $N^{-1} |\eta_t (x;y)| \leq C e^{c|t|}$ from Lemma \ref{lm:eta}. The contribution of the first three terms in the parenthesis can be bounded by Cauchy-Schwarz, since $\| \wtph_t \|_\infty \leq C e^{c|t|}$. We find
\[ |\langle \xi ,\text{E} \xi \rangle | \leq C^{k+n} n e^{c|t|} \| \eta_t \|^{k+n-1} \| (\cN+1)^{1/2} \xi \| \| (\cV_N + \cN + 1)^{1/2} \xi \| \]

Since the expectation of (\ref{eq:term-nk}) is the sum of $2^{n+k}k!n!$ such contributions, inserting in (\ref{eq:ebbe}) and taking into account also (\ref{eq:P}), we conclude that 
\[ \left| \langle \xi , \wt\cE_{N,t}^{(2)} \xi \rangle \right| \leq \, C e^{c|t|} \, \| (\cN+1)^{1/2} \xi \| \| (\cV_N + \cN+ 1)^{1/2} \xi \| \]
if $\sup_t \| \eta_t \|$ is small enough. As usual, we can prove similarly that the same bounds hold true 
for the expectation of the commutators of $\wt{\cE}_{N,t}^{(2)}$ with the number of particles operator $\cN$ and with $a^* (g_1) a(g_2)$, for arbitrary $g_1, g_2 \in H^2 (\bR^3)$ (this assumption allows us to extract $\| g_j \|_{\infty} \leq C \| g_j \|_{H^2}$) and also for the time derivative of $\wt{\cE}_{N,t}^{(2)}$. 
\end{proof}

\subsection{Analysis of $e^{-B(\eta_t)} \cL^{(3)}_{N,t} e^{B(\eta_t)}$}
\label{sec:L3}

Recall from (\ref{eq:L0-4}) that
\[ \cL^{(3)}_{N,t} = \int dx dy\, N^{5/2} V(N(x-y))\wtph_t(y)\left[ b_x^* a_y^* a_x + \text{h.c.} \right] 
\]
We conjugate $\cL^{(3)}_{N,t}$ with the unitary operator $e^{B(\eta_t)}$. We define the error term $\cE^{(3)}_{N,t}$ through the equation
\begin{equation}\label{eq:eL3e-in} e^{-B(\eta_t)} \cL_{N,t}^{(3)} e^{B(\eta_t)} = -\sqrt{N} \, \left[ b (\cosh_{\eta_t} (h_{N,t})) + b^* (\sinh_{\eta_t} (\bar{h}_{N,t})) + \text{h.c.} \right]  + \cE^{(3)}_{N,t} \end{equation}
where we recall, from (\ref{eq:eL1e}) that, $h_{N,t} = (N^3 V(N.) w_\ell (N.) * |\wtph_t|^2) \wtph_t$. In the next proposition we collect the important properties of the error term $\cE^{(3)}_{N,t}$
\begin{prop}\label{prop:L3}
Under the same assumptions as in Theorem \ref{thm:gene}, there exist constants $C, c >0$ such that
\begin{equation}\label{eq:E3-bds} \begin{split} 
\left| \langle \xi, \cE^{(3)}_{N,t} \xi \rangle \right| 
&\leq C e^{c|t|} \| (\cN+1)^{1/2} \xi \| \| (\cV_N + \cN + 1)^{1/2} \xi \| \\
\left| \langle \xi, \left[ \cN, \cE^{(3)}_{N,t} \right] \xi \rangle \right| &\leq C e^{c|t|}  \| (\cN+1)^{1/2} \xi \| \| (\cV_N + \cN + 1)^{1/2} \xi \| \\
\left| \langle \xi , \left[ a^* (g_1) a(g_2) , \cE^{(3)}_{N,t} \right] \xi \rangle \right| &\leq C e^{c|t|} \| g_1 \|_{H^2} \| g_2 \|_{H^2} \| (\cN+1)^{1/2} \xi \| \| (\cV_N + \cN + 1)^{1/2} \xi \| \\
\left| \partial_t \langle \xi, \cE^{(3)}_{N,t} \xi \rangle \right| &\leq C e^{c|t|}  \| (\cN+1)^{1/2} \xi \| \| (\cV_N + \cN + 1)^{1/2} \xi \|
\end{split} \end{equation}
for all $\xi \in \cF^{\leq N}$. 
\end{prop}

\begin{proof}
We start by writing 
\[ \begin{split} 
e^{-B(\eta_t)} a_y^* a_x e^{B(\eta_t)} &= a_y^* a_x + \int_0^1 ds \, e^{-sB (\eta_t)} [a_y^* a_x , B(\eta_t)] e^{sB(\eta_t)} \\ &= a_y^* a_x + \int_0^1 e^{-sB(\eta_t)}  \left[ b_y^* b^* (\eta_x) + b(\eta_y) b_x \right] e^{sB(\eta_t)} \end{split} \]
{F}rom Lemma \ref{lm:conv-series}, we conclude that 
\[ \begin{split} e^{-B(\eta_t)} a_y^* a_x e^{B(\eta_t)} = \; &a_y^* a_x + \sum_{k,r \geq 0} \frac{(-1)^{k+r}}{k! r! (k+r+1)}\\ &\hspace{1.5cm} \times \left[ \text{ad}^{(k)}_{B(\eta_t)} (b_y^*) \text{ad}^{(r)}_{B(\eta_t)} (b^* (\eta_x)) + \text{ad}^{(k)}_{B(\eta_t)} (b(\eta_y)) \text{ad}^{(r)}_{B(\eta_t)} (b_x) \right] \end{split} \]
Inserting in the expression for $\cL^{(3)}_{N,t}$, we conclude that 
\[ \begin{split} 
e^{-B(\eta_t)} & \cL^{(3)}_{N,t} e^{B(\eta_t)} \\ = \; &\sum_{n \geq 0} \frac{(-1)^{n}}{n!} \int dx dy N^{5/2} V(N (x-y)) \wtph_t (y)  \text{ad}^{(n)}_{B(\eta_t)} (b^*_x) a_y^* a_x \\
&+ \sum_{n,k,r \geq 0} \frac{(-1)^{n+k+r}}{n! k! r! (k+r+1)} \int dx dy N^{5/2} V(N (x-y)) \wtph_t (y) \, \text{ad}^{(n)}_{B(\eta_t)} (b^*_x) \\ &\hspace{2cm} \times \left[ \text{ad}^{(k)}_{B(\eta_t)} (b_y^*) \text{ad}^{(r)}_{B(\eta_t)} (b^* (\eta_x)) + \text{ad}^{(k)}_{B(\eta_t)} (b(\eta_y)) \text{ad}^{(r)}_{B(\eta_t)} (b_x) \right] \\ &+\text{h.c.} \end{split} \]
We divide the triple sum in several parts. We find 
\[ \begin{split} 
e^{-B(\eta_t)} & \cL^{(3)}_{N,t} e^{B(\eta_t)} \\ = \; &\sum_{n \geq 0} \frac{(-1)^{n}}{n!} \int dx dy N^{5/2} V(N (x-y)) \wtph_t (y) \text{ad}^{(n)}_{B(\eta_t)} (b^*_x) a_y^* a_x \\
&+ \sum_{n,r \geq 0} \frac{(-1)^{n+r}}{n! (r+1)!} \int dx dy N^{5/2} V(N (x-y)) \wtph_t (y) \text{ad}^{(n)}_{B(\eta_t)} (b^*_x) b_y^* \text{ad}^{(r)}_{B(\eta_t)} (b^* (\eta_x)) \\
&+ \sum_{n,r \geq 0, k \geq 1} \frac{(-1)^{n+k+r}}{n!k! r! (k+r+1)} \int dx dy N^{5/2} V(N (x-y)) \wtph_t (y)  \\ &\hspace{6.5cm} \times \text{ad}^{(n)}_{B(\eta_t)} (b^*_x) \text{ad}^{(k)} (b_y^*) \text{ad}^{(r)}_{B(\eta_t)} (b^* (\eta_x)) \\ 
&+ \sum_{n,r \geq 0, k \geq 1} \frac{(-1)^{n+k+r}}{n!k! r! (k+r+1)}\int dx dy N^{5/2} V(N (x-y)) \wtph_t (y)   \\ &\hspace{6.5cm} \times  \text{ad}^{(n)}_{B(\eta_t)}(b_x^*)\text{ad}^{(k)}_{B(\eta_t)} (b(\eta_y)) \text{ad}^{(r)}_{B(\eta_t)} (b_x) \\ &+\text{h.c.} \end{split}  \]
In the terms with $k=0$, we distinguish furthermore the case $n=1$ from $n\not = 1$. We find
\begin{equation}\label{eq:L3-sum} \begin{split} e^{-B(\eta_t)} & \cL^{(3)}_{N,t} e^{B(\eta_t)} \\ = \; &- 
\int dx dy N^{5/2} V(N (x-y)) \wtph_t (y) [B(\eta_t), b^*_x] a_y^* a_x  \\ 
&-\sum_{r \geq 0} \frac{(-1)^r}{(r+1)!}   \int dx dy N^{5/2} V(N (x-y)) \wtph_t (y) [ B(\eta_t) , b^*_x] b_y^* \, \text{ad}^{(r)}_{B(\eta_t)} (b^* (\eta_x)) \\
&+ \sum_{n \not = 1} \frac{(-1)^n}{n!} \int dx dy N^{5/2} V(N(x-y)) \wtph_t (y) \text{ad}^{(n)}_{B(\eta_t)} (b_x^*) a_y^* a_x \\
&+ \sum_{n \not = 1, r \geq 0} \frac{(-1)^{n+r}}{n! (r+1)!} \int dx dy N^{5/2} V(N (x-y)) \wtph_t (y) \text{ad}_{B(\eta_t)}^{(n)}  (b_x^*) b_y^* \text{ad}^{(r)}_{B(\eta_t)} (b^* (\eta_x)) \\
&+\sum_{n,r \geq 0, k \geq 1} \frac{(-1)^{n+k+r}}{n!k! r! (k+r+1)} \int dx dy N^{5/2} V(N (x-y)) \wtph_t (y)   \\ &\hspace{6.5cm} \times  \text{ad}^{(n)}_{B(\eta_t)} (b^*_x) \text{ad}^{(k)}_{B(\eta_t)} (b_y^*) \text{ad}^{(r)}_{B(\eta_t)} (b^* (\eta_x)) \\ 
&+ \sum_{n,r \geq 0, k \geq 1} \frac{(-1)^{n+k+r}}{n!k! r! (k+r+1)}\int dx dy N^{5/2} V(N (x-y)) \wtph_t (y) 
\\ &\hspace{6.5cm} \times  \text{ad}^{(n)}_{B(\eta_t)}(b_x^*) \text{ad}^{(k)}_{B(\eta_t)} (b(\eta_y)) \text{ad}^{(r)}_{B(\eta_t)} (b_x) \\ &+\text{h.c.} 
\end{split}  \end{equation}

We start by estimating the contribution of the last term on the r.h.s. of (\ref{eq:L3-sum}). We are interested in the expectation  
\[ \begin{split} &\left| \int dx dy \, N^{5/2} V(N (x-y)) \wtph_t (y) \, \langle \xi, \text{ad}^{(n)}_{B(\eta_t)} (b_x^*) \text{ad}^{(k)}_{B(\eta_t)} (b(\eta_y)) \text{ad}^{(r)}_{B(\eta_t)} (b_x) \xi \rangle \right| \\ &\hspace{1cm} \leq \int dx dy \, N^{5/2} V(N(x-y)) |\wtph_t (y)| \| \text{ad}^{(n)}_{B(\eta_t)} (b_x) \xi \| \| \text{ad}^{(k)}_{B(\eta_t)} (b(\eta_y)) \text{ad}^{(r)}_{B(\eta_t)} (b_x) \xi \| \end{split} \]
for $n,r \geq 0$ and $k \geq 1$. According to Lemma \ref{lm:conv-series}, the norm $\| \text{ad}^{(n)}_{B(\eta_t)} (b_x) \xi \|$ is bounded by the sum of $2^n n!$ terms of the form  
\[ \text{P}_1 = \| \Lambda_1 \dots \Lambda_{i} N^{-k} \Pi^{(1)}_{\sharp,\flat} (\eta^{(j_1)}_{t,\natural_1}, \dots , \eta^{(j_k)}_{t,\natural_k} ; \eta^{(s)}_{x,\lozenge}) \xi \| \]
for $i,k,s \geq 0$, $j_1, \dots, j_k \geq 1$, where each $\Lambda_i$ is either a factor $(N-\cN)/N$ or $(N+1-\cN)/N$ or a $\Pi^{(2)}$-operator of the form
\begin{equation}\label{eq:Pi2-L3} N^{-p} \Pi^{(2)}_{\sharp',\flat'} (\eta^{(q_1)}_{t,\natural'_1}, \dots , \eta^{(q_p)}_{t,\natural'_p}) \end{equation}
{F}rom Lemma \ref{lm:prel1}, we find 
\begin{equation}\label{eq:P1} \text{P}_1 \leq \left\{ \begin{array}{ll} C^{n} \| \eta \|^{n-1} \| \eta_x \| \| (\cN+1)^{1/2} \xi \| \quad &\text{if $s \geq 1$} \\ C^n \| \eta \|^n \| a_x \xi \| \quad & \text{if $s =0$} \end{array} \right. \end{equation}
Similarly, the norm $\| \text{ad}^{(k)}_{B(\eta_t)} (b(\eta_y)) \text{ad}^{(r)}_{B(\eta_t)} (b_x) \xi \|$ is bounded by the sum of $2^{k+r} k!r!$ terms having the form
\[\begin{split} & \text{P}_2 = \left\| \Lambda_1 \dots \Lambda_{i_1} N^{-k_1} \Pi^{(1)}_{\sharp,\flat} (\eta^{(j_1)}_{t,\natural_1}, \dots , \eta^{(j_{k_1})}_{t,\natural_{k_1}} ; \eta^{(\ell_1+1)}_{y,\lozenge}) \right. \\  &\hspace{4cm} \left. \times \Lambda'_1 \dots \Lambda'_{i_2} N^{-k_2} \Pi^{(1)}_{\sharp',\flat'} (\eta^{(m_1)}_{t,\natural'_1} , \dots , \eta^{(m_{k_2})}_{t,\natural'_{k_2}} ; \eta^{(\ell_2)}_{x,\lozenge'})  \xi \right\| \end{split}  \]
which can be estimated (again with Lemma \ref{lm:prel1}) by 
\[ \text{P}_2 \leq \left\{ \begin{array}{ll} C^{k+r} \| \eta_t \|^{k+r-2} \| \eta_x \| \| \eta_y \| \| (\cN +1) \xi \| \quad &\text{if } \ell_2 \geq 1 \\ C^{k+r} \| \eta_t \|^{k+r-1} \| \eta_y \| \| a_x (\cN+1)^{1/2} \xi \| \quad &\text{if } \ell_2 = 0 \end{array} \right.  
\]
Combining this estimate with (\ref{eq:P1}), distinguishing different cases depending on the values of $s$ and $\ell_2$, and using the estimate $\sup_y \| \eta_{y} \| \leq  C e^{c|t|} < \infty$ from Lemma \ref{lm:eta}, we easily find by Cauchy-Schwarz that
\begin{equation}\label{eq:L3-sixf} \begin{split}
&\left| \int dx dy \, N^{5/2} V(N (x-y)) \wtph_t (y) \, \langle \xi, \text{ad}^{(n)}_{B(\eta_t)} (b_x^*) \text{ad}^{(k)}_{B(\eta_t)} (b(\eta_y)) \text{ad}^{(r)}_{B(\eta_t)} (b_x) \xi \rangle \right| \\ &\hspace{3cm} \leq n!k!r! \, C^{n+k+r} N^{-1/2} \| \eta_t \|^{k+r-1} \| (\cN+1)^{1/2} \xi \| \| (\cN+1) \xi \| \\ &\hspace{3cm} \leq n!k!r!\, C^{n+k+r} \| \eta_t \|^{k+r-1} \|(\cN+1)^{1/2} \xi \|^2
\end{split} \end{equation}
for all $\xi \in \cF^{\leq N}$. 
	
Let us now consider the fifth sum on the r.h.s. of (\ref{eq:L3-sum}). The expectation of every term in this sum is bounded by 
\begin{equation}\label{eq:L3-fif} \begin{split} &\left| \int dx dy \, N^{5/2} V(N(x-y)) \wtph_t (y) \langle \xi , \text{ad}^{(n)}_{B(\eta_t)} (b_x^*) \text{ad}^{(k)}_{B(\eta_t)} (b_y^*) \text{ad}^{(r)}_{B(\eta_t)} (b^*(\eta_{x})) \xi \rangle \right| 
\\ &\hspace{.5cm} \leq \int dx dy \, N^{5/2} V(N(x-y)) |\wtph_t (y)| \, \| \text{ad}^{(k)}_{B(\eta_t)} (b_y) \,  \text{ad}^{(n)}_{B(\eta_t)} (b_x) \xi \| \, \| \text{ad}^{(r)}_{B(\eta_t)} (b^* (\eta_x)) \xi \| \end{split} \end{equation} 
where we assume $k \geq 1$, $n,r \geq 0$. According to Lemma \ref{lm:indu}, $\| \text{ad}^{(r)}_{B(\eta_t)} (b^* (\eta_x)) \xi \|$ is bounded by the sum of $2^r r!$ terms of the form
\begin{equation*} \text{Q}_1 = \| \Lambda_1 \dots \Lambda_{i_1} N^{-k_1} \Pi^{(1)}_{\sharp,\flat} (\eta_{t,\natural_1}^{(j_1)}, \dots , \eta_{t,\natural_{k_1}}^{(j_{k_1})} ; \eta^{(\ell_1 + 1)}_x ) \, \xi \| \end{equation*}
for a $i_1, k_1, \ell_1 \geq 0$ and $j_1, \dots , j_{k_1} \geq 1$. Each $\Lambda_i$ is either a factor $(N-\cN)/N$, a factor $(N+1-\cN)/N$ or a $\Pi^{(2)}$-operator of the form (\ref{eq:Pi2-L3}). 
{F}rom Lemma \ref{lm:prel1}, we have  
\[  \text{Q}_1 \leq C^r \| \eta_t \|^r \| \eta_x \| 
\| (\cN+1)^{1/2} \xi\|  \]

On the other hand, using again Lemma \ref{lm:indu} the norm $\| 
\text{ad}^{(k)}_{B(\eta_t)} (b_y) \,  \text{ad}^{(n)}_{B(\eta_t)} 
(b_x) \xi \|$ is bounded by the sum of $2^{n+k} k!n!$ terms having the form
\[ \begin{split} \text{Q}_2 &= \| \Lambda_1 \dots \Lambda_{i_1} N^{-k_1} \Pi^{(1)}_{\sharp, \flat} (\eta_{t,\natural_1}^{(j_1)} , \dots , \eta_{t,\natural_{k_2}}^{(j_{k_1})} ; \eta_{y,\lozenge}^{(\ell_1)}) \\ &\hspace{3cm} \times \Lambda'_1 \dots \Lambda'_{i_2} N^{-k_2} \Pi^{(1)}_{\sharp', \flat'} (\eta_{t,\natural'_1}^{(m_1)},  \dots , \eta_{t,\natural'_{k_2}}^{(m_{k_2})} ; \eta_{x,\lozenge'}^{(\ell_2)}) \xi \| \end{split} \]
where $i_1, i_2, k_1, k_2, \ell_1, \ell_2 \geq 0$ and $j_1, \dots , j_{k_1}, m_1, \dots, m_{k_2} \geq 1$ and where each $\Lambda_i$ and $\Lambda'_i$ operator is either a factor $(N-\cN)/N$, $(N-\cN+1)/N$ or a $\Pi^{(2)}$-operator of the form (\ref{eq:Pi2-L3}). 
Using part iv) of Lemma \ref{lm:prel1}, we obtain (using the assumption $k \geq 1$ to apply (\ref{eq:prelim-iv3}) and using (\ref{eq:prelim-iv4a}) with $\alpha = 1$)  
\[ \begin{split} \text{Q}_2 \leq C^{n+k}& \| \eta_t \|^{n+k-2} \Big\{ \left[ (n+1) \| \eta_x \| \| \eta_y \| + \| \eta_t \| N^{-1} |\eta_t (x;y) \right] \| (\cN+1) \xi \| \\ & \hspace{3cm} + \| \eta_y \| \| \eta_t \|  \| a_x (\cN+1)^{1/2} \xi \| + \| \eta_t \|^2 \| a_x a_y \xi \| \Big\} \end{split} \]
With the bound $\sup_x \| \eta_x \|, \sup_{x,y} N^{-1} |\eta_t (x;y)| \leq C e^{c|t|}$ from Lemma \ref{lm:eta}, we conclude that 
\begin{equation}\label{eq:L3-fiff}  \begin{split} &\left| \int dx dy \, N^{5/2} V(N(x-y)) \wtph_t (y) \langle \xi , \text{ad}^{(n)}_{B(\eta_t)} (b_x^*) \text{ad}^{(k)}_{B(\eta_t)} (b_y^*) \text{ad}^{(r)}_{B(\eta_t)} (\eta_{x}) \xi \rangle \right| 
\\ &\hspace{.5cm} \leq n!k!r! \, C^{n+k+r} e^{c|t|} \| \eta_t \|^{n+k+r}
\| (\cN+1)^{1/2} \xi \| \| (\cV_N + \cN + 1)^{1/2} \xi \|  \end{split} \end{equation}
for all $\xi \in \cF^{\leq N}$. 

Let us now study the fourth term on the r.h.s. of (\ref{eq:L3-sum}). As we did for the other terms, we bound the expectation 
\begin{equation}\label{eq:L3-four} \begin{split} &\left| \int dx dy \, N^{5/2} V(N (x-y)) \wtph_t (y) \langle \xi, \text{ad}^{(n)}_{B(\eta_t)} (b_x^*) b_y^* \text{ad}_{B(\eta_t)}^{(r)} (b^* (\eta_x)) \xi \rangle \right| \\ &\hspace{1cm} \leq \int dx dy \, N^{5/2} V(N(x-y)) |\wtph_t (y)| \| b_y \text{ad}^{(n)}_{B(\eta_t)} (b_x) \xi \| \| \text{ad}^{(r)}_{B(\eta_t)} (b^* (\eta_x)) \xi \| \end{split} 
\end{equation}
where we assume that $n \not = 1$, $r \geq 0$. According to Lemma \ref{lm:indu}, $\|\text{ad}^{(r)}_{B(\eta_t)} (b^* (\eta_x)) \xi \|$ can be bounded by the sum of $2^r r!$ terms of the form
\[ \text{R}_1 = \| \Lambda_1 \dots \Lambda_{i_1} N^{-k_1} \Pi^{(1)}_{\sharp, \flat} (\eta_{t,\natural_1}^{(j_1)}, \dots , \eta^{(j_{k_1})}_{t,\natural_{k_1}} ; \eta^{(\ell_1 + 1)}_{x,\lozenge}) \xi \| \]
for $i_1, k_1, \ell_1 \geq 0$ and $j_1, \dots , j_{k_1} \geq 1$. According to Lemma \ref{lm:prel1}, such a term can always be estimated by 
\begin{equation}\label{eq:R1} \text{R}_1 \leq C^r \| \eta_t \|^{r} \| \eta_{x} \| \| (\cN+1)^{1/2} \xi \| \end{equation}
On the other hand, the norm $\| b_y \text{ad}^{(n)}_{B(\eta_t)} (b_x) \xi \|$ can be bounded by the sum of $2^n n!$ contributions having the form
\begin{equation}\label{eq:typL3-4} \text{R}_2 = \| b_y \Lambda_1 \dots \Lambda_{i_1} \Pi^{(k_1)}_{\sharp,\flat} (\eta^{(j_1)}_{t,\natural_1} , \dots, \eta^{(j_{k_1})}_{t,\natural_{k_1}} ; \eta^{(\ell_1)}_{x,\lozenge}) \xi \| \end{equation}
for $i_1, k_1, \ell_1 \geq 0$ and $j_1, \dots, j_{k_1} \geq 1$. With Lemma \ref{lm:prel1}, we find that
\[ \begin{split}
 \text{R}_2 &\leq C^n \| \eta_t \|^{n-2} \Big\{ \left[ (1+n/N)  \| \eta_x \| \| \eta_y \| + \| \eta_t \| N^{-1} |\eta_t (x;y)| \right] \| (\cN+1) \xi \|  \\ &\hspace{2.5cm} + \| \eta_t  \| \| \eta_x \| \| a_y (\cN+1)^{1/2} \xi \| + (n/N) \| \eta_t \| \| \eta_y \|  \| a_x (\cN+1)^{1/2} \xi \| \\ &\hspace{2.5cm} + \| \eta_t \|^2 \| a_x a_y \xi \| \Big\} \end{split} \]
With $\| \wtph_t \|_{\infty} \leq C e^{c|t|}$ and $\sup_{x,y} N^{-1} |\eta_t (x;y)| \leq C e^{c|t|}$ we conclude, similarly to (\ref{eq:L3-fiff}), that
\begin{equation}\label{eq:L3-fourf}  \begin{split} &\left| \int dx dy \, N^{5/2} V(N (x-y)) \wtph_t (y) \langle \xi, \text{ad}^{(n)}_{B(\eta_t)} (b_x^*) b_y^* \text{ad}_{B(\eta_t)}^{(r)} (b^* (\eta_x)) \xi \rangle \right| \\ &\hspace{2cm} \leq (n+1)!r! C^{n+r} e^{c|t|} \| \eta_t \|^{r+n} \| (\cN+1)^{1/2} \xi \| \| (\cV_N + \cN+1)^{1/2} \xi \| \end{split} \end{equation}

The expectation of terms in the third sum on the r.h.s. of (\ref{eq:L3-sum}) are bounded by 
\[ \begin{split} 
&\left|  \int dx dy \, N^{5/2} V(N(x-y)) \wtph_t (y) \langle \xi , \text{ad}^{(n)}_{B(\eta_t)} (b_x^*) a_y^* a_x \xi \rangle \right| \\ &\hspace{2cm} 
\leq \int dx dy \, N^{5/2} V(N(x-y)) |\wtph_t (y)| \| a_y \text{ad}^{(n)}_{B(\eta_t)} \xi \| \| a_x \xi \| \end{split} \]
which is similar to the r.h.s. of (\ref{eq:L3-four}), the only difference being that instead of $\| \text{ad}^{(r)}_{B(\eta_t)} (b^* (\eta_x)) \xi \|$ we have $\| a_x \xi \|$ (and the fact that in the other norm, we have the field $a_y$ instead of $b_y$; it is clear, however, that both fields can be treated similarly). Analogously to (\ref{eq:L3-fourf}), we conclude that
\begin{equation}\label{eq:L3-thirdf} \begin{split} 
&\left|  \int dx dy \, N^{5/2} V(N(x-y)) \wtph_t (y) \langle \xi , \text{ad}^{(n)}_{B(\eta_t)} (b_x^*) a_y^* a_x \xi \rangle \right| \\ &\hspace{2cm} \leq (n+1)! \,  C^{n} e^{c|t|} \| \eta_t \|^{n-1} \| (\cN+1)^{1/2} \xi \| \| (\cV_N + \cN+1)^{1/2} \xi \| \end{split} 
\end{equation}

Let us now switch to the second term on the r.h.s. of (\ref{eq:L3-sum}) (the sum over $r \geq 0$). First of all, we compute the commutator
\[ [ B(\eta_t) , b_x^* ] = -  b (\eta_{x})\left( 1 - \frac{\cN}{N} \right) + \frac{1}{N} \int dz dw \bar{\eta} (z;w) a_x^* a_w b_z  \]
Hence the $r$-th term in the sum is proportional to 
\begin{equation} \label{eq:S1S2} \begin{split} -\int dx dy &N^{5/2} V(N(x-y)) \wtph_t (y) \frac{(N-1-\cN)}{N} b (\eta_{x}) b_y^* \text{ad}^{(r)}_{B(\eta_t)} (b^* (\eta_x)) \\ &+ \int dx dy N^{5/2} V(N(x-y)) \wtph_t (y) N^{-1} \Pi^{(1)}_{(*,\cdot),*} (\eta_t , \delta_x)^* b_y^* \text{ad}^{(r)}_{B(\eta_t)} (b^* (\eta_x)) \\ =: & \;  \text{S}_1 + \text{S}_2 \end{split} \end{equation}
The expectation of $\text{S}_2$ can be bounded as follows.
\[ |\langle \xi , \text{S}_2 \xi \rangle | \leq \int dx dy \, N^{5/2} V(N(x-y)) |\wtph_t (y)| \| b_y N^{-1} \Pi^{(1)}_{(*,\cdot),*} (\eta_t , \delta_x) \xi \| \| \text{ad}^{(r)}_{B(\eta_t)} (b^* (\eta_x)) \xi\| \]
As in (\ref{eq:R1}), we find   
\[ \| \text{ad}^{(r)}_{B(\eta_t)} (b^* (\eta_x)) \xi\| \leq C^r r! \| \eta_t \|^r \| \eta_x \| \| (\cN+1)^{1/2} \xi \| \]
Since, on the other hand,
\[ \|  b_y N^{-1} \Pi^{(1)}_{(*,\cdot),*} (\eta_t , \delta_x) \xi \| \leq C N^{-1} \| \eta_{y} \| \| a_x (\cN+1)^{1/2} \xi \| + C\| \eta_{t} \| \| a_xa_y\xi \| \]
we conclude that
\[  |\langle \xi , \text{S}_2 \xi \rangle | \leq C^r e^{c|t|} \| \eta_t \|^{r+1} \| (\cN+1)^{1/2} \xi \| \| (\mathcal{V}_N+\cN+1)^{1/2} \xi \| \]
for all $\xi \in \cF^{\leq N}$. We are left with the operator $\text{S}_1$ defined in (\ref{eq:S1S2}). Commuting $b (\eta_{x})$ with $b_y^*$ we write it as
\[ \begin{split} \text{S}_1 = \; &-\int dx dy N^{5/2} V(N(x-y)) \eta_t (x;y) \wtph_t (y) \frac{(N-\cN)(N-\cN-1)}{N^2} \text{ad}^{(r)}_{B(\eta_t)} (b^* (\eta_x))\\
&- \int dx dy N^{5/2} V(N(x-y)) \wtph_t (y) \frac{(N-\cN-1)}{N} \left[b_y^* b (\eta_{x})-\frac1N a_y^* a (\eta_{x})\right]  \text{ad}^{(r)}_{B(\eta_t)} (b^* (\eta_x)) \\
=: \; & \text{S}_{11} + \text{S}_{12} \end{split} \]
The expectation of $\text{S}_{12}$ is estimated by 
\[ |\langle \xi ,\text{S}_{12} \xi \rangle | \leq C^r e^{c|t|} \| \eta_t \|^{r+1} \| (\cN+1)^{1/2} \xi \|^2 \]
As for $\text{S}_{11}$, we decompose
\[ \begin{split} 
\text{S}_{11} = &-\int dx dy N^{5/2} V(N(x-y)) k_t (x;y) \wtph_t (y) \frac{(N-\cN)(N-\cN-1)}{N^2} \text{ad}^{(r)}_{B(\eta_t)} (b^* (\eta_x)) \\ &-\int dx dy N^{5/2} V(N(x-y)) \mu_t (x;y) \wtph_t (y) \frac{(N-\cN)(N-\cN-1)}{N^2} \text{ad}^{(r)}_{B(\eta_t)} (b^* (\eta_x)) \\ =: &\; \text{S}_{111} + \text{S}_{112} \end{split} \]
Since $|\mu_t (x;y)| \leq C e^{c|t|}$ from Lemma \ref{lm:eta}, it is easy to  estimate the expectation of the term $\text{S}_{112}$ by 
\[ |\langle \xi , \text{S}_{112} \xi \rangle | \leq C^r e^{c|t|}  \| \eta_t \|^{r+1} \| (\cN+1)^{1/2} \xi \|^2  \]
As for the term $\text{S}_{111}$, we use the fact that, by Lemma \ref{lm:indu}, the nested commutator $\text{ad}^{(r)}_{B(\eta_t)} (b^* (\eta_x))$ is given by 
\[ \left( 1-\frac{\cN-1}{N} \right)^{m} \left( 1- \frac{\cN-2}{N} \right)^{m}b^* ((\eta_t \bar{\eta}_t)^m \eta_{x}) \]
if $r = 2m$ is even and by 
\[ - \left( 1-\frac{\cN+1}{N} \right)^{m+1} \left( 1- \frac{\cN}{N} \right)^{m} b ((\eta_t \bar{\eta}_t)^{m+1}_x) \]
if $r= 2m+1$ is odd, up to terms ($2^r r!-1$ of them) having the form
\[ \Lambda_1 \dots \Lambda_{i_1} N^{-k_1} \Pi^{(1)}_{\sharp,\flat} (\eta^{(j_1)}_{t,\natural_1}, \dots , \eta^{(j_{k_1})}_{t,\natural_{k_1}} ; \eta^{(\ell_1+1)}_{x,\lozenge} ) \]
where either $k_1 \geq 1$ or at least one of the $\Lambda$-operators is a $\Pi^{(2)}$-operator of the form (\ref{eq:Pi2-L3}). We conclude that, if $r = 2m$ is even,
\begin{equation}\label{eq:S11-even} \text{S}_{111} =  \sqrt{N} \int dx dy N^3 V(N(x-y))  w_\ell (N(x-y)) |\wtph_t (y)|^2 \wtph_t (x) b^* ((\eta_t \bar{\eta}_t)^m \eta_{x}) + \text{S}_{1112} \end{equation}
while, if $r = 2m+1$ is odd, 
\begin{equation}\label{eq:S11-odd} \text{S}_{111} = -\sqrt{N} \int dx dy N^{3}  V(N(x-y)) w_\ell (N(x-y)) |\wtph_t (y)|^2 \wtph_t (x) b^* ((\eta_t \bar{\eta}_t)_{x}^{m+1} ) + \text{S}_{1112} \end{equation}
where, in both cases, the expectation of the error term $\text{S}_{1112}$ is bounded by 
 \[\begin{split}  |\langle \xi , \text{S}_{1112} \xi \rangle | &\leq C^r \| \eta_t \|^r \int dx dy \, N^{3/2} V(N(x-y)) \, |k_t (x;y)| \| \eta_x \| \| (\cN+1)^{1/2} \xi \| \| (\cN+1) \xi \| \\ &\leq C^r \| \eta_t \|^{r+1} \| (\cN+1)^{1/2} \xi \|^2 \end{split} \]
for all $\xi \in \cF^{\leq N}$. Here, once again, we used the fact that $N^{-1} |\eta_t (x;y)| \leq C$. Summing over all $r \geq 0$, we conclude that
\[ \begin{split} 
- \sum_{r \geq 0} \frac{(-1)^r}{(r+1)!} &\int dx dy \, N^{5/2} V(N(x-y)) \wtph_t (y) [B(\eta_t), b_x^*] b_y^* \text{ad}^{(r)}_{B(\eta_t)} (b^* (\eta_{x})) \\ &= -\sqrt{N} \left[ b ((\cosh_{\eta_t} -1) (h_{N,t})) + b^* (\sinh_{\eta_t} (h_{N,t})) \right] + \text{S} \end{split} \]
where 
\begin{equation}\label{eq:exp-S} |\langle \xi, S \xi \rangle | \leq e^{c|t|} \sum_{r \geq 0} (C \| \eta_t \|)^r \| (\cN+1)^{1/2} \xi \|^2 \leq 
C e^{c|t|} \| (\cN+1)^{1/2} \xi \|^2 \end{equation}
for all $\xi \in \cF^{\leq N}$. 

Finally, we consider the first term on the r.h.s. of (\ref{eq:L3-sum}). This term can be handled similarly as we did with the second term (the sum over $r \geq 0$). We obtain that 
\[ - \int dx dy N^{5/2} V(N(x-y)) \wtph_t (y) [B(\eta_t),b^*_x] a_y^* a_x = -\sqrt{N} b(h_{N,t}) + \wt{S} \]
where the expectation of $\wt{S}$ can be bounded as we did with the expectation of $S$ in (\ref{eq:exp-S}). 

Recalling the definition of $\cE^{(3)}_{N,t}$ in (\ref{eq:eL3e-in}), it follows from (\ref{eq:L3-sixf}), (\ref{eq:L3-fiff}), (\ref{eq:L3-fourf}), (\ref{eq:L3-thirdf}) and (\ref{eq:exp-S}) that 
\[ |\langle \xi , \cE^{(3)}_{N,t} \xi \rangle | \leq C e^{c|t|} \| (\cN+1)^{1/2} \xi \| \| (\cV_N + \cN+1)^{1/2} \xi \| \]
The bounds in (\ref{eq:E3-bds}) for the expectation of the commutators $[ \cN , \cE^{(3)}_{N,t}]$, $[a^* (g_1) a(g_2), \cE^{(3)}_{N,t}]$ and of the time-derivative $\partial_t \cE^{(3)}_{N,t}$ can be proven analogously. We omit the details. 
\end{proof}

\subsection{Analysis of $e^{-B(\eta_t)} \cL^{(4)}_{N,t} e^{B(\eta_t)}$}
\label{sec:L4}

Recall from (\ref{eq:L0-4}) that
\[ \cL^{(4)}_{N,t} = \cV_N = \frac{1}{2} \int dx dy \, N^2 V(N(x-y)) a_x^* a_y^* a_y a_x  \]
We conjugate $\cL^{(4)}_{N,t}$ with the unitary operator $e^{B(\eta_t)}$. We define the error term $\cE^{(4)}_{N,t}$ through the equation

\begin{equation}\label{eq:eL4e-in} 
\begin{split} 
e^{-B(\eta_t)} \cL_{N,t}^{(4)} e^{B(\eta_t)} 
= \; &\cV_N + \frac{1}{2} \int dx dy \, N^2 V(N(x-y)) |k_t (x;y)|^2 \\ &+ \frac{1}{2} \int dx dy \, N^2 V(N(x-y)) \left[ k_t (x;y) b_x^* b_y^* + \text{h.c.} \right]\\
&+\cE_{N,t}^{(4)} \end{split} \end{equation}
In the next proposition we collect some important properties of the operator $\cE^{(4)}_{N,t}$.
\begin{prop}\label{prop:L4}
Under the same assumptions as in Theorem \ref{thm:gene}, there exist constants $C, c >0$ such that
\begin{equation}\label{eq:E4-bds} \begin{split} 
\left| \langle \xi, \cE^{(4)}_{N,t} \xi \rangle \right| 
&\leq C e^{c|t|} \| (\cN+1)^{1/2} \xi \| \| (\cV_N + \cN + 1)^{1/2} \xi \| \\
\left| \langle \xi, \left[ \cN, \cE^{(4)}_{N,t} \right] \xi \rangle \right| &\leq C e^{c|t|}  \| (\cN+1)^{1/2} \xi \| \| (\cV_N + \cN + 1)^{1/2} \xi \| \\
\left| \langle \xi , \left[ a^* (g_1) a(g_2) , \cE^{(4)}_{N,t} \right] \xi \rangle \right| &\leq C e^{c|t|} \| g_1 \|_{H^2} \| g_2 \|_{H^2} \| (\cN+1)^{1/2} \xi \| \| (\cV_N + \cN + 1)^{1/2} \xi \| \\
\left| \partial_t \langle \xi, \cE^{(4)}_{N,t} \xi \rangle \right| &\leq C e^{c|t|}  \| (\cN+1)^{1/2} \xi \| \| (\cV_N + \cN + 1)^{1/2} \xi \|
\end{split} \end{equation}
for all $\xi \in \cF^{\leq N}$. 
\end{prop}

\begin{proof}
We start by writing 
\[ e^{-B(\eta_t)} a_x^* a_y^* a_y a_x e^{B(\eta_t)} = a_x^* a_y^* a_y a_x + \int_0^1 ds \, e^{-s B(\eta_t)} \left[ a_x^* a_y^* a_y a_x , B(\eta_t) \right] e^{s B(\eta_t)} \]
A straighforward computation gives 
\begin{equation}\label{eq:L4in} \begin{split}
e^{-B(\eta_t)} &a_x^* a_y^* a_y a_x e^{B(\eta_t)} \\ = \; &a_x^* a_y^* a_y a_x + \int_0^1 ds \, e^{-s B(\eta_t)} \left[ b_x^* b_y^* \, \left( a_x a^* (\eta_y) +  a^* (\eta_x) a_y \right) + \text{h.c.} \right] e^{sB(\eta_t)} \end{split} \end{equation}
Now we observe that
\[ \begin{split} 
e^{-sB(\eta_t)}  & \left[ a_x a^* (\eta_y) +  a^* (\eta_x) a_y \right] e^{sB(\eta_t)} \\ = \; &a_x a^* (\eta_y) +  a^* (\eta_x) a_y + \int_0^s d\tau \, e^{-\tau B(\eta_t)} \left[ a_x a^* (\eta_y) +  a^* (\eta_x) a_y, B(\eta_t) \right] e^{\tau B(\eta_t)} \\ = \; & \eta_t (x;y) + a^* (\eta_y) a_x +  a^* (\eta_x) a_y  \\ &+ \int_0^s d\tau \, e^{-\tau B(\eta_t)} \left[ 2b^* (\eta_x) b^* (\eta_y) + b (\eta^{(2)}_y) b_x + b(\eta^{(2)}_x) b_y \right] e^{\tau B(\eta_t)}
\end{split} \]
Inserting in (\ref{eq:L4in}), expanding as in Lemma \ref{lm:conv-series}, and integrating over $s,\tau$, we obtain 
\begin{equation}\label{eq:L4-sum} 
e^{-B(\eta_t)} \cL^{(4)}_{N,t} e^{B(\eta_t)}  = \cV_N + \text{W}_1 + \text{W}_2 + \text{W}_3 + \text{W}_4 \end{equation}
where 
\[ \begin{split}  
\text{W}_1 &=  \frac{1}{2} \sum_{n,k \geq 0} \frac{(-1)^{n+k}}{n!k!(n+k+1)}  \int dxdy \, N^2 V(N(x-y)) \eta_t (x;y) \,  \text{ad}^{(n)}_{B(\eta_t)} (b_x^*) \text{ad}^{(k)}_{B(\eta_t)} (b_y^*)  \\
\text{W}_2 &= \sum_{n,k \geq 0} \frac{(-1)^{n+k}}{n!k!(n+k+1)}  \int dxdy \, N^2 V(N(x-y)) \, \text{ad}^{(n)}_{B(\eta_t)} (b_x^*) \text{ad}^{(k)}_{B(\eta_t)} (b_y^*)  a^* (\eta_x) a_y  \\
\text{W}_3 &= \sum_{n,k,m,r \geq 0} \frac{(-1)^{n+k+m+r}}{n!k!m!r!(m+r+1)(n+k+m+r+2)} \\ &\hspace{.5cm} \times  \int dx dy \, N^2 V(N(x-y))  \, \text{ad}^{(n)}_{B(\eta_t)} (b_x^*) \text{ad}^{(k)}_{B(\eta_t)} (b_y^*) \text{ad}^{(m)}_{B(\eta_t)} (b (\eta^{(2)}_x)) \text{ad}^{(r)}_{B(\eta_t)} (b_y) \\
\text{W}_4 &= \sum_{n,k,m,r \geq 0} \frac{(-1)^{n+k+m+r}}{n!k!m!r!(m+r+1)(m+r+n+k+2)} \\ &\hspace{.5cm} \times \int dx dy N^2 V(N(x-y)) \, \text{ad}^{(n)}_{B(\eta_t)} (b_x^*) \text{ad}^{(k)}_{B(\eta_t)} (b_y^*) \text{ad}^{(m)}_{B(\eta_t)} (b^* (\eta_x)) \text{ad}^{(r)}_{B(\eta_t)} (b^* (\eta_y)) \end{split} \]
Let us now estimate the expectation of $\text{W}_2$. By Cauchy-Schwarz, we have
\[ \begin{split} 
&\left| \int dx dy N^2 V(N(x-y)) \langle \xi,  \text{ad}^{(n)}_{B(\eta_t)} (b_x^*) \text{ad}^{(k)}_{B(\eta_t)} (b^*_y)  a^* (\eta_x) a_y \xi \rangle \right| \\ &\hspace{.5cm} \leq  \int dx dy N^2 V(N(x-y)) \\ &\hspace{2cm} \times  \| (\cN+1)^{1/2} \text{ad}^{(k)}_{B(\eta_t)} (b_y) \text{ad}^{(n)}_{B(\eta_t)} (b_x) \xi \| \| (\cN+1)^{-1/2} a^* (\eta_x) a_y \xi \| \end{split} \]
We bound 
\begin{equation}\label{eq:W2-1} \| (\cN+1)^{-1/2} a^* (\eta_x) a_y \xi \| \leq \| \eta_x \| \| a_y \xi \|   
\end{equation}
On the other hand, according to Lemma \ref{lm:conv-series}, 
$\| (\cN+1)^{1/2} \text{ad}^{(k)}_{B(\eta_t)} (b_y) \text{ad}^{(n)}_{B(\eta_t)} (b_x) \xi \|$ is bounded by the sum of $2^{n+k} n!k!$ contributions having the form
\begin{equation}\label{eq:norm-T} \begin{split} 
\text{T} = &\left\| (\cN+1)^{1/2} \Lambda_1 \dots \Lambda_{i_1} N^{-k_1} \Pi^{(1)}_{\sharp, \flat} (\eta^{(j_1)}_{t,\natural_1}, \dots , \eta^{(j_{k_1})}_{t,\natural_{k_1}} ; \eta_{y,t,\lozenge}^{(\ell_1)}) \right. \\ &\hspace{4cm} \left. \times  \Lambda'_1 \dots \Lambda'_{i_2} N^{-k_2} \Pi^{(1)}_{\sharp,\flat} (\eta^{(m_1)}_{t,\natural'_1} , \dots, \eta^{(m_{k_2})}_{t,\natural'_{k_2}} ; \eta^{(\ell_2)}_{x,\lozenge'}) \xi \right\| \end{split} \end{equation}
with $i_1, i_2, k_1, k_2, \ell_1, \ell_2 \geq 0$, $j_1, \dots , j_{k_1}, m_1, \dots , m_{k_2} \geq 0$ and where each $\Lambda_i$ and $\Lambda'_i$ operator is either a factor $(N-\cN)/N$, $(N-\cN+1)/N$ or a $\Pi^{(2)}$-operator of the form 
\begin{equation}\label{eq:Pi2-L4}
N^{-p} \Pi^{(2)}_{\underline{\sharp}, \underline{\flat}} (\eta^{(q_1)}_{t,\underline{\natural}_1} ,\dots , \eta^{(q_{p})}_{t,\underline{\natural}_p}) 
\end{equation}
According to Lemma \ref{lm:prel1}, part iv), we have
\begin{equation}\label{eq:Tf} \begin{split} \text{T} &\leq (n+1) C^{k+n} \| \eta_t \|^{k+n-2} \Big\{ \| \eta_x \| \| \eta_y \| \| (\cN+1)^{3/2} \xi \| \\ &\hspace{3cm} + \| \eta_t \| \| \eta_x \| \| a_y (\cN+1) \xi \| + \| \eta_t \| \| \eta_y \| \| a_x (\cN+1) \xi \| \\ &\hspace{3cm} + \| \eta_t \| |\eta_t (x;y)| \| (\cN+1)^{1/2} \xi \| + \| \eta_t \|^2\sqrt N \| a_x a_y \xi \| \Big\}  \end{split} \end{equation}
For $\xi \in \cF^{\leq N}$, we obtain 
\[ \begin{split} 
&\left| \int dx dy N^2 V(N(x-y)) \eta_t (x;y) \langle \xi,  \text{ad}^{(n)}_{B(\eta_t)} (b_x^*) \text{ad}^{(k)}_{B(\eta_t)} (b^*_y)  a^* (\eta_x) a_y \xi \rangle \right| \\ &\hspace{1cm} \leq  (n+1)!k! \, C^{n+k} \| \eta_t \|^{n+k-2}  \int dx dy N^2 V(N(x-y))  \| \eta_x \| \| a_y \xi \| \\ &\hspace{2cm} \times \Big\{ \left[ N \| \eta_x \| \| \eta_y \| + \| \eta_t \| |\eta_t (x;y)| \right] \| (\cN+1)^{1/2} \xi \|  \\ &\hspace{2.8cm} + N \| \eta_t \| \| \eta_y \|  \| a_x \xi \| + N \| \eta_t \| \| \eta_x \| \| a_y \xi \| + N^{1/2}  \|a_x a_y \xi \| \Big\} 
\\ &\hspace{1cm} \leq (n+1)! k! \, C^{n+k} \| \eta_t \|^{n+k} \| (\cN+1)^{1/2} \xi \| \| (\cV_N + \cN + 1 )^{1/2} \xi \| \end{split} \] 
and therefore
\[ |\langle \xi , \text{W}_2 \xi \rangle | \leq C e^{c|t|} \| (\cN+1)^{1/2} \xi \| \| (\cV_N + \cN + 1)^{1/2} \xi \| \]
if $\sup_t \| \eta_t \|$ is small enough.

Now, let us consider the expectation of the term $\text{W}_3$. By Cauchy-Schwarz, we have
\[ \begin{split} & \left| \int dx dy N^2 V(N(x-y)) \langle \xi , \text{ad}^{(n)}_{B(\eta_t)} (b_x^*) \text{ad}^{(k)}_{B(\eta_t)} (b_y^*) \text{ad}^{(m)}_{B(\eta_t)} (b (\eta^{(2)}_x)) \text{ad}^{(r)} (b_y) \xi \rangle \right| \\
&\hspace{2cm} \leq \int N^2 V(N(x-y)) \, \| (\cN+1)^{1/2} \text{ad}^{(k)}_{B(\eta_t)} (b_y)  \text{ad}^{(n)}_{B(\eta_t)} (b_x) \xi \| \\ &\hspace{5cm} \times  \|  (\cN+1)^{-1/2} 
\text{ad}^{(m)}_{B(\eta_t)} (b (\eta^{(2)}_x)) \text{ad}^{(r)} (b_y) \xi \| 
\end{split} \]
Expanding $\text{ad}^{(m)}_{B(\eta_t)} (b (\eta^{(2)}_x)) \text{ad}_{B(\eta_t)} ^{(r)} (b_y)$ as in Lemma \ref{lm:indu} and using Lemma \ref{lm:prel1}, we obtain  
\begin{equation}\label{eq:W3-1} \begin{split} 
\| (\cN+1)^{-1/2} &\text{ad}^{(m)}_{B(\eta_t)} (b (\eta^{(2)}_x)) \text{ad}_{B(\eta_t)} ^{(r)} (b_y) \xi \| \\ &\leq m!r! \, C^{m+r} \| \eta_t \|^{m+r} \left[ \| \eta_x \| \| \eta_y \| \| (\cN+1)^{1/2} \xi \| +  \| \eta_t \| \| \eta_x \| \| a_y \xi \| \right] \end{split} \end{equation}
As for the norm $\| (\cN+1)^{1/2} \text{ad}^{(k)}_{B(\eta_t)} (b_y)  \text{ad}^{(n)}_{B(\eta_t)} (b_x) \xi \|$, we can estimate it as the sum of $2^{n+k} n!k!$ contributions of the form (\ref{eq:norm-T}). Using  (\ref{eq:Tf}) and integrating over $x,y$, we conclude 
\[ |\langle \xi , \text{W}_3 \xi \rangle | \leq C e^{c|t|} \| (\cN+1)^{1/2} \xi \| \| (\cV_N + \cN + 1)^{1/2} \xi \| \]
if $\sup_t \| \eta_t \|$ is small enough. 

Let us now switch to $\text{W}_4$. We proceed analogously as we did for $\text{W}_3$. The only difference is that, instead of (\ref{eq:W3-1}), we need to bound
\[ \begin{split} \| (\cN+1)^{-1/2} \text{ad}^{(m)}_{B(\eta_t)} &(b (\eta_x)) \text{ad}^{(r)}_{B(\eta_t)}  (b (\eta_y)) \xi \| \\ &\leq m!r! \, C^{m+r} \| \eta_t \|^{m+r} \| \eta_x \| \| \eta_y \| \| (\cN+1)^{1/2} \xi \| \end{split} \]
We find 
\[ |\langle \xi , \text{W}_4 \xi \rangle | \leq C e^{c|t|} \| (\cN+1)^{1/2} \xi \| \| (\cV_N + \cN + 1)^{1/2} \xi \| \]
if $\sup_t \| \eta_t \|$ is small enough.

Finally, we consider the term $\text{W}_1$ in (\ref{eq:L4-sum}). We extract from the sum over $n,k \geq 0$ the terms with $(n,k) = (0,0)$ and $(n,k) = (0,1)$. We obtain that 
\begin{equation}\label{eq:W1} \begin{split} \text{W}_1 = \; &\frac{1}{2} \int dx dy N^2 V(N(x-y)) \eta_t (x;y) b_x^* b_y^* \\ &- \frac{1}{4}\int dx dy \, N^2 V(N(x-y)) \eta_t (x;y) [B(\eta_t) , b_x^*] b_y^* + \wt{\text{W}}_1 \end{split} \end{equation}
with 
\begin{equation}\label{eq:wtW10} \wt{\text{W}}_1 =  \;\frac{1}{2} \sum_{n,k}^* \frac{(-1)^{n+k}}{n!k!(n+k+1)} \int dx dy N^2 V(N(x-y)) \eta_t (x;y) \text{ad}^{(n)}_{B(\eta_t)} (b_x^*) \text{ad}^{(k)}_{B(\eta_t)} (b_y^*) \end{equation}
where $\sum^*$ excludes the terms $(n,k) = (0,0), (1,0)$. 
We bound the expectation of $\wt{W}_1$ by 
\[\begin{split} &\left| \int dx dy \, N^2 V(N(x-y)) \eta_t (x;y) \langle \xi, \text{ad}^{(n)}_{B(\eta_t)} (b_x^*) \text{ad}^{(k)}_{B(\eta_t)} (b_y^*) \xi \rangle \right|  \\ &\hspace{.5cm} \leq \int dx dy N^2 V(N(x-y)) |\eta_t (x;y)| \\ &\hspace{2cm} \times \left\| (\cN+1)^{-1/2} \text{ad}^{(k)}_{B(\eta_t)} (b_y) \text{ad}^{(n)}_{B(\eta_t)} (b_x) \xi \right\| \| (\cN+1)^{1/2} \xi \| \end{split} \]
Following Lemma \ref{lm:conv-series}, we can bound the norm $\| (\cN+1)^{-1/2} \text{ad}^{(k)}_{B(\eta_t)} (b_y) \text{ad}^{(n)}_{B(\eta_t)} (b_x) \xi \|$ by the sum of $2^{n+k} n!k!$ terms of the form 
\begin{equation}\label{eq:wtT} \begin{split} \wt{\text{T}} = &\left\| (\cN+1)^{-1/2} \Lambda_1 \dots \Lambda_{i_1} N^{-k_1} \Pi^{(1)}_{\sharp, \flat} (\eta^{(j_1)}_{t,\natural_1}, \dots , \eta^{(j_{k_1})}_{t,\natural_{k_1}} ; \eta_{y,t,\lozenge}^{(\ell_1)}) \right. \\ &\hspace{4cm} \left. \times  \Lambda'_1 \dots \Lambda'_{i_2} N^{-k_2} \Pi^{(1)}_{\sharp,\flat} (\eta^{(m_1)}_{t,\natural_1} , \dots, \eta^{(m_{k_2})}_{t,\natural_{k_2}} ; \eta^{(\ell_2)}_{x,\lozenge'}) \xi \right\| \end{split} \end{equation}
with $i_1, i_2, k_1, k_2, \ell_1, \ell_2 \geq 0$, $j_1, \dots , j_{k_1}, m_1, \dots , m_{k_2} \geq 0$ and where each $\Lambda_i$ and $\Lambda'_i$ operator is either a factor $(N-\cN)/N$, $(N-\cN+1)/N$ or a $\Pi^{(2)}$-operator of the form (\ref{eq:Pi2-L4}). With Lemma \ref{lm:prel1} we find 
\[ \begin{split} 
\wt{\text{T}} &\leq (n+1) C^{k+n} \| \eta_t \|^{k+n-2} \\ &\hspace{1.5cm} \times \Big\{ \| \eta_x \| \| \eta_y \| \| (\cN+1)^{1/2} \xi \| + \| \eta_t \| \| \eta_x \| \| a_y \xi \| + \| \eta_t \| \eta_y \| \| a_x \xi \|  \\ &\hspace{2.5cm} + \| \eta_t \| N^{-1} |\eta_t (x;y)| \| (\cN+1)^{1/2} \xi \| + \| \eta_t \|^2 \| a_x a_y \xi \| \Big\} \end{split} \] 
The important difference with respect to (\ref{eq:Tf}) is that here, when we consider the cases $\ell_1=\ell_2 =0$ and $\ell_1=0,\ell_2 =1$ we can apply (\ref{eq:prelim-iv3}) and (\ref{eq:prelim-iv4}), rather than (\ref{eq:prelim-iv3a}) and (\ref{eq:prelim-iv4a}), because the assumption $(n,k) \not = (0,0), (1,0)$ implies that $k+n \geq 2$ (the case $(n,k) = (0,1)$ is not compatible with $\ell_2 =1$). Using $\sup_{x,y} N^{-1} |\eta_t (x;y)| \leq C e^{c|t|}$ from Lemma \ref{lm:eta}, we conclude that 
\[ |\langle \xi , \wt{W}_1 \xi \rangle| \leq C e^{c|t|} \| (\cN+1)^{1/2} \xi \| \| (\cV_N + \cN + 1)^{1/2} \xi \| \]
if $\sup_t \| \eta_t \|$ is small enough. 

As for the second term on the r.h.s. of (\ref{eq:W1}), we have
\[ [ B(\eta_t) , b_x^* ] = -b (\eta_x) \frac{N-\cN}{N} + \frac{1}{N} \int dz dw a_x^* a_z b_w \, \eta_t (z;w) \]
Hence
\[ \begin{split} -\int dx dy &N^2 V(N(x-y)) \eta_t (x;y) [B(\eta_t), b_x^* ]  b_y^* \\ = \; & 
\int dx dy N^2 V(N(x-y)) \eta_t (x;y) b (\eta_x) b^*_y \frac{N-\cN+1}{N} \\ &- N^{-1} \int dx dy dz dw N^2 V(N(x-y)) \eta_t (x;y) \eta_t (z;w) \,  a_x^* a_z b_w b_y^*  \\ = \; & \int dx dy N^2 V(N(x-y)) |\eta_t (x;y)|^2 \frac{N-\cN}{N} \frac{N-\cN+1}{N} \\ &- N^{-1} \int dx dy dz N^2 V(N(x-y)) \eta_t (x;y) \eta_t (x;z) a_y^* a_z \frac{N-\cN+1}{N} \\ &- N^{-1} \int dx dy dz dw N^2 V(N(x-y)) \eta_t (x;y) \eta_t (z;w) \,  a_x^* a_z b_w b_y^* \end{split} \]
We conclude that 
\[ -\int dx dy N^2 V(N(x-y)) \eta_t (x;y) [B(\eta_t), b_x^* ]  b_y^* = \int dx dy N^2 V(N(x-y)) |k_t (x;y)|^2 + \text{W}_{12} \]
where 
\[ | \langle \xi , \text{W}_{12} \xi \rangle |  \leq C e^{c|t|} \| (\cN+1)^{1/2} \xi \| \| (\cV_N + \cN+1)^{1/2} \xi \| \]
Similarly, the first term on the r.h.s. of (\ref{eq:W1}) can be decomposed as 
\[ \int dx dy N^2 V(N(x-y)) \eta_t (x;y) b_x^* b_y^* = \int dx dy N^2 V(N(x-y)) k_t (x;y) b_x^* b_y^* + \text{W}_{11} \]
where
\[ \text{W}_{11} = \int dx dy N^2 V(N(x-y)) \mu_t (x;y) b_x^* b_y^* \]
is such that 
\[ |\langle \xi , \text{W}_{11} \xi \rangle | \leq C e^{c|t|} \| (\cN+1)^{1/2} \xi \| \| \cV_N^{1/2} \xi \| \]
since $|\mu (x;y)| \leq C e^{c|t|}$ uniformly in $N$.
\end{proof}

\subsection{Analysis of $(i\partial_t e^{-B(\eta_t)}) e^{B(\eta_t)}$} \label{subsec:partialB}

This subsection is devoted to the study of the first term in the generator $\cG_{N,t}$ in (\ref{eq:GNt2}). The properties of $(i\partial_t e^{-B(\eta_t)}) e^{B(\eta_t)}$ are collected in the next proposition.
\begin{prop}\label{prop:partialB} 
Under the same assumptions as in Theorem \ref{thm:gene}, there exist constants $C, c >0$ such that
\begin{equation}\label{eq:bds-partial}
\begin{split}
|\langle \xi , (i\partial_t e^{-B(\eta_t)}) e^{B(\eta_t)} \xi \rangle| &\leq C  \| (\cN+1)^{1/2} \xi \|^2 \\
|\langle \xi , \left[ \cN ,  (i\partial_t e^{-B(\eta_t)}) e^{B(\eta_t)} \right] \xi \rangle| &\leq C  \| (\cN+1)^{1/2} \xi \|^2 \\
|\langle \xi , \left[ a^* (g_1) a(g_2) , (i\partial_t e^{-B(\eta_t)}) e^{B(\eta_t)} \right] \xi \rangle| &\leq C \| g_1 \| \| g_2 \| \| (\cN+1)^{1/2} \xi \|^2 \\
|\langle \xi , \left[ \partial_t (i\partial_t e^{-B(\eta_t)}) e^{B(\eta_t)} \right] \xi \rangle| &\leq C e^{c|t|} \| (\cN+1)^{1/2} \xi \|^2 
\end{split} \end{equation}
for all $\xi \in \cF^{\leq N}$.
\end{prop}

\begin{proof}
As in Section 6.5 of \cite{BDS}, we expand $(i\partial_t e^{-B(\eta_t)}) e^{B(\eta_t)}$ as 
\begin{equation}\label{eq:partial1} \begin{split} 
(i\partial_t e^{-B(\eta_t)}) e^{B(\eta_t)} &= - \int_0^1 ds \, e^{-s B(\eta_t)} \left[ i\partial_t B(\eta_t) \right] e^{s B(\eta_t)}  \\ &= \frac{i}{2} \sum_{k,n \geq 0} \frac{(-1)^{n+k}}{k!n! (n+k+1)} \int dx \, \text{ad}^{(k)}_{B(\eta_t)} (b ((\partial_t \eta_t)_x)) \text{ad}^{(n)}_{B(\eta_t)} (b_x) + \text{h.c.} 
\end{split} \end{equation}
We bound the expectations 
\[ \begin{split}\Big|  \int dx \, \langle \xi , & \text{ad}^{(k)}_{B(\eta_t)} (b ((\partial_t \eta_t)_x)) \text{ad}^{(n)}_{B(\eta_t)} (b_x) \xi \rangle \Big| \\ &\leq \| (\cN+1)^{1/2} \xi \|  \int dx \, \| (\cN+1)^{-1/2}  \text{ad}^{(k)}_{B(\eta_t)} (b ((\partial_t \eta_t)_x)) \text{ad}^{(n)}_{B(\eta_t)} (b_x) \xi \| \end{split} \]
According to Lemma \ref{lm:conv-series}), the norm $\| (\cN+1)^{-1/2}  \text{ad}^{(k)}_{B(\eta_t)} (b ((\partial_t \eta_t)_x)) \text{ad}^{(n)}_{B(\eta_t)} (b_x) \xi \|$ is bounded by the sum of $2^{n+k} n!k!$ terms of the form
\begin{equation}\label{eq:Z} \begin{split} \text{Z} = & \, \| (\cN+1)^{-1/2} \Lambda_1 \dots \Lambda_{i_1} N^{-k_1} \Pi^{(1)}_{\sharp,\flat} (\eta^{(j_1)}_{t,\natural_1}, \dots , \eta^{(j_{k_1})}_{t,\natural_{k_1}} ; (\eta_{t,\lozenge}^{(\ell_1)} \partial_t \eta_t)_x ) \\ 
&\hspace{3cm} \times \Lambda'_1 \dots \Lambda'_{i_2} N^{-k_2} \Pi^{(1)}_{\sharp', \flat'} (\eta^{(m_1)}_{t,\natural'_1} , \dots , \eta^{(m_{k_2})}_{t,\natural'_{k_2}} ; \eta_{x,\lozenge'}^{(\ell_2)}) \xi \| \end{split} \end{equation}
with integers $i_1, k_1, \ell_1, i_2, k_2 , \ell_2 \geq 0$, $j_1, \dots , j_{k_1}, m_1, \dots, m_{k_2} \geq 1$ and where each $\Lambda_i$ and $\Lambda'_i$ is either a factor $(N-\cN)/N$ or $(N+1-\cN)/N$ or a $\Pi^{(2)}$-operator of the form 
\[ N^{-p} \Pi^{(2)}_{\underline{\sharp}, \underline{\flat}} (\eta^{(q_1)}_{t,\underline{\natural}_1}, \dots , \eta^{(q_p)}_{t,\underline{\natural}_p}) \]
{F}rom Lemma \ref{lm:prel1}, part iii), we conclude that
\[ \text{Z} \leq \left\{ \begin{array}{ll} C^{n+k} \| \eta_t \|^{n+k-1} \| (\partial_t \eta_t)_x \| \| \eta_x \| \| (\cN+1)^{1/2} \xi \|  \quad &\text{if } \ell_2  > 0 \\ 
C^{n+k} \| \eta_t \|^{n+k} \| (\partial_t \eta_t)_x \| \| a_x \xi \|  \quad &\text{if } \ell_2  = 0 \end{array} \right. \]
With Cauchy-Schwarz, we obtain
\[ \Big|  \int dx \, \langle \xi , \text{ad}^{(k)}_{B(\eta_t)} (b ((\partial_t \eta_t)_x)) \text{ad}^{(n)}_{B(\eta_t)} (b_x) \xi \rangle \Big| \leq n!k!\,  C^{n+k} \| \eta_t \|^{n+k} \| \partial_t \eta_t \| \| (\cN+1)^{1/2} \xi \|^2 \]
{F}rom (\ref{eq:partial1}), we conclude that, if $\sup_t \| \eta_t \|$ is sufficiently small, 
\[ |\langle \xi , (i\partial_t e^{-B(\eta_t)}) e^{B(\eta_t)} \xi \rangle | \leq C \| (\cN+1)^{1/2} \xi \|^2 \]
The other bounds in (\ref{eq:bds-partial}) can be proven analogously, first expanding $(i\partial_t e^{-B(\eta_t)}) e^{B(\eta_t)}$ as in (\ref{eq:partial1}), then using Lemma \ref{lm:conv-series} and Lemma \ref{lm:indu} to write the nested commutators on the r.h.s. of (\ref{eq:partial1}) as sums of factors like in (\ref{eq:Z}), and then commuting each of these factors with $\cN$, with $a^*(g_1) a(g_2)$, or taking its time-derivative; we omit the details.
\end{proof}

\subsection{Proof of Theorem \ref{thm:gene}}

Combining the results of Subsections \ref{subsec:L0}-\ref{subsec:partialB} and using the scattering equation (\ref{eq:scatlN}), we conclude that
\begin{equation}\label{eq:cGN-1}
\begin{split}
 \cG_{N,t} = \; &C_{N,t} + \cH_N + \wt{\cE}_{N,t} \\ &+ N \int dx dy \left[ -\Delta + \frac{1}{2} N^2 V(N(x-y))  \right] (1-w_\ell (N(x-y))) \wtph_t (x) \wtph_t (y) b_x^* b_y^* + \text{h.c.} \\
 =\; &C_{N,t} + \cH_N  + \wt{\cE}_{N,t} + \text{A}
 \end{split} 
\end{equation}
with 
\[ \text{A} = N^3 \lambda_\ell \int dx dy \, f_\ell (N(x-y)) \chi (|x-y| \leq \ell) \left[ \wtph_t (x) \wtph_t (y) b_x^* b_y^* + \text{h.c.} \right] \]
and where $C_{N,t}$ is defined as in (\ref{eq:CNt}). The error term $\wt{\cE}_{N,t}$ is such that
\begin{equation}\label{eq:bds-cE} 
\begin{split}  \left| \langle \xi , \wt{\cE}_{N,t} \xi \rangle \right|  &\leq C e^{c|t|} \| (\cH_N+\cN+1)^{1/2} \xi \| \| (\cN+1)^{1/2} \xi \|  \\
\left| \langle \xi , \left[ \wt{\cE}_{N,t} , \cN \right] \xi \rangle \right| &\leq C e^{c|t|} \| (\cH_N+\cN+1)^{1/2} \xi \| \| (\cN+1)^{1/2} \xi \| \\
\left| \langle \xi , \left[ \wt{\cE}_{N,t} , a^* (g_1) a(g_2) \right] \xi \rangle \right| &\leq C e^{c|t|} \| g_1 \|_{H^2} \| g_2 \|_{H^2} \| (\cH_N+\cN+1)^{1/2} \xi \| \| (\cN+1)^{1/2} \xi \| \\
\left| \langle \xi , \left[ \partial_t  \wt{\cE}_{N,t}\right] \xi \rangle \right| &\leq  C e^{c|t|} \| (\cH_N+\cN+1)^{1/2} \xi \| \| (\cN+1)^{1/2} \xi \| \end{split} 
\end{equation}
Since $N^3\lambda_\ell \leq C$ and $f_\ell (N(x-y)) \leq 1$, we have , with Lemma \ref{lm:bbds}, 
\[ | \langle \xi, \text{A} \xi \rangle | \leq C \| (\cN+1)^{1/2} \xi \|^2 \]
and similarly, $\pm [\cN , \text{A}], \pm [ a^* (g_1) a(g_2) , \text{A}], \pm \partial_t \text{A} \leq C (\cN+1)$. Setting $\cE_{N,t} = \text{A} + \wt{\cE}_{N,t}$, we conclude that
\[ \cG_{N,t} = C_{N,t} + \cH_N + \cE_{N,t} \]
where $\cE_{N,t}$ satisfies the same bounds (\ref{eq:bds-cE}) as $\wt{\cE}_{N,t}$. This immediately implies that, in the sense of forms on $\cF^{\leq N}_{\perp \wtph_t} \times \cF^{\leq N}_{\perp \wtph_t}$,  
\[ \begin{split} 
\frac{1}{2} \cH_N - C e^{c|t|} (\cN+1) &\leq \cG_{N,t} - C_{N,t}  \leq 2 \cH_N + C e^{c|t|} (\cN+1) \\ \pm i\left[ \cG_{N,t} , \cN \right] &\leq \cH_N + C e^{c|t|}(\cN + 1) \\ 
\partial_t  \left[ \cG_{N,t} - C_{N,t} \right] &\leq \cH_N + C e^{c|t|} (\cN+1)  
 \end{split} \]
Moreover, since
\[ \begin{split} 
[\cH_N , a^* (g_1) a(g_2) ] = \; &\int dx \nabla g_1 (x) \nabla_x a_x^* a (g_2) - \int dx a^* (g_1) \nabla \bar{g}_2 (x) \nabla_x a_x \\ &+ \int dx dy \, N^2 V(N(x-y)) g_1 (y) a_x^* a_y^* a_x a (g_2) \\ &- \int dx dy \, N^2 V(N(x-y)) \, \bar{g}_2 (x) a^* (g_1) a_y^* a_y a_x \end{split} \]
we obtain that
\[ \begin{split} 
|\langle \xi , &[\cH_N , a^* (g_1) a(g_2) ] \xi \rangle | \\ \leq & \left[ \| \nabla g_1 \| \| g_2 \| + \| g_1 \| \| \nabla g_2 \| \right]  \, \| \cK^{1/2} \xi \| \| \cN^{1/2} \xi \| \\ &+ \left[ \| g_2 \| \| g_1 \|_\infty + \| g_1 \| \| g_2 \|_\infty \right] 
\left[ \int dx dy N^2 V(N(x-y)) \| a_x a_y \xi \|^2 \right]^{1/2}
\\ & \hspace{3cm} \times \left[ \int dx dy N^2 V(N(x-y)) \| a_y (\cN+1)^{1/2} \xi \|^2 \right]^{1/2} \\ \leq \; & \| g_1 \|_{H^2} \| g_2 \|_{H^2} \| \cH_N^{1/2} \xi \| \| (\cN+1)^{1/2} \xi \| \end{split} \]
for all $\xi \in \cF^{\leq N}$. Combining with (\ref{eq:bds-cE}), and choosing $g_1 = \partial_t \wtph_t$ and $g_2 = \wtph_t$, we find 
\[ \pm \text{Re }  \left[ \cG_{N,t} , a^* (\partial_t \wtph_t) a(\wtph_t) \right]  \leq \cH_N + C e^{K|t|} (\cN+1) \]
This concludes the proof of Theorem \ref{thm:gene}.

\section{Bounds on the Growth of Fluctuations} 
\label{sec:main}

In this section, we are going to complete the proof of Theorem \ref{thm:main} and of Theorem \ref{thm:main2}. 
The main ingredient to reach this goal is a bound on the growth of the expectation of the number of particles operator with respect to the fluctuation dynamics $\cW_{N,t}$, that we prove in the next proposition using the properties of the generator $\cG_{N,t}$ established in Theorem \ref{thm:gene}. 
\begin{prop}\label{prop:growN}
Under the same assumptions as in Theorem \ref{thm:gene}, there exist constants $C,c > 0$ such that 
\begin{equation}\label{eq:growN} \begin{split} \langle \cW_{N,t} \, \xi , \cN \cW_{N,t} \xi \rangle &\leq C \, \langle \xi , (\cG_{N,0} - C_{N,0}) + (\cN+1)) \xi \rangle \exp (c \exp (c|t|) \\ 
\langle \cW_{N,t} \, \xi , \cH_N \cW_{N,t} \xi \rangle &\leq C \, \langle \xi , (\cG_{N,0} - C_{N,0}) + (\cN+1)) \xi \rangle \exp (c \exp (c|t|)
\end{split} \end{equation}
for all $\xi \in \cF^{\leq N}_{\perp \varphi}$. Here $\cH_N$ is the Hamilton operator defined in (\ref{eq:HN-thm}). 
\end{prop}
{\it Remark:} From (\ref{eq:gene-bds}), we also have 
\[ \begin{split} \langle \cW_{N,t} \, \xi , \cN \cW_{N,t} \xi \rangle &\leq C \, \langle \xi , (\cH_N + \cN+1) \xi \rangle \exp (c \exp (c|t|) \\ 
\langle \cW_{N,t} \, \xi , \cH_N \cW_{N,t} \xi \rangle &\leq C \, \langle \xi , (\cH_N + \cN+1) \xi \rangle \exp (c \exp (c|t|)
\end{split} \]

\begin{proof}
First of all, we observe that, from the first bound in (\ref{eq:gene-bds}), 
\begin{equation}\label{eq:HN-G} \frac{1}{2} \cH_N + \cN \leq (\cG_{N,t} - C_{N,t}) + C e^{K|t|} (\cN+1) \end{equation}
Hence, it is enough to control the growth of the expectation of the operator on the right hand side. We follow here the approach of \cite{LNS}. We define $q_t = 1-|\wtph_t \rangle \langle \wtph_t|$ as the orthogonal projection onto $L^2_{\perp \wtph_t} (\bR^3)$. We define moreover $\Gamma_t : \cF^{\leq N} \to \cF^{\leq N}_{\perp \ph_t}$ by imposing that  $\Gamma_t |_{\cF_j} = q_t^{\otimes j}$ for all $j=1, \dots , N$ ($\cF_j$ is the sector of $\cF^{\leq N}$ with exactly $j$ particles). We have, restricting our attention to $t \geq 0$ (the case $t < 0$ can be handled very similarly) 
\[ \begin{split} 
&\left\langle  \cW_{N,t} \, \xi , \left[ (\cG_{N,t} - C_{N,t}) + C e^{K t} (\cN+1) \right] \cW_{N,t} \, \xi \right\rangle \\ &\hspace{4cm} = \left\langle  \cW_{N,t} \, \xi , \left[ (\Gamma_t \cG_{N,t} \Gamma_t - C_{N,t}) + C e^{Kt} (\cN+1) \right] \cW_{N,t} \, \xi \right\rangle \end{split} \]
Hence, since $\cN$ commutes with $\Gamma_t$, 
\begin{equation}\label{eq:idelW} \begin{split}  i\partial_t &\left\langle  \cW_{N,t} \, \xi , \left[ (\cG_{N,t} - C_{N,t}) + C e^{Kt} (\cN+1) \right] \cW_{N,t} \, \xi \right\rangle \\  &\hspace{1cm} = \left\langle \cW_{N,t} \, \xi , \left[ \Gamma_t \cG_{N,t} \Gamma_t , (\Gamma_t \cG_{N,t} \Gamma_t - C_{N,t}) + C e^{Kt} (\cN+1) \right] \cW_{N,t} \, \xi \right\rangle
\\ &\hspace{1.4cm} + \left\langle \cW_{N,t} \, \xi, \partial_t \left[ (\Gamma_t \cG_{N,t} \Gamma_t - C_{N,t}) + C e^{Kt} (\cN+1) \right] \cW_{N,t} \, \xi \right\rangle\\  &\hspace{1cm} = C e^{Kt} \left\langle \cW_{N,t} \, \xi , \left[ \cG_{N,t}, \cN  \right] \cW_{N,t} \, \xi \right\rangle
\\ &\hspace{1.4cm} + \left\langle \cW_{N,t} \, \xi, \partial_t \left[ (\Gamma_t \cG_{N,t} \Gamma_t - C_{N,t}) + C e^{Kt} (\cN+1) \right] \cW_{N,t} \, \xi \right\rangle 
 \end{split} \end{equation}
We observe that
\[ 0 = \partial_t \| \wtph_t \|_2^2 = \langle \dot{\wtph}_t, \wtph_t \rangle + \langle \wtph_t , \dot{\wtph}_t  \rangle \]
This implies that
\[ \dot{q}_t = - |\wtph_t \rangle \langle \dot{\wtph}_t| - |\dot{\wtph}_t \rangle \langle \wtph_t| = - |\wtph_t \rangle\langle q_t \dot{\wtph}_t | - |q_t \dot{\wtph}_t \rangle \langle \wtph_t| \]
Therefore  
\[\begin{split} \partial_t \Gamma_t^{(j)} &= - \sum_{i=1}^j q_t \otimes \dots \otimes \left[ |\wtph_t \rangle \langle q_t \dot{\wtph}_t| q_t + q_t | q_t \dot{\wtph}_t \rangle \langle \wtph_t| \right] \otimes \dots \otimes q_t \\ & = - \sum_{i=1}^j \left[ |\wtph_t \rangle \langle q_t \dot{\wtph}_t |_i  \Gamma_t^{(j)} - \Gamma_t^{(j)} |q_t \dot{\wtph}_t \rangle \langle \wtph_t|_i \right]  \end{split}
\]
We conclude that
\[ \partial_t \Gamma_t = - a^* (\wtph_t) a(q_t \dot{\wtph}_t) \Gamma_t - \Gamma_t a^* (q_t \dot{\wtph}_t) a (\wtph_t) \]
Thus 
\[ \begin{split}
&\left\langle \cW_{N,t} \, \xi, \partial_t \left[ (\Gamma_t \cG_{N,t} \Gamma_t - C_{N,t}) + C e^{Kt} (\cN+1) \right] \cW_{N,t} \, \xi \right\rangle  \\ &\hspace{2cm} = \left\langle \cW_{N,t} \, \xi, \left[ (\partial_t \Gamma_t) (\cG_{N,t} - C_{N,t}) + (\cG_{N,t} - C_{N,t}) (\partial_t \Gamma_t) \right] \cW_{N,t} \, \xi \right\rangle \\ &\hspace{2.4cm} + \left\langle \cW_{N,t} \, \xi,  \left[ \partial_t (\cG_{N,t} - C_{N,t}) + C K e^{Kt} (\cN+1) \right] \cW_{N,t} \, \xi \right\rangle  \\ &\hspace{2cm} = 2 \text{Re} \left\langle \cW_{N,t} \, \xi, \left[ a^* (q_t \dot{\wtph}_t) a (\wtph_t) , \cG_{N,t} \right] \cW_{N,t} \, \xi \right\rangle \\ 
 &\hspace{2.4cm} + \left\langle \cW_{N,t} \, \xi,  \left[ \partial_t (\cG_{N,t} - C_{N,t}) + C K e^{Kt} (\cN+1) \right] \cW_{N,t} \, \xi \right\rangle 
\end{split} \] 
where we used the fact that $a(\wtph_t) \cW_{N,t} \xi = 0$, for all $t \in \bR$. Together with (\ref{eq:idelW}), we find
\[ \begin{split} 
i\partial_t &\left\langle  \cW_{N,t} \, \xi , \left[ (\cG_{N,t} - C_{N,t}) + C e^{Kt} (\cN+1) \right] \cW_{N,t} \, \xi \right\rangle\\  &\hspace{1cm} = C e^{Kt} \left\langle \cW_{N,t} \, \xi , \left[ \cG_{N,t}, \cN  \right] \cW_{N,t} \, \xi \right\rangle \\ &\hspace{1.4cm} + \left\langle \cW_{N,t} \, \xi,  \left[ \partial_t (\cG_{N,t} - C_{N,t}) + C K e^{Kt} (\cN+1) \right] \cW_{N,t} \, \xi \right\rangle\\ &\hspace{1.4cm} +2 \text{Re} \left\langle \cW_{N,t} \, \xi, \left[ a^* (q_t \dot{\wtph}_t) a (\wtph_t) , \cG_{N,t} \right] \cW_{N,t} \, \xi \right\rangle  \end{split} \]
{F}rom Theorem \ref{thm:gene}, we obtain that 
\[ \begin{split} &\left|\partial_t  \left\langle  \cW_{N,t} \, \xi , \left[ (\cG_{N,t} - C_{N,t}) + C e^{Kt} (\cN+1) \right] \cW_{N,t} \, \xi \right\rangle \right| \\ &\hspace{3cm} \leq \wt{C} e^{K|t|} \left\langle \cW_{N,t} \, \xi , \left[ \cH_N + C e^{Kt} (\cN + 1) \right] \cW_{N,t} \xi \right\rangle 
 \\ &\hspace{3cm} \leq \wt{C} e^{K|t|} \left\langle \cW_{N,t} \, \xi , \left[  (\cG_{N,t} - C_{N,t}) + C e^{K|t|} (\cN+1) \right] \cW_{N,t} \xi \right\rangle  \end{split} \]
Applying Gronwall's inequality, we find a constant $c > 0$ such that 
\[ \begin{split} 
&\left\langle \cW_{N,t} \, \xi , \left[  (\cG_{N,t} - C_{N,t}) + C e^{K|t|} (\cN+1) \right] \cW_{N,t} \xi \right\rangle \\ &\hspace{3cm} \leq \left\langle \xi, \left[ (\cG_{N,0} - C_{N,0}) + C (\cN + 1) \right] \xi \right\rangle \, \exp (c \exp (c |t|)) 
\end{split} \]
With (\ref{eq:HN-G}), we conclude that
\[ \begin{split}
\langle \cW_{N,t} \xi, \cN \cW_{N,t} \xi \rangle &\leq
C \left\langle \xi, \left[ (\cG_{N,0} - C_{N,0}) + (\cN + 1) \right] \xi \right\rangle \, \exp (c \exp (c |t|)) \\
\langle \cW_{N,t} \xi, \cH_N \cW_{N,t} \xi \rangle &\leq
C \left\langle \xi, \left[ (\cG_{N,0} - C_{N,0}) + (\cN + 1) \right] \xi \right\rangle \, \exp (c \exp (c |t|))
\end{split} \]
as claimed.
\end{proof}

To apply Prop. \ref{prop:growN} to the proof of Theorems \ref{thm:main} and \ref{thm:main2}, we need to control the expectation on 
the r.h.s. of (\ref{eq:growN}) for vectors $\xi \in \cF_{\perp \ph}^{\leq N}$ describing orthogonal excitations around the condensate wave function $\ph$ for initial $N$-particle wave functions $\psi_N$ satisfying (\ref{eq:assN}). To this end, we use the next lemma.
\begin{lemma}\label{lm:CN}
As in (\ref{eq:CNt}), let
\[ \begin{split} C_{N,t} =& \; \frac{1}{2} \left\langle \wtph_t, \left( [N^3V(N.) (N-1-2N f_{\ell} (N.))] * |\wtph_t|^2 \right)\wtph_t \right\rangle \\
	& + \int dxdy\, \left|\nabla_x k_t (x;y) \right|^2+ \frac{1}{2}\int dx dy\, N^2 V(N(x-y)) |k_t (x;y)|^2 \\
	&+ \operatorname{Re} \int dxdy\,  N^3V(N(x-y))\bar{\wtph_t} (x)\bar{\wtph_t} (y)k_t (x;y).
	\end{split}\]
where $\wtph_t$ is the solution of the modified Gross-Pitaevskii equation (\ref{eq:GPmod}), with initial data $\wtph_{t=0} = \ph$ (we assumed in the construction of the fluctuation dynamics that $\ph \in H^4 (\bR^3)$; in this lemma, we only need $\ph \in H^1 (\bR^3)$). Then there is a constant $C > 0$, independent of $N$ and $t$, such that 
\[ \left|  \left[ C_{N,t} + N \langle i\partial_t \wtph_t , \wtph_t \rangle\right] -  N \cE_\text{GP} (\ph) \right| \leq C  \]
with the translation invariant Gross-Pitaevskii energy functional $\cE_\text{GP}$ defined in (\ref{eq:GP-func}).
\end{lemma}
\begin{proof}
We have
\[ N\product{i\partial_t \wtph_t}{\wtph_t} = N\product{\wtph_t}{-\Delta\wtph_t} + N \product{\wtph_t}{\lb(N^3V(N.) f_\ell (N.) *|\wtph_t|^2\rb)\wtph_t}.  \]
    Therefore  
    \begin{equation}\label{eq:CNt-bd}  \begin{split} 
    C_{N,t} + N&\langle i\partial_t \wtph_t , \wtph_t \rangle \\ = 
&\; N \| \nabla \wtph_t \|^2 + \frac{(N-1)}{2} \langle \wtph_t , \left[ N^3 V(N.) * |\wtph_t|^2 \right] \wtph_t \rangle 
     \\ &+ \int dxdy\, \left|\nabla_x k_t (x;y) \right|^2+ \frac{1}{2}\int dx dy\, N^2 V(N(x-y)) |k_t (x;y)|^2 \\
	&+ \operatorname{Re} \int dxdy\,  N^3V(N(x-y))\bar{\wtph_t} (x)\bar{\wtph_t} (y)k_t (x;y).
	\end{split}\end{equation}
Obviously,
\begin{equation}\label{eq:aux1} 	\frac{(N-1)}{2} \langle \wtph_t , \left[ N^3 V(N.) * |\wtph_t|^2 \right] \wtph_t \rangle =  \frac{N}{2} \left\langle \wtph_t , \left[ N^3 V(N.) * |\wtph_t|^2 \right] \wtph_t \right\rangle  + \mathcal{O} (1) \end{equation}
where $\mathcal{O} (1)$ denotes a quantity with absolute value bounded by a constant, independent of $N$ and of $t$. 
Furthermore 
  \begin{equation}\label{A1.auxiliary3eq2}
    \begin{split}
    \frac12 \int dxdy\, &N^2V(N(x-y)) | k_t (x,y)|^2 \\
    &=\frac N2 \int dxdy\, N^3V(N(x-y)) w_{\ell} (N(x-y))^2 |\wtph_t (x)|^2 |\wtph_t (y)|^2 
     \end{split} 
    \end{equation}
Finally, we consider the third term on the r.h.s. of (\ref{eq:CNt-bd}), the one with $\nabla_x k_t$. We recall that $k_t (x;y) = -N w_\ell (N(x-y))\wtph_t (x) \wtph_t (y)$. Hence, we find
    \begin{equation}\label{A1.auxiliary3eq5}
    \begin{split}
    -\Delta_x k_t (x;y) = \;& N^3 (\Delta w_\ell)(N(x-y))\wtph_t (x)\wtph_t (y)+Nw_\ell (N(x-y))\Delta\wtph_t (x) \wtph_t (y) \\
    &  + 2 N^2 (\nabla w_\ell) (N(x-y))\cdot\nabla\wtph_t (x)\wtph_t (y).
    \end{split}
    \end{equation}
Since, by (\ref{eq:scatl}), $\Delta w_\ell = - \Delta f_\ell = -(1/2) V f_\ell  + \lambda_\ell f_\ell$ we have
\begin{equation}\label{A1.auxiliary3eq7}
    \begin{split}
     &\int dxdy \, \bar{k}_t (x;y) (-\Delta_x  k_t)(x;y) \\
     &= -\frac N2 \int dx dy\, N^3 V(N(y-x))(w_\ell (N(x-y))-1) w_\ell (N(x-y)) |\wtph_t (x)|^2 |\wtph_t (y)|^2 \\
     &\quad - N^3 \lambda_\ell \int dx dy\,f_\ell (N(x-y)) N w_\ell (N(x-y))\ |\wtph_t (x)|^2 |\wtph_t (y)|^2 \\
     & \quad + 2 \int dx dy\, N w_\ell (N(y-x))N^2 (\nabla w_\ell)(N(y-x))\cdot\nabla \bar{\wtph}_t (x) \wtph_t (x) |\wtph_t (y)|^2 \\
     & \quad -\int dx dy\, N^2 w_\ell^2 (N(x-y)) 
     (\Delta\wtph_t )(x)\wtph_t (x) |\wtph_t (y)|^2\\
     &= \frac N2 \int dx dy\, N^3 V(N(y-x))(1-w_\ell (N(x-y))) w_\ell (N(x-y)) |\wtph_t (x)|^2 |\wtph_t (y)|^2 \\
     & \quad+ 2 \int dx dy\, N w_\ell (N(y-x)) N^2 (\nabla w_\ell)(N(y-x))\cdot\nabla\bar{\wtph}_t (x)\wtph_t (x) |\wtph_t(y)|^2 + \mathcal{O} (1).
    \end{split}
    \end{equation}
In the last step, we used the bounds $N^3 \lambda_\ell = \mathcal{O} (1)$, $Nw_\ell (N(x-y))\leq C|x-y|^{-1}$ and  $0 \leq f_\ell (N(x-y)) \leq 1$. Integrating by parts in the last term, we find 
\[ \begin{split}
     & 2 \int dx dy\, N^2 (\nabla w_\ell)(N(y-x))\cdot\nabla\bar{\wtph}_t (x) Nw_\ell (N(y-x))\wtph_t (x) |\wtph_t (y)|^2\\
     & = - \int dx dy\, \nabla_x (N^2 w_\ell (N(y-x))^2)\cdot\nabla\bar{\wtph}_t (x)\wtph_t (x) |\wtph_t (y)|^2\\
     & = \int dx dy\, N^2w_\ell (N(x-y))^2\Delta\bar{\wtph}_t (x)\wtph_t (x)  |\wtph_t (y)|^2\\
     &\quad+ \int dx dy\, N^2 w_\ell (N(x-y))^2\nabla\bar{\wtph}_t (x)\cdot\nabla\wtph_t (x)  |\wtph_t (y)|^2\end{split} 
\]    
With (\ref{A1.auxiliary3eq7}), this leads us (using again the bound $Nw_\ell (N (x-y)) \leq C|x-y|^{-1}$) to 
\[ \begin{split} \int dx &dy \, \bar{k}_t (x;y) (-\Delta_x k_t) (x;y) \\ = &\; \frac{N}{2} \int dx dy \, N^3 V(N(y-x))(1-w_\ell (N(x-y))) w_\ell (N(x-y)) |\wtph_t (x)|^2 |\wtph_t (y)|^2  \\ &+ \mathcal{O} (1) \end{split} \]
Combining this bound with (\ref{eq:aux1}) and (\ref{A1.auxiliary3eq2}), we find 
\[ \begin{split} C_{N,t} & + N \langle i\partial_t \wtph_t, \wtph_t \rangle \\ = \; &N \left[ \int |\nabla \wtph_t (x)|^2 dx + \frac{1}{2} \int dx dy N^3 V(N(x-y)) f_\ell (N(x-y)) |\wtph_t (x)|^2 |\wtph_t (y)|^2 \right] \\ &+ \mathcal{O} (1) \end{split} \]
The expression in the parenthesis on the r.h.s. is exactly the energy functional associated with the time-dependent modified Gross-Pitaevskii equation (\ref{eq:GPmod}). By energy conservation, we conclude that
\begin{equation}\label{eq:CNt-GP} \begin{split} C_{N,t} & + N \langle i\partial_t \wtph_t, \wtph_t \rangle \\ = \; &N \left[ \int |\nabla \ph (x)|^2 dx + \frac{1}{2} \int dx dy N^3 V(N(x-y)) f_\ell (N(x-y)) |\ph (x)|^2 |\ph (y)|^2 \right] \\ &+ \mathcal{O} (1) \end{split} \end{equation}
Observe that, with (\ref{eq:Vfa0}),  
\begin{equation}\label{eq:GPfun-mod} \begin{split} \int dx &dy \, N^3 V(N(x-y)) f_\ell (N(x-y)) |\ph (x)|^2 |\ph (y)|^2 \\ =\; & \int dx dy V(y) f_\ell (y) |\ph (x)|^2 |\ph (x+y/N)|^2 \\ = \; &\left[ 8\pi a_0 + \mathcal{O} (N^{-1}) \right] \, \int |\ph (x)|^4 dx \\ &+ \int dx dy V(y) f_\ell (y) |\ph (x)|^2 \left[ |\ph (x+y/N)|^2 - |\ph (x)|^2 \right] \end{split} \end{equation}
where
\[\begin{split} \Big| \int dx dy &\, V(y) f(y)  |\ph (x)|^2 \left[ |\ph (x+y/N)|^2 - |\ph (x)|^2 \right] \Big| \\
&\leq N^{-1} \int_0^1 ds \int dx dy \, V(y) f(y) |\ph(x)|^2 |\nabla \ph (x+sy/N)| |\ph (x+y/N)| |y| \\
&\leq C N^{-1} \end{split} \]
for a constant $C > 0$ depending only on the $H^1$-norm of $\ph$. Inserting the last bound and (\ref{eq:GPfun-mod}) in (\ref{eq:CNt-GP}), we conclude that 
\[ C_{N,t}  + N \langle i\partial_t \wtph_t, \wtph_t \rangle = N\cE_{GP} (\ph) + \mathcal{O} (1) \]
as claimed.
\end{proof}

With Proposition \ref{prop:growN} and Lemma \ref{lm:CN}, we can now conclude the proof of our main theorems. 

\begin{proof}[Proof of Theorems \ref{thm:main} and \ref{thm:main2}]
We observe, first of all, that, by Proposition \ref{prop:phph}, 
\begin{equation}\label{eq:wtph-ph}
\left| \langle \ph_t , \gamma^{(1)}_{N,t} \ph_t \rangle - \langle \wtph_t , \gamma_{N,t}^{(1)} \wtph_t \rangle \right| \leq 2 \| \ph_t - \wtph_t \| \leq C N^{-1} \exp (c \exp (c|t|)) 
\end{equation}
Hence, it is enough to compute 
\[ \begin{split} \langle \wtph_t , \gamma^{(1)}_{N,t} \wtph_t \rangle &= \frac{1}{N} \langle e^{-iH_N t} \psi_N , a^* (\wtph_t) a(\wtph_t) e^{-iH_N t} \psi_N \rangle \\
&= \frac{1}{N} \langle U_{N,t} e^{-iH_N t} \psi_N, (N-\cN) U_{N,t} e^{-iH_N t} \psi_N \rangle \\
&= 1 - \frac{1}{N} \langle U_{N,t} e^{-iH_N t} \psi_N, \cN U_{N,t} e^{-iH_N t} \psi_N \rangle
\end{split} \]
We define $\xi = e^{-B(\eta_0)} U_{N,0} \psi_N \in \cF_{\perp \ph}^{\leq N}$. Then we have $\psi_N = U_{N,0}^* e^{B(\eta_0)} \xi$ and therefore
\[ \begin{split} 1-\langle \wtph_t , \gamma^{(1)}_{N,t} \wtph_t \rangle &= \frac{1}{N} \langle \cW_{N,t} \xi, e^{-B(\eta_t)} \cN e^{B(\eta_t)} \cW_{N,t} \xi \rangle \leq \frac{C}{N} \langle \cW_{N,t} \xi, \cN \cW_{N,t} \xi \rangle
\end{split} \]
where we applied Lemma \ref{lm:Npow}. By Prop. \ref{prop:growN}, we conclude that 
\begin{equation}\label{eq:1-pgp} \begin{split} 1-\langle \wtph_t , \gamma^{(1)}_{N,t} \wtph_t \rangle  \leq N^{-1} \exp (c \exp (c |t|)) \, \langle \xi, \left[ (\cG_{N,0} - C_{N,0}) + C (\cN+1) \right] \xi \rangle \end{split} \end{equation}
In order to apply Prop. \ref{prop:growN}, we used here the assumption (valid in the proof of both Theorem \ref{thm:main} and Theorem \ref{thm:main2}) that $\wt{\ph}_{t=0} = \ph \in H^4 (\bR^3)$.

Recalling from (\ref{eq:assN}) the definition $a_N = 1 - \langle \ph, \gamma^{(1)}_N \ph \rangle$, we bound, with the above definition of $\xi$, 
\[ \begin{split} \langle \xi, \cN \xi \rangle &= \langle U_{N,0} \psi_N , e^{B(\eta_0)} \cN e^{-B(\eta_0)} U_{N,0} \psi_N \rangle \\ &\leq C\langle U_{N,0} \psi_N , \cN U_{N,0} \psi_N \rangle \\ &= C \langle \psi_N, (N- a^* (\ph) a(\ph)) \psi_N \rangle \\& = C N (1 - \langle \ph , \gamma^{(1)}_{N} \ph \rangle ) = C N a_N \end{split} \] 
We still have to bound the expectation of $(\cG_{N,0} - C_{N,0})$ in the state 
$\xi$. We have
\[ \cG_{N,0} = i\partial_t e^{-B(\eta_t)} |_{t=0} e^{B(\eta_0)} + e^{-B (\eta_0)} \left[ (i\partial_t U_{N,t})|_{t=0} U_{N,0}^* + U_{N,0} H_N U_{N,0}^* \right] e^{B(\eta_0)} \]
With Proposition \ref{prop:partialB}, we find  
\begin{equation}\label{eq:bg-GN1} \left| \langle \xi, i\partial_t e^{-B(\eta_t)} |_{t=0} e^{B(\eta_0)} \xi \rangle \right| \leq C \langle \xi,  (\cN+1)\xi \rangle \leq CNa_N + C \end{equation}
From Eq. (\ref{eq:delUU}), we obtain  
\[ \begin{split} \langle e^{B(\eta_0)} \xi, &\, (i\partial_t U_{N,t}) |_{t=0} \, U_{N,0}^* \, e^{B(\eta_0)} \xi \rangle \\ = \; &- \langle (i\partial_t \wtph_t)|_{t=0} , \ph \rangle \langle U_{N,0} \psi_N , (N-\cN) U_{N,0} \psi_N \rangle \\ &- 2 \text{Re} \langle U_{N,0} \psi_N, \sqrt{N-\cN} a(q_0 (i\partial_t \wtph_t) |_{t=0}) U_{N,0} \psi_N \rangle \\ = \; & -N \langle (i\partial_t \wtph_t)|_{t=0} , \ph \rangle + N \langle (i\partial_t \wtph_t)|_{t=0} , \ph \rangle (1 - \langle \ph, \gamma_N^{(1)} \ph \rangle ) \\ &-2N\text{Re } \langle \ph, \gamma_N^{(1)} q_0 (i\partial_t \wtph_t)|_{t=0} \rangle \end{split} \]
Combining this identity with the bound (\ref{eq:bg-GN1}) and with the observation that, by definition of $\xi$, 
\[ \langle \xi , e^{-B(\eta_0)} U_{N,0} H_N U_{N,0}^* e^{B(\eta_0)} \xi \rangle = \langle \psi_N, H_N \psi_N \rangle \]
we conclude that 
\[ \begin{split} \langle \xi, (\cG_{N,0} - C_{N,0}) \xi \rangle \leq \; &\big[ \langle \psi_N, H_N \psi_N \rangle - (C_{N,0} + N \langle (i\partial_t \wt{\ph}_t)|_{t=0} , \ph \rangle  ) \big]  \\ & -2N \text{Re } \langle \ph, \gamma_N^{(1)} q_0 (i\partial_t \wt{\ph}_t)|_{t=0} \rangle +C N a_N + C \end{split} \]
Hence, with Lemma \ref{lm:CN}, we get
\begin{equation}\label{eq:G0-C0f} \begin{split} \langle \xi, (\cG_{N,0} - C_{N,0}) \xi \rangle \leq \; &\big[ \langle \psi_N, H_N \psi_N \rangle - N \cE_\text{GP} (\ph) \big]  -2N \text{Re } \langle \ph, \gamma_N^{(1)} q_0 (i\partial_t \wt{\ph}_t)|_{t=0} \rangle \\ & +C N a_N + C \end{split} \end{equation}
where $\cE_\text{GP}$ denotes the translation invariant Gross-Pitaevskii functional defined in (\ref{eq:GP-func}). 

To bound the second term on the r.h.s. of the last equation, we proceed differently depending on whether we want to show Theorem \ref{thm:main} or Theorem \ref{thm:main2}. To prove Theorem~\ref{thm:main2}, we notice that
\[ \begin{split} \langle \ph,\gamma_N^{(1)} q_0 (i\partial_t \wtph_t)|_{t=0} \rangle  =&\; \langle \ph,\gamma_N^{(1)} (i\partial_t \wtph_t)|_{t=0} \rangle - \langle \ph,\gamma_N^{(1)} \ph \rangle \langle \ph , (i\partial_t \wtph_t)|_{t=0} \rangle \\ 
 = &\; \langle \ph, (i\partial_t \wtph_t)|_{t=0} \rangle (1-\langle \ph, \gamma_N^{(1)} \ph \rangle) + \langle \ph, (\gamma^{(1)} - |\ph \rangle \langle \ph|) (i\partial_t \wtph_t)|_{t=0} \rangle 
\end{split} \]
With $\wt{a}_N = \tr \, |\gamma^{(1)}_N - |\ph \rangle \langle \ph| |$, we obtain that
\[ |\langle \ph,\gamma_N^{(1)} q_0 (i\partial_t \wtph_t)|_{t=0} \rangle| \leq C (a_N + \wt{a}_N) \]
Since $a_N \leq \wt{a}_N$, we conclude from (\ref{eq:G0-C0f}) that 
\[ \langle \xi , (\cG_{N,0} - C_{N,0}) \xi \rangle \leq C \left[ N\wt{a}_N + N \wt{b}_N + 1 \right] \] 
Inserting in (\ref{eq:1-pgp}) and using (\ref{eq:wtph-ph}), we arrive at
\[ 1 - \langle \ph_t, \gamma^{(1)}_N \ph_t \rangle \leq C \big[ \wt{a}_N + \wt{b}_N + N^{-1} \big] \exp (c \exp (c |t|)) \, . \]
This concludes the proof of Theorem \ref{thm:main2}.

To show Theorem \ref{thm:main}, we use instead the fact that 
\[ i\partial_t \wt{\ph}_t |_{t=0} = -\Delta \ph + (N^3 V(N.) f_\ell (N.) * |\ph|^2) \ph \]
Since here we assume that the initial data $\ph = \phi_\text{GP}$ is the minimizer of the Gross-Pitaevskii energy functional (\ref{eq:GPen-functr}), it must satisfy the Euler-Lagrange equation
\[ -\Delta \ph + V_\text{ext} \ph + 8\pi a_0 |\ph|^2 \ph = \mu \ph \]
for some $\mu \in \bR$. We find
\[ i\partial_t \wt{\ph}_t |_{t=0} = \mu \ph - V_\text{ext} \ph + \big[  (N^3 V(N.) f_\ell (N.) * |\ph|^2) - 8 \pi a_0 |\ph|^2 \big] \ph \]
Using (\ref{eq:Vfa0}) the fact that the minimizer $\ph$ of (\ref{eq:GPen-functr}) is continuously differentiable and vanishes at infinity (see \cite[Theorem 2.1]{LSY}), we obtain
\[ \Big\| \big[  (N^3 V(N.) f_\ell (N.) * |\ph|^2) - 8 \pi a_0 |\ph|^2 \big] \ph \Big\|_2 \leq C N^{-1} \]
and therefore
\[ -2 N \text{Re }  \langle \ph,\gamma_N^{(1)} q_0 (i\partial_t \wtph_t)|_{t=0} \rangle \leq 2N \text{Re } \langle \ph, \gamma_N^{(1)} q_0 (V_\text{ext} + \kappa) \ph \rangle + C \]
for any constant $\kappa \in \bR$. Choosing $\kappa \geq 0$ so that $V_\text{ext} + \kappa \geq 0$ (from the assumptions, $V_\text{ext}$ is bounded below), we find 
\[ \begin{split} -2N  \text{Re }  \langle \ph,\gamma_N^{(1)} q_0 &(i\partial_t \wtph_t)|_{t=0} \rangle \\ &\leq 2N \text{Re } \langle \ph, \gamma_N^{(1)} (V_\text{ext} + \kappa) \ph \rangle -2N \langle \ph, \gamma_N^{(1)} \ph \rangle \langle \ph, (V_\text{ext} + \kappa) \ph \rangle + C \\ &\leq 2N \text{Re } \langle \ph, \gamma_N^{(1)} (V_\text{ext} + \kappa) \ph \rangle -2N \langle \ph, (V_\text{ext} + \kappa) \ph \rangle +C (N a_N +1) \end{split} \]
With Cauchy-Schwarz and since $0 \leq \gamma^{(1)}_N \leq 1$ implies that $(\gamma^{(1)}_N)^2 \leq \gamma_N^{(1)}$, we get 
\[ \begin{split} -2N  \text{Re }  \langle \ph, &\gamma_N^{(1)} q_0 (i\partial_t \wtph_t)|_{t=0} \rangle\\ &\leq N \langle \ph, \gamma_N^{(1)} (V_\text{ext} + \kappa) \gamma_N^{(1)} \ph \rangle - N \langle \ph, (V_\text{ext} + \kappa) \ph \rangle  + C (N a_N + 1) \\ &\leq N \tr \, \gamma_N^{(1)} V_\text{ext} - N \langle \ph, V_\text{ext} \ph \rangle + C (N a_N + 1) \end{split} \]
Inserting back in (\ref{eq:G0-C0f}) we conclude that, under the assumptions of Theorem \ref{thm:main}, 
\[\langle \xi, (\cG_{N,0} - C_{N,0}) \xi \rangle \leq \; \big[ \langle \psi_N, H^\text{trap}_N \psi_N \rangle - N \cE^\text{trap}_\text{GP} (\ph) \big] + C N a_N + C \leq C \big[ N a_N + N b_N + 1\big] \]
With (\ref{eq:1-pgp}) and (\ref{eq:wtph-ph}), we find now 
\[ 1- \langle \ph_t , \gamma^{(1)}_{N,t} \ph_t \rangle \leq C \big[  a_N + b_N + N^{-1} \big] \exp (c \exp (c |t|)) \]
This concludes the proof of Theorem \ref{thm:main}. 
\end{proof}

\medskip

{\it Acknowledgement.} B.S. gratefully acknowledge support from the NCCR SwissMAP and from the Swiss National Foundation of Science through the SNF Grants ``Effective equations from quantum dynamics'' and ``Dynamical and energetic properties of Bose-Einstein condensates''.

\end{document}